%% file: main.tex
\documentclass{article}

\usepackage[preprint]{neurips_2026}

\usepackage[utf8]{inputenc} 
\usepackage[T1]{fontenc}    
\usepackage{hyperref}       
\usepackage{url}            
\usepackage{booktabs}       
\usepackage{amsfonts}       
\usepackage{nicefrac}       
\usepackage{microtype}      
\usepackage{xcolor}         

\definecolor{darkblue}{RGB}{47,103,177}
\definecolor{darkred}{RGB}{191,44,35}

\usepackage{graphicx}
\usepackage{multirow}
\usepackage{wrapfig}
\usepackage[most]{tcolorbox}

\usepackage{amsmath}
\usepackage{amssymb}
\usepackage{mathtools}
\usepackage{amsthm}
\usepackage{thmtools, thm-restate}
\usepackage[bottom]{footmisc}

\theoremstyle{plain}
\newtheorem{theorem}{Theorem}[section]
\newtheorem{proposition}{Proposition}
\newtheorem{lemma}{Lemma}
\newtheorem{corollary}{Corollary}
\theoremstyle{definition}
\newtheorem{definition}{Definition}
\newtheorem{assumption}{Assumption}
\theoremstyle{remark}
\newtheorem{remark}{Remark}

\usepackage{algorithm}
\usepackage{algpseudocode}
\usepackage{comment}

\newcommand{\x}{\mathbf{x}}
\newcommand{\y}{\mathbf{y}}
\renewcommand{\u}{\mathbf{u}}
\renewcommand{\v}{\mathbf{v}}
\usepackage{dsfont}
\usepackage{enumitem}
\newcommand{\indicator}[2]{\mathds{1}_{#1}\big({#2}\big)}
\newcommand{\apprx}{\mathrm{apprx}}
\newcommand{\subN}{\underline{N}}
\newcommand{\barN}{\bar{N}}

\newcommand{\Residue}{\mathcal{O}\Big(\frac{1}{\sqrt{N}}\Big)}
\newcommand{\ResidueNew}{\mathcal{O}\Big(\frac{1}{\sqrt{\bar\subN}}\Big)}

\newcommand{\RMSet}[2]{\mathcal{R}_{\mu,#1}^{#2}(\mu_{#1}, \nu_{#1})}
\newcommand{\RNSet}[2]{\mathcal{R}_{\nu,#1}^{#2}(\mu_{#1}, \nu_{#1})}

\newcommand{\empMu}[1]{\mathrm{Emp}_{\mu}(#1)}
\newcommand{\empNu}[1]{\mathrm{Emp}_{\nu}(#1)}

\newcommand{\bfX}{\mathbf{X}}
\newcommand{\bfY}{\mathbf{Y}}

\newcommand{\bfx}{\mathbf{x}}
\newcommand{\bfy}{\mathbf{y}}
\newcommand{\bfu}{\mathbf{u}}
\newcommand{\bfv}{\mathbf{v}}

\renewcommand{\t}{^{\mbox{\tiny\sf T}}}

\newcommand{\lowervalue}{\underline{J}}

\newcommand{\X}{\mathcal{X}}
\newcommand{\Y}{\mathcal{Y}}
\newcommand{\M}{\mathcal{M}}
\newcommand{\N}{\mathcal{N}}
\newcommand{\U}{\mathcal{U}}
\newcommand{\V}{\mathcal{V}}
\newcommand{\F}{\mathcal{F}}
\newcommand{\G}{\mathcal{G}}
\newcommand{\R}{\mathcal{R}}
\renewcommand{\P}{\mathcal{P}}
\renewcommand{\S}{\mathcal{S}}

\newcommand{\Z}{\mathcal{Z}}
\newcommand{\ED}{ED}

\newcommand{\st}{~\mathrm{s.t.}~}

\usepackage{pifont}
\newcommand{\cmark}{\ding{51}}%
\newcommand{\xmark}{\ding{55}}%

\newcommand{\abs}[1]{\left\lvert#1\right \rvert}

\newcommand{\expct}[2]{\mathbb{E}_{#2}\left[#1\right]}

\newcommand{\dtv}[1]{\mathrm{d}_{\mathrm{TV}} \big( #1 \big)}
    
\newcommand{\edit}[1]{\textcolor{black}{#1}}

\title{Learning Large-Scale Competitive Team Behaviors with Mean-Field Interactions}

\author{%
  Bhavini Jeloka \\
  Department of Aerospace Engineering\\
  Georgia Institute of Technology\\
    Atlanta, GA 30318, USA \\
  \texttt{bjeloka3@gatech.edu} \\
   \And
  \begingroup
    \hypersetup{pdfborder={0 0 0}}
    Yue Guan\thanks{Work done while at the Georgia Institute of Technology.}
    \endgroup \\
  \begingroup
    \scalebox{0.92}{Oden Institute for Computational Engineering \& Sciences}
\endgroup \\
  The University of Texas at Austin\\
  Austin, TX 78712, USA \\
  \texttt{yue.guan@austin.utexas.edu} \\
  \AND
   Panagiotis Tsiotras \\
  Department of Aerospace Engineering\\
  Georgia Institute of Technology\\
    Atlanta, GA 30318, USA \\
  \texttt{tsiotras@gatech.edu} \\
}

\begin{document}

\maketitle


\begin{abstract}
While multi-agent reinforcement learning (MARL) has shown strong empirical performance, existing methods struggle with a large number of agents, due to the combinatorial growth of joint interactions.
Mean-field (MF) approximations address this challenge by replacing pairwise interactions with surrogate interactions against population distributions, yielding tractable formulations to solve for the large-population policies.
However, existing MF formulations focus on fully cooperative or purely competitive settings and fail to capture the mixed cooperative–competitive structure of team-based games.
We extend mean-field learning to \emph{zero-sum team games}, where agents cooperate within each team and compete at the team level.
We show that such games admit $\epsilon$-optimal decentralized policies that depend only on the local states of each agent and the population team distributions.
Building on this structure, we propose MF-MAPPO, a scalable algorithm with a shared actor and a minimally informed critic per team. 
MF-MAPPO is trained directly in finite-population simulators rather than using mean-field oracles, thereby enabling deployment to realistic scenarios with thousands of agents.
We further extend MF-MAPPO to partially observable settings via a simple gradient-regularized training scheme. 
Experiments on large-scale benchmarks in our simulation platform (\texttt{MFEnv}), including population games with analytical solutions and high-dimensional battlefield scenarios, demonstrate that MF-MAPPO outperforms existing MARL baselines and yields rich, heterogeneous behaviors.
\end{abstract}

\section{Introduction}\label{sec:intro}
\vspace{-0.05in}

Large-scale multi-agent reinforcement learning (MARL) suffers from a scalability bottleneck: 
the joint state–action space grows exponentially with the number of agents, rendering standard approaches such as MADDPG~\citep{lowe2017multi} and MAPPO~\citep{yu2022surprising} intractable in large populations. 
A principled remedy is provided by mean-field theory~\citep{huang2006large}, which enables tractable analysis and learning by approximating interactions among agents through an interaction between a representative agent and the population distribution (i.e., the mean-field).

Mean-field methods have led to a rich body of work in both non-cooperative settings, through mean-field games (MFGs)~\citep{huang2006large, lasry2007mean}, and fully cooperative settings, through mean-field control (MFC)~\citep{bensoussan2013mean, gu2021meanfieldcontrolsqlearningcooperative}. 
Recent learning-based methods include, MF-Q~\citep{yang2018mean}, DDPG-based extensions for mean-field team games~\citep{shao2024reinforcementlearningfinitespace}, and GAN-based approaches for continuous domains~\citep{wang2022multi}. 
However, existing methods often rely on access to a \textbf{mean-field oracle} for exact population distribution trajectories, which are typically available only for simple scenarios. 
Additionally, these methods are primarily designed for either fully competitive or fully cooperative regimes. 
Consequently, scenarios involving \textbf{mixed cooperative–competitive interactions}, where agents collaborate within teams while competing across teams, remain underexplored despite their prevalence in real-world systems such as team sports~\citep{gaviria2025socio} and social dilemmas~\citep{leibo2017multi}.

Recent work has begun to address this gap through team-based mean-field formulations. 
Early single-team results appear in~\cite{arabneydi2014team}, while multi-population extensions have been studied in~\cite{perolat2021scalingmeanfieldgames}, typically under restrictive independence assumptions. 
More recently, \citet{guan2024zero} introduced \emph{zero-sum mean-field team games} (ZS-MFTGs), later generalized to general-sum team settings by~\citet{subramanian2023mean}. 
While these formulations provide a principled foundation, existing numerical solution approaches rely primarily on linear–quadratic structure or discretization. 
Although such methods scale in the number of agents, they struggle \emph{computationally} to handle complex environments with high-dimensional state spaces and nonlinear dynamics, limiting their applicability in modern settings.

Beyond computational challenges, introducing competition fundamentally alters the structure of mean-field interactions. 
In standard mean-field settings, optimality is characterized through single-agent deviations~\citep{ cui2021approximately}, whereas competitive team settings require reasoning about \emph{team-level deviations}.
As a result, optimal policies must depend on feedback from the opponent’s distribution~\citep{guan2024zero}. 
This introduces another key challenge: agents require access to the opponent population distribution during execution, which is difficult under \textbf{partial observability of population distributions}. 
Most existing approaches assume perfect mean-field measurement or communication~\citep{guan2024zero,subramanian2023mean}. 
While partial-observation and estimation methods~\citep{inoue2021model, perrin2021meanfieldgamesflock, benjamin2025networkedcommunicationmeanfieldgames} have been explored, they often struggle in discrete spaces or in competitive environments. 
Moreover, addressing partial information requires not only accurate estimation of mean-field distributions, but also systematic adaptation of the learning framework to ensure robustness to estimation errors.

In this work, we address these challenges by developing a scalable learning framework for mixed cooperative–competitive mean-field systems.

\begin{tcolorbox}[colback=gray!5, colframe=black, boxrule=0.1pt]
\textbf{Contributions.}
We introduce \textbf{Mean-Field Multi-Agent Proximal Policy Optimization} (MF-MAPPO), which achieves scalability for mixed cooperative–competitive settings through three key design principles with theoretical guarantees:
(i) a \emph{minimally-informed critic} that depends only on team distributions, ensuring that the input dimension remains independent of the number of agents;
(ii) a \emph{shared actor per team}, which reduces the joint strategy space to identical team policies and enables efficient learning (Theorem~\ref{thm:finite-population-training}); and
(iii) the ability to train with $N$-agent simulators and generalize to larger populations (Theorems~\ref{thm:convergence-of-gradients} and~\ref{thm:extend-to-many}).

\vspace{+0.05in}
To address partial observability, we show that MF-MAPPO can be combined with gradient regularization and deployed using estimated mean-fields without requiring retraining, and we provide performance bounds depending on the estimation accuracy (Proposition~\ref{thm:estimator-eval-dyn}).

\vspace{+0.05in}
We validate our approach across several environments, demonstrating that MF-MAPPO produces stable learning dynamics and rich behaviors without reward engineering.

\vspace{+0.05in}
Overall, our contributions advance \textbf{scalable} reinforcement learning for realistic large-population systems by jointly addressing \textbf{mixed cooperative–competitive interactions}, \textbf{partial observability}, and \textbf{finite-population training}.
\end{tcolorbox}

\begin{table}[t!]
\centering
\caption{Comparison between our approach and related works.}
\resizebox{\linewidth}{!}{\begin{tabular}{cccccc}
\toprule
 Method  & \shortstack{Multiple\\Populations} & \shortstack{MF-Dependent\\ Policies} & \shortstack{Mixed Coop-Comp \\ Setting} & \shortstack{Finite-Population \\ Simulator} & \shortstack{Partial MF \\ Observation} \\ 
 \midrule
MFG-MOMD~\citep{ perolat2021scalingmeanfieldgames} & \cmark & \xmark & \xmark & \xmark & \xmark \\ 
Networked-MOMD~\citep{benjamin2025networkedcommunicationmeanfieldgames} & \xmark& \cmark & \xmark & \cmark & \cmark \\ 
DDPG-MFTG~\citep{shao2024reinforcementlearningfinitespace} & \cmark & \cmark & \cmark & \xmark & \xmark \\
MFTG-Discretization~\citep{guan2024zero} & \cmark & \cmark & \cmark & \xmark & \xmark \\ 
 \midrule
 \textbf{MF-MAPPO (Ours)}& \cmark  & \cmark & \cmark & \cmark & \cmark  \\
\bottomrule
\end{tabular}}
\label{table:lit-review}
\end{table}

\textbf{Related Work.} Table~\ref{table:lit-review} highlights works most relevant to our setting and clarifies our contributions.
%
%
\citep{perolat2021scalingmeanfieldgames} developed an online mirror descent approach for multi-population MFG, but focuses on open-loop policies and does not capture team-level competition.
\citep{benjamin2025networkedcommunicationmeanfieldgames} extended this framework to networked settings with communication-based MF estimation, but is restricted to a single population.
\citep{shao2024reinforcementlearningfinitespace} introduced DDPG-MFTG for mixed cooperative--competitive scenarios, but it relies on a mean-field oracle providing exact infinite-population trajectories and assumes fully observable mean-field information, limiting its applicability to complex environments.
\citet{guan2024zero} studied zero-sum mean-field team games via discretization, which scales in the number of agents but becomes intractable in large state spaces.
\vspace{-0.05in}
\section{Problem Formulation and Mean-Field Approximation}
\label{sec: problem-formulation}
\vspace{-0.05in}

A zero-sum mean-field team game (ZS-MFTG) models a discrete-time stochastic game between two large teams of agents (Red and Blue)~\citep {guan2024zero}.  
The Blue and Red teams consist of $N_1$ and $N_2$ \emph{identical} agents for each team.
Let \( X^{N_1}_{i,t} \in \mathcal{X} \) and \( U^{N_1}_{i,t} \in \mathcal{U} \) represent the state and action of Blue agent \( i \in [N_1] \) at time \( t \).
Here, \(\mathcal{X}\) and \(\mathcal{U}\) are the finite state and action spaces of the Blue team. 
Similarly, \( Y^{N_2}_{j,t} \in \mathcal{Y} \) and \( V^{N_2}_{j,t} \in \mathcal{V} \) denote the state and action of Red agent \( j \in [N_2] \) at time \( t \). 
The joint state-action variables for the Blue and Red teams are denoted as \((\mathbf{X}^{N_1}_t, \mathbf{U}^{N_1}_t)\) and \((\mathbf{Y}^{N_2}_t, \mathbf{V}^{N_2}_t)\), respectively.
{We denote the space of probability measures over a set $E$} as $\P({E})$.
Below, $\dtv{\mu, \mu'}$ represents the total variation between $\mu, \mu' \in \P(E)$.

\begin{definition} 
    \label{def:empirical-dist}
    The \textit{empirical distributions} (ED) for the two teams are respectively defined as $\M^{N_1}_t\!(x) \!=\! \frac{1}{N_1} \sum_{i=1}^{N_1} \!\mathds{1}_x\!(X^{N_1}_{i,t})$ and $\N^{N_2}_t(y) \!=\! \frac{1}{N_2} \sum_{j=1}^{N_2}\!\mathds{1}_y\!(Y^{N_2}_{j,t})$,  where $\indicator{\cdot}{\cdot}$ is the indicator.
\end{definition}
Specifically, $\M^{N_1}_t(x)$ gives the fraction of Blue agents at state~$x$. 
We use $\M_t^{N_1} =\empMu{\bfX^{N_1}_t}$ and 
$\N_t^{N_2} =\empNu{\bfY^{N_2}_t}$
to denote the EDs computed from the given joint states.
%

\paragraph{Dynamics and Rewards.}
We consider weakly-coupled dynamics where the dynamics of each individual agent are coupled with other agents through the EDs~\citep{huang2006large}. 
For the Blue agent $i$, its stochastic transition is governed by the kernel 
$f_t : \X \times \U \times \P(\X) \times \P(\Y) \rightarrow \P(\X)$:
\begin{equation}
\label{ft:DEF}
    \mathbb{P}(X_{i,t+1}^{N_1}=x_{i,t+1}^{N_1} | U_{i,t}^{N_1} = u_{i,t}^{N_1}, \bfX_{t}^{N_1} = \bfx_{t}^{N_1}, \bfY_{t}^{N_2} \!=\! \bfy_{t}^{N_2}) 
    = f_t(x_{i,t+1}^{N_1}\vert x_{i,t}^{N_1}, u_{i,t}^{N_1}, \mu^{N_1}_t, \nu^{N_2}_t),   
\end{equation}
where $\mu^{N_1}_t \!\!=\!\! \empMu{\bfx_t^{N_1}}$ and $\nu^{N_2}_t \!\!=\!\! \empNu{\bfy_t^{N_2}}$. Similarly, the dynamics of a Red agent is governed by the transition kernel $g_t(y_{j,t+1}^{N_2}\vert y_{j,t}^{N_2}, v_{j,t}^{N_2}, \mu^{N_1}_t, \nu^{N_2}_t)$.
All agents in the Blue team receive an identical \emph{weakly-coupled} team reward, i.e., $r_t \triangleq r_t(\mu_t, \nu_t): \P(\X) \times \P(\Y) \to \mathbb{R}$. 
The Red agents receive $-r_t(\mu_t, \nu_t)$ as their rewards {(zero-sum)}. 
We assume that the Blue team is \textit{maximizing} while the Red team is \textit{minimizing} and $r_t\in[-R_{\max}, R_{\max}]$ for all $t$.

\begin{assumption}[Lipschitz Model]
    \label{assmpt:lipschitiz-model}
    For all $x\in \X, u\in\U$, $\mu, \mu' \in \P(\X)$, $\nu, \nu' \in \P(\Y)$ and all $t$,
    there exist constants $L_{f}, L_r > 0$ such that {
    \(\textstyle \sum_{x' \in \X} \abs{f_t(x'|x,u,\mu, \nu) - f_t(x'|x,u,\mu', \nu')}
    \leq L_{f}\big(\dtv{\mu, \mu'} + \dtv{\nu, \nu'}\big)\) and \(\textstyle \abs{r_t(\mu, \nu) - r_t(\mu', \nu')} \leq L_r \big(\dtv{\mu,\mu'} + \dtv{\nu,\nu'}\big)\).}
\end{assumption}
We similarly assume that $g_t$ is $L_g$-Lipschitz. 
The Lipschitz model is a standard assumption in the mean-field literature~\citep{huang2006large, gu2021meanfieldcontrolsqlearningcooperative}, without which the mean-field approximation may fail as shown in the counterexamples~\citep{cui2023learning, guan2024zero}.

\paragraph{Information Structure.}
We consider a mean-field sharing information structure~\citep{arabneydi2015team}, where each agent's decision depends on its own state and the two teams' EDs. 
{We start by assuming full observation of EDs and later we relax this assumption in Section~\ref{sec:d-pc}.}
Specifically, the Blue and Red teams seek to construct \emph{mixed} Markov policies 
\[
\phi_{i,t}(u|x_{i,t}^{N_1}, \mu^{N_1}_t, \nu^{N_2}_t) ~~\text{and}~ ~ \psi_{j,t}(v|y_{j,t}^{N_2}, \mu^{N_1}_t, \nu^{N_2}_t).
\]
%
We use $\Phi_t$ to denote the set of individual Blue policies at time~$t$.
We define the Blue team policy $\phi^{N_1}_t \!=\! \{\phi_{i, t}\}_{i=1}^{N_1}$ as the collection of the $N_1$ Blue agent individual policies, and denote the set of Blue team policies as $\Phi^{N_1}_t \!=\! \times_{N_1} \Phi_t$.
Similarly, the Red team policy is denoted as $\psi_t^{N_2} \in \Psi_t^{N_2}\!=\!\times_{N_2} \Psi_t$.

\paragraph{Optimization Objective.}

The performance of a team policy pair $\textstyle(\phi^{N_1}\!, \psi^{N_2})$ is measured by 
\(
    J^{N, \phi^{N_1}, \psi^{N_2}} \big(\bfx_0^{N_1}, \bfy_0^{N_2} \big)=
    \mathbb{E}_{{\phi^{N_1}, \psi^{N_2}}}
    \Big[\sum_{t=0}^{T} r_t(\M^{N_1}_t, \N^{N_2}_t) \big \vert \bfx_0^{N_1}, \bfy^{N_2}_0\! \Big].
\)
When the maximizing Blue team considers its worst-case performance, we have the following max-min optimization:
\begin{equation}\label{eq:blue-max-min}
    \hspace{-0.01in}
    \underline{J}^{N*}(\bfx_0^{N_1} , \bfy_0^{N_2})  = \max_{\mathbf{\phi}^{N_1} \in \Phi^{N_1}} \min_{\mathbf{\psi}^{N_2} \in \Psi^{N_2}} J^{N, \phi^{N_1}, \psi^{N_2}}\hspace{-0.02in}(\bfx_0^{N_1}\hspace{-0.05in}, \bfy_0^{N_2}), \hspace{-0.1in}
\end{equation}
where $\underline{J}^{N*}$ is the lower game value for the $N$-agent finite-population game. 
Similarly, the minimizing Red team considers a min-max optimization and the upper game value.
Note that the optimization in~\eqref{eq:blue-max-min} allows agents within each team to follow different strategies. 
Such \emph{non-identical team policies} preserve maximum flexibility for the Blue team and, more crucially, avoid underestimating the adversary Red team's strategy flexibility. 
However, this generality leads to exponential growth in the joint policy space, rendering the problem intractable~\citep{bernstein2002complexity}.

\paragraph{Battlefield Scenario.}

\begin{wrapfigure}{r}{0.23\linewidth}
    \vspace{-0.15in}
    \centering
    \includegraphics[width=\linewidth]{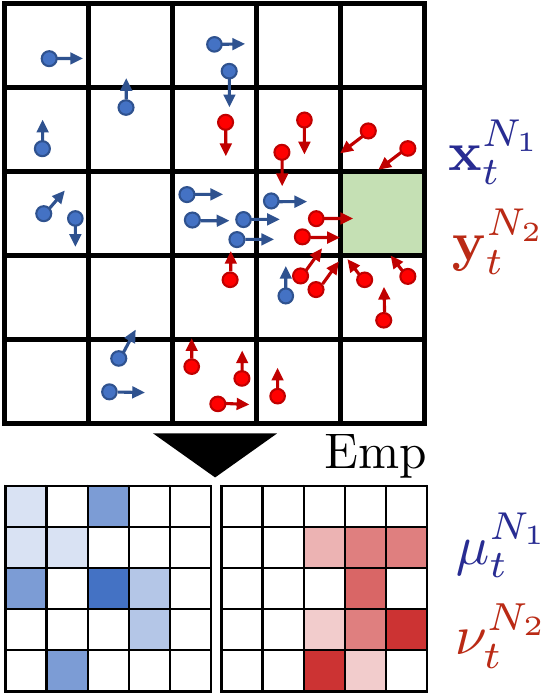}
    \vspace{-0.2in}
    \caption{Battlefield Scenario.}
    \label{fig:battle}
    \vspace{-10pt}
\end{wrapfigure}

To motivate and ground the ZS-MFTG formulation, we consider a grid-based battlefield scenario in which Blue agents aim to reach designated targets (green cells), while Red agents attempt to block their access. 
At each cell, an activation–deactivation mechanism governs interactions: 
an agent’s probability of being deactivated is proportional to the opponent’s numerical advantage relative to its own team’s presence in that cell. 
This induces strategic grouping and dispersion behaviors. 
The Blue team seeks to maximize the fraction of its population that reaches the target cells, while the Red team aims to minimize this quantity, resulting in a zero-sum objective.
Applying the $\mathrm{Emp}$ operator yields the empirical distributions of both teams.

\textbf{Mean-Field Approximation.}
To reduce the complexity of team policy optimization domains in~\eqref{eq:blue-max-min}, previous works~\citep{guan2024zero, subramanian2023mean} propose to approximate complex team behaviors using \textit{identical team policies} in the \textit{infinite-population} limit.
\begin{definition} [Identical Team Policy]
    \label{def:identical-policy}
    The Blue team policy $\phi^{N_1}_t = (\phi_{1,t}, \ldots, \phi_{N_1,t})$ is \textit{identical}, if $\phi_{i_1,t} = \phi_{i_2,t}$ for all times $t$ and all $i_1,i_2 \in [N_1]$. 
    The set of identical Blue team policies is denoted as $\Phi$.
    Similarly, $\Psi$ represents the set of identical Red team policies.
\end{definition}

With an identical team policy, an infinite-population team's joint states can be represented using the team population distribution, which 
coincides with the state distribution of a typical agent and is commonly known as the \emph{mean-field} (MF)~\citep{huang2006large}.
We use $\mu_t$ and $\nu_t$ to denote the Blue and Red MFs, respectively. 
\citet{guan2024zero} showed that MFs induced by identical team policies closely approximate finite-population EDs under non-identical policies. 
This justifies the restriction of the optimization domain in~\eqref{eq:blue-max-min} to identical policies, yielding the infinite-population optimization 
\begin{equation}   
    \label{eqn:lower-value-inf}
    \underline{J}_\infty^{*} (\mu_0, \nu_0) = \max_{\phi\in \Phi} \; \min_{\psi \in \Psi} ~ J_\infty^{\phi, \psi}(\mu_0, \nu_0),
\end{equation}
where $\mu_t$ and $\nu_t$ follow \emph{deterministic} dynamics, corresponding to the evolution of the state distributions of the Blue and Red representative agents.

While the lower (max--min) and upper (min--max) values may not coincide in general, i.e., a Nash value may not exist in the infinite-population game, we can show that a Nash equilibrium exists under a convex--concave reward structure and independent dynamics (where transitions depend only on an agent’s state and action). 
See Appendix~\ref{appdx-sec:game-value} for details.

\vspace{-0.05in}
\section{Mean-Field Multi-Agent Proximal Policy Optimization}\label{sec:mf-mappo}
\vspace{-0.05in}

Note that the optimization~\eqref{eqn:lower-value-inf} is defined in the infinite-population limit and depends on the MF trajectories $(\mu_t, \nu_t)$. 
To obtain such trajectories, existing learning approaches~\citep{shao2024reinforcement, perolat2021scalingmeanfieldgames} typically rely on mean-field oracles that propagate the state distributions of representative agents. 
However, these oracles are only available when the scenarios are simple enough to permit exact distributional propagation. 
In realistic environments with complex dynamics and high-dimensional state spaces, one needs to leverage finite-population simulators to generate ED trajectories using samples of joint states and actions as approximations of the underlying MF trajectories. 
Motivated by this need for generality, we design the MF-MAPPO algorithm to operate directly on finite-population simulators, 
which introduces two key challenges:
(i) The learned policies are optimal only for finite-population games, whereas existing theoretical guarantees rely on optimality in the infinite-population limit; thus, new performance guarantees are required.
(ii) Policy gradients estimated from finite-population simulators are stochastic and often deviate from their mean-field counterparts, requiring additional convergence analysis.


\citet{guan2024zero} show that the optimal policy $\phi^*$ obtained from the \emph{infinite-population} optimization in~\eqref{eqn:lower-value-inf} is $\mathcal{O}(1/\sqrt{\underline{N}})$-optimal in the finite-population game, even against non-identical exploiting Red team policies,
where $\underline{N} = \min \{N_1, N_2\}$.
The following theorem presents a similar guarantee when the identical policy is obtained \edit{as an \emph{exact} solution to the finite-population optimization problem}.

\begin{restatable}{theorem}{maintheoremtwo}\label{thm:finite-population-training}
Let $\tilde \phi^*$ be a solution to the finite-population optimization restricted to identical policies, i.e., 
$\max_{\phi \in \Phi} \min_{\psi \in \Psi} J^{N, \phi, \psi} (\bfx_0^{N_1}, \bfy_0^{N_2})$.
Then, the exploitability of $\tilde \phi^*$ satisfies 
\[
\min_{\psi^{N_2} \in \Psi^{N_2}} 
       J^{N,\tilde\phi^*,\psi^{N_2}} (\bfx^{N_1}, \bfy^{N_2}) 
       \geq \underline{J}^{N*}(\bfx^{N_1}, \bfy^{N_2})
       - \mathcal{O}(1/\sqrt{\underline{N}}),
\quad \forall \bfx^{N_1} \in \X^{N_1}, \bfy^{N_2} \in \Y^{N_2}.
\]
\end{restatable}

Theorem~\ref{thm:finite-population-training} provides a principled justification for learning \emph{identical} team policies in finite-population games, even when the opponent may exploit using non-identical strategies (i.e., $\psi^{N_2} \in \Psi^{N_2}$). 
Moreover, the exploitability vanishes as $N_1, N_2 \to \infty$, recovering the classical infinite-population mean-field result~\citep{huang2006large}.
%

\begin{figure}[t!]
    \centering
            \includegraphics[width=0.8\linewidth]{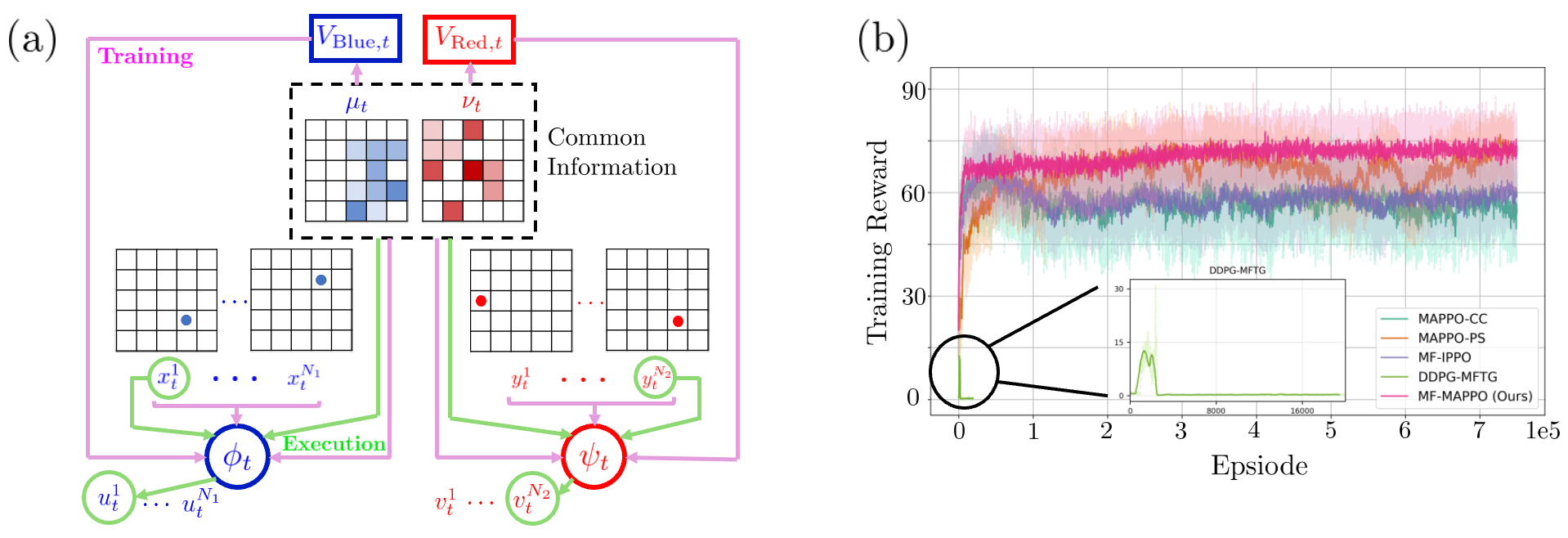}
    \caption{%
    (a) Architecture of MF-MAPPO.
    (b) Training curve for Battlefield on a 4x4 grid (Blue).
    }
    \label{fig:training-architecture}
\end{figure}

With identical policies justified, we introduce our main contribution: the Mean-Field Multi-Agent Proximal Policy Optimization (MF-MAPPO) algorithm. 
MF-MAPPO builds on the proximal policy optimization (PPO) framework with two key adaptations:
\emph{minimally-informed critics}
and 
\emph{shared team actors}. 
These components exploit the underlying game structures to enable efficient learning while maintaining scalability. 
We describe these design details below.

\textbf{Minimally-Informed Critic.} 
The shared team reward and identical policy approximation imply that the value function depends only on the team distributions (See Proposition~\ref{thm:minimal-critic} in Appendix~\ref{appdx-sec:MI-critic}). 
Accordingly, we employ a \emph{minimally-informed critic} per team, which takes $(\mu, \nu)$ as input, independent of joint agent states and actions.
The Blue team critic evaluates the value function $V_\text{Blue}(\mu, \nu)$. 
We use $\zeta_\text{Blue}$ to parameterize the critic network and train it via minimizing the MSE loss as follows
\begin{equation}
\label{eq:critic network}
    L_{\text{critic}}(\zeta_\text{Blue}) = \frac{1}{|B|}\sum_{k=1}^{|B|}\Bigl(V_{\zeta_\text{Blue}}(\mu_k, \nu_k) - \hat{R}_{\text{Blue}, k}\Bigr)^2,
\end{equation}
where $B$ is the mini-batch size and $\hat{R}_{\text{Blue}, k}$ is the discounted reward-to-go for sample $k$. 
The Red team critic is parameterized with $\zeta_\text{Red}$ and updated similarly.
The following proposition justifies the usage of minimally-informed critics.

This result reduces the learning problem to a single critic per team. 
Unlike standard centralized training/decentralized execution (CTDE), 
which requires joint observations of all agents $(\x^{N_1}_t, \y^{N_2}_t)$, our critic depends only on the empirical distributions. 
This design is critical for scalability, as the input dimension is independent of the number of agents.

\textbf{Shared Team Actor.}
The exploitability bound of identical team strategies in Theorem~\ref{thm:finite-population-training} allows us to employ a shared actor per team. 
The actor optimizes a PPO-based objective with a decaying entropy bonus~\citep{schulman2017proximal, shengyi2022the37implementation}, which promotes exploration and stabilizes learning in mean-field settings~\citep{cui2021approximately, guan2022shaping}.
Identical team policies further allow a single replay buffer per team, reducing memory costs and simplifying experience collection.
The PPO-based objective function of the Blue team actor is given by: 
{\small{\begin{equation}     \label{eq:actor network}
    L_{\text{actor}}(\theta_{\text{Blue}}) =  \frac{1}{|B|} \sum_{k=1}^{|B|} \Big[\text{min}\Bigl(g_k(\theta_{\text{Blue}})A_k, \operatorname{clip}_{[1-\epsilon, 1+\epsilon]}(g_k(\theta_{\text{Blue}}) )A_k\Bigr) + \omega S(\phi_{\theta_{\text{Blue}}}(x_k, \mu_k, \nu_k))\Big], 
 \end{equation}}}
where, $\textstyle g(\theta) = \frac{\phi_{\theta}(u|x, \mu, \nu)}{\phi_{\theta^{\text{old}}}(u|x, \mu, \nu)}$, $A_k$ is the generalized advantage function estimate~\citep{schulman2015high} and the tunable parameter $\omega$  weighs the entropy term $S(\cdot)$ and decays as training progresses.

\subsection{Theoretical Guarantees for Finite-Population Training}

%
As discussed earlier, the infinite-population limit remains of significant theoretical interest. 
It is therefore desirable to ensure that policy gradients computed using finite-population simulators converge to their infinite-population counterparts as the population size grows, which motivates the following theorem.

\begin{restatable}[\textit{Informal}]{theorem}{maintheoremone}\label{thm:convergence-of-gradients}
The approximate policy gradient of the infinite-population Blue (Red) team policy computed using identical local policies $(\phi_{\theta_{\mathrm{Blue}}}, \psi_{\theta_{\mathrm{Red}}})$ via MF-MAPPO from the $(N_1, N_2)$-agent simulator converges to the infinite-population policy gradient, i.e., $ \|\nabla_{\theta_{\mathrm{Blue}}}  J_\infty(\phi_{\theta_{\mathrm{Blue}}},\psi) - \nabla_{\theta_{\mathrm{Blue}}}  \hat J^{N_1, N_2}(\phi_{\theta_{\mathrm{Blue}}}, \psi)\|_2 \to 0$ as $(N_1, N_2)\to\infty$.
\end{restatable}


Theorem~\ref{thm:convergence-of-gradients} relates policies learned using finite-population simulators to their infinite-population counterparts. 
This connection allows us to use the infinite-population solution as a bridge to analyze the performance of the learned policy under different population sizes, leading to Theorem~\ref{thm:extend-to-many}.
In particular, Theorem~\ref{thm:extend-to-many} allows MF-MAPPO to be trained on smaller populations and deployed to larger teams without additional tuning, while maintaining performance guarantees.



\begin{restatable}{theorem}{maintheoremthree}\label{thm:extend-to-many}
Let $\mathcal{G}_N$ denote the finite-population game in which agents use identical team policies $(\phi^*_t, \psi^*_t)$ with an $(N_1, N_2)$-agent simulator. 
Consider another game $\mathcal{G}_{\tilde N}$ with the same setup but population sizes $(\tilde N_1, \tilde N_2)$ satisfying ${\tilde N_1}/{\tilde N_2} = {N_1}/{N_2}$ and $\tilde N_1 \geq N_1$. 
If $(\phi^*_t, \psi^*_t)$ is $\epsilon$-optimal in $\mathcal{G}_N$, then it remains $\epsilon$-optimal in $\mathcal{G}_{\tilde N}$.
\end{restatable}



\vspace{-0.05in}
\section{MF-MAPPO for Partially Observable MFTGs}\label{sec:d-pc}
\vspace{-0.05in}

Unlike standard mean-field formulations based on single-agent deviations, ZS-MFTGs consider team-level deviations~\eqref{eq:blue-max-min}, requiring strategies that depend on feedback from both teams’ distributions. 
Since access to the opponent distribution is often unrealistic, we extend MF-MAPPO to partially observable settings where opponent distributions must be estimated.

\begin{wrapfigure}{r}{0.4\linewidth}
    \vspace{-0.22in}
    \centering
    \includegraphics[width=\linewidth]{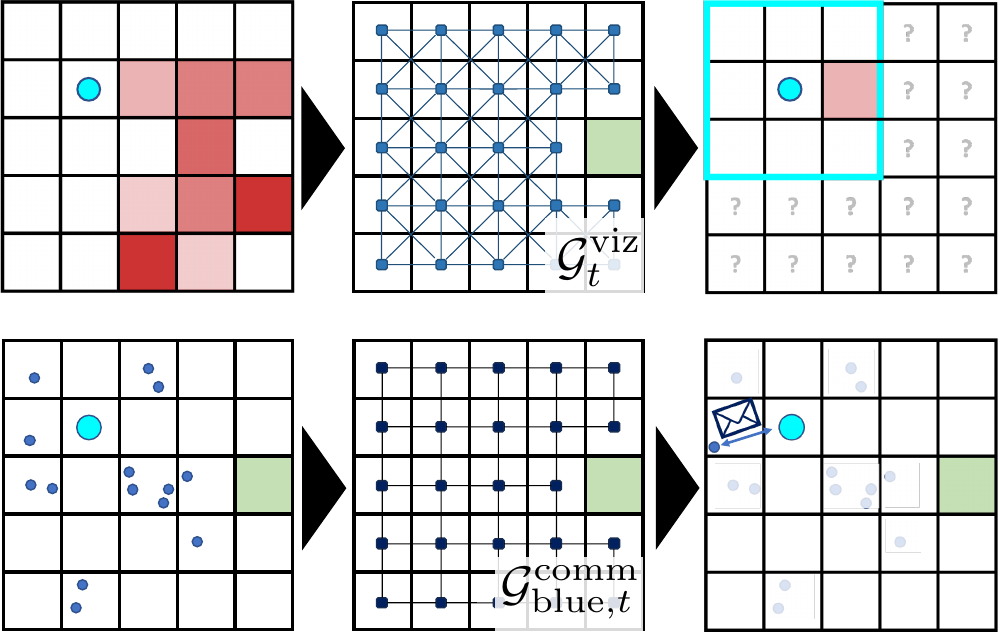}
    \vspace{-0.25in}
    \caption{PO-MFTG}
    \vspace{-0.3in}
\end{wrapfigure}

\textbf{Partially Observable Battlefield.} 
To motivate partial observability, consider a battlefield scenario with a state-dependent visibility graph $\mathcal{G}^{\mathrm{viz}}_t$ and communication graph $\mathcal{G}^{\mathrm{comm}}_t$, which define the states an agent can observe and the agents it can communicate with. 
In the example to the right, the cyan ego agent observes a $3\!\times\!3$ neighborhood under $\mathcal{G}^{\mathrm{viz}}_t$. 
Based on $\mathcal{G}^{\mathrm{comm}}_t$, it can communicate with agents in four neighboring cells, leading to information exchange with the teammate to the left.

Given the inevitable noise in the estimates, it is first necessary to ensure that the learned policies are robust to errors in the estimated mean-field distributions $(\hat{\mu}_t, \hat{\nu}_t)$. 
To this end, we introduce gradient regularization~\citep{chen2024learning} in the actor objective~\eqref{eq:actor network} to enforce Lipschitz continuity with respect to the mean-field inputs. 
This promotes robustness, so that small distribution estimation errors induce only minor changes in the action distributions. 
Specifically, 

{\small{\begin{equation}     \label{eq:gr-actor-network}
 L^{\text{GR}}_\text{actor}(\theta_{\text{Blue}}) = L_{{\text{actor}}}(\theta_{\text{Blue}})  + \frac{\lambda^{\text{GR}}}{|B|} \sum_{k=1}^{|B|}\|\nabla_\eta \log \phi(u_k | x_k, \eta)\big\vert_{\eta = (\mu_k, \nu_k)}\|^2_2,
 \end{equation}}}
where $L_{{\text{actor}}}(\theta_{\text{Blue}})$ follows from~\eqref{eq:actor network}, $\eta = (\mu, \nu)$ and $\lambda^{\text{GR}}$ is a tunable hyperparameter to control the regularization.
The following proposition relates the gradient regularization to the Lipschitz constant $L_\phi$ of the learned Blue team policy.

\begin{restatable}{proposition}{mainpropositionlipschitz}
    \label{prop:lipschitz-proposition}
    If the log-probability of the Blue team policy is uniformly bounded such that 
    for all $x\in\X, u\in\U, \mu\in\P(\X), \nu\in \P(\Y)$, $\textstyle \left\|\nabla_{\eta} \log \phi(u \mid x, \mu, \nu)\right\|_2 \leq {L_\phi}/{2|\U|}$, where the gradient is taken with respect to $\eta \in \P(\X)\times\P(\Y)$, then $\phi(u \mid x, \mu, \nu)$ is $L_\phi$-Lipschitz continuous, i.e., $\sum_{u}|\phi_t\left(u|x, \mu,\nu\right) - 
    \phi_t\left(u|x, \mu',\nu'\right)|\leq L_\phi\left(\dtv{\mu, \mu'} + \dtv{\nu, \nu'}\right)$ for all $x\in\X$.
\end{restatable} 
%

%
While we still train gradient-regularized MF-MAPPO (GR-MF-MAPPO) with full mean-field information similar to the CTDE framework, we rely on mean-field estimators being $\epsilon$-accurate to ensure performance of the resulting team policies under partial observability.

\begin{definition}
    Suppose Blue agent $i$ receives an observation $o_{i, t} \in \mathcal{O}_t$ at time $t$. 
    A Red mean-field estimator $\mathcal{H}_\nu:\mathcal{O}_t \to \mathcal{P}(\Y)$ is said to be $\epsilon$-accurate if $\dtv{\mathcal{H}_\nu(o_{i, t}), \nu_t^{N_2}} \!\leq\! \epsilon$ for all $i\in [N_1]$, where $\nu_t^{N_2}$ is the true Red distribution. 
    We similarly define $\mathcal{H}_\mu:\!\mathcal{O}_t \!\to \!\mathcal{P}(\X)$ for the Red team.
\end{definition}


%

We provide performance guarantees for GR-MF-MAPPO when deployed with $\epsilon$-accurate estimators. 
Let $(\phi^*, \psi^*)$ denote the learned policies. 
Under full mean-field feedback, the induced trajectories are $(\M^{N_1}, \N^{N_2})$, where Blue actions are sampled according to $\phi^*(\cdot \vert x_{i,t}, \M^{N_1}_t, \N^{N_2}_t)$. 
Under partial observability, Blue agents instead act using estimated opponent distributions, yielding trajectories $(\hat{\M}^{N_1}, \hat{\N}^{N_2})$ with actions sampled from $\phi^*(\cdot \vert x_{i,t}, \hat{\M}^{N_1}_t, \mathcal{H}_\nu(o_{i,t}))$, with Red actions sampled analogously. 
We quantify the performance degradation through the cumulative regret
\(
\Delta J(\phi^*\!,\! \psi^*) \!=\! \mathbb{E}_{\phi^*\!, \psi^*}\!\big[\big|\sum_{t=0}^T r_t(\M^{N_1}_t\!\!, \N^{N_2}_t) \!-\! \sum_{t=0}^T r_t(\hat{\M}^{N_1}_t\!\!, \hat{\N}^{N_2}_t)\big|\big].
\)
See Appendix~\ref{appdx-sec:estimator-thm} for details.

%
%
%

\begin{restatable}{proposition}{mainpropthree}\label{thm:estimator-eval-dyn}

 When a policy pair $(\phi^*_t, \psi^*_t)$ learned via GR-MF-MAPPO is deployed with an $\epsilon$-accurate estimator \edit{and satisfies a uniform Lipschitz bound}, the cumulative regret satisfies $\Delta J(\phi^*_t, \psi^*_t)\leq K\epsilon + \mathcal{O}(1/\sqrt{\underline{N}})$. 

\end{restatable}
Consequently, GR-MF-MAPPO can incorporate different estimators without retraining, with performance degrading linearly in $\epsilon$. 
%
%
We present two estimation algorithms for the Battlefield scenario, both of which can be shown to be $\epsilon$-accurate.
Details of the estimators are provided in Appendix~\ref{AppendixDPCDeets}.

\textit{Classical Consensus.} 
In consensus-based mean-field estimation~\citep{benjamin2025networked}, agents observe opponent counts based on their visibility graph $\mathcal{G}^{\mathrm{viz}}_t$ and exchange the counts through the communication graph $\mathcal{G}^{\mathrm{com}}_t$ over $R_{\mathrm{com}}$ rounds. 
Each agent then constructs a mean-field estimate from the aggregated counts, assigning the remaining probability mass uniformly to unobserved states. 
The resulting estimator is $\big(1-\mathcal{O}(R_\mathrm{com}/|\X|)\big)$-accurate.


\textit{Dynamic-Projected Consensus.}
Instead of communicating agent counts, we propose the dynamic-projected consensus (D-PC) algorithm that exchanges mean-field estimates directly. 
The agent at state $x$ first constructs an initial estimate $\hat \nu^{(0)}_{x,t} \in \mathcal{R}(x)$, where $\mathcal{R}(x) \subseteq \mathcal{P}(\Y)$ denotes the set of distributions consistent with locally observed opponent counts. 
During communication round $\tau$, agents exchange their previous estimates and compute the weighted consensus
$\xi^{(\tau)}_{x,t} = \sum_{x' \in \mathrm{Neighbor}} w(x,x') \hat \nu^{(\tau-1)}_{x',t}$,
followed by the projection
$\hat \nu^{(\tau)}_{x,t} = \Omega_{\mathcal{R}(x)} (\xi^{(\tau)}_{x,t})$, which preserves consistency with locally observed data.
Compared to classical consensus, D-PC is $\mathcal{O}\!\left(e^{-cR_{\mathrm{com}}}\right)$-accurate (Theorem~\ref{thm:estimator-eval-dyn-dpc}), i.e., exponentially decreasing error with $R_{\mathrm{com}}$, and is empirically more robust to noisy communication.

\vspace{-0.05in}
\section{Numerical Experiments}\label{sec:numerical-exp}
\vspace{-0.05in}
%

%
%

\begin{table}[b!]
\centering
\caption{Performance comparison for cRPS}
\resizebox{\linewidth}{!}{\begin{tabular}{ccccc} 
\toprule
 Algorithm  & Critic  Input & Average Reward & NE Attained & Training Time \\ 
 \midrule
  MF-IPPO  & $\{(x^{N_1}_{i,t}, \mu_t)\}_{i=1}^{N_1}$ & -0.340 & \cmark & 2:45:11\\ 
 MAPPO-PS  & $\{(x^{N_1}_{i,t}, \mu_t, \nu_t)\}_{i=1}^{N_1}$ & -0.313 & \cmark & 2:47:54\\ 
 MAPPO-CC  & $(\x^{N_1}_t, \y^{N_2}_t, \mu_t, \nu_t)$ & -0.342 & \cmark & 3:20:46\\ 
 DDPG-MFTG  & $(\phi_t, \mu_t, \nu_t)$ & 3.774 & \xmark & 60:49:35\\ 
 MF-MAPPO  (Ours) & $(\mu_t, \nu_t)$ & \textbf{-0.331} & \cmark & \textbf{2:17:15}\\
\bottomrule
\end{tabular}}
\label{table:crps-eval}
\end{table}

We showcase MF-MAPPO's efficacy on the following three large-population scenarios \edit{using our novel benchmark simulation platform, \texttt{MFEnv}\footnote{\url{https://github.com/bhavini-jeloka/mf-env}}}: 
(1) a population-based rock–paper–scissors game, 
(2) a battlefield setting where Blue and Red teams engage in attack–defense tasks, and
(3) a SIS/epidemiology scenario with probabilistic infection dynamics.
Additional details of the experiment setups and results are presented in Appendix~\ref{appendix-additional-results}.
To the best of our knowledge, DDPG-MFTG~\citep{shao2024reinforcement} is the only method explicitly designed for mixed cooperative--competitive large-population team problems. 
Accordingly, we compare against DDPG-MFTG and three mean-field ``house-made" adaptations of standard MARL methods: MF-IPPO and two MAPPO variants, parameter sharing (MAPPO-PS) and centralized critic (MAPPO-CC)~\citep{yu2022surprising}. 
Centralized Q-learning methods such as QMix~\citep{rashid2020monotonic}, MADDPG~\citep{lowe2017multi}, and FACMAC~\citep{peng2021facmacfactoredmultiagentcentralised} are omitted, since their centralized Q-functions depend on joint actions and states, making them intractable in large-population settings.


\vspace{-0.05in}
\subsection{Validation of Convergence: Constrained Rock-Paper-Scissors (cRPS)}
\vspace{-0.05in}

We use a population-based rock–paper–scissors (RPS) game to validate the convergence of MF-MAPPO to equilibrium before deploying it to more complex scenarios. 
This setting is chosen for its simplicity, which allows the equilibrium to be computed analytically.
%
%
The individual state space is $\{\texttt{R,P,S}\}$,
and we consider a restriction of the action space to $ \{\texttt{CW}, \texttt{Stay}\}$, i.e,  either move clockwise ($\texttt{R} \!\rightarrow\! \texttt{P}$, $\texttt{P} \!\rightarrow\! \texttt{S}$, $\texttt{S} \!\rightarrow\! \texttt{R}$) or remain on current state.
%
%
At each time step $t$, the reward is given by $r_{t}(\mu_t, \nu_t) = {\mu}_t\t A {\nu}_t$ where $A$ is the standard RPS payoff matrix. 
The equilibrium state distribution is uniform for both teams. 
%

Table~\ref{table:crps-eval} compares MF-MAPPO with baseline algorithms, along with their critic inputs. 
We test the learned policies with a fixed initial distribution, with all Blue agents initialized at \texttt{R} and Red agents at \texttt{P}.
One can show analytically that the unique game value is $-1/3$ for this initial condition.
Furthermore, we have the optimal conditions for the trajectories:
(i) ${\mu}_t^* \!=\! {\nu}_t^* \!=\! [\frac{1}{3}, \frac{1}{3}, \frac{1}{3}]\t$ for all $t\geq2$, and 
(ii) ${\mu}^*_{1} \!=\! [0, 1\!-\!\delta, \delta]\t$ where $\delta \!\in\! [\frac{1}{3}, 
\frac{2}{3}]$ and ${\nu}^*_{1} \!=\! [0, \frac{2}{3}, \frac{1}{3}]\t$.
See Proposition~\ref{prop:two-step-convergence} in Appendix~\ref{AppendixB} for details. 
As seen in Table~\ref{table:crps-eval}, MF-MAPPO achieves rewards closest to the analytical Nash value $(-1/3)$, indicating that the minimally-informed critic preserves value information despite its reduced input dimensionality. 
This leads to faster and more efficient training as shown by the training time.

Figure~\ref{fig:crps-simplex-combined}(a) compares trajectories of MF-MAPPO and DDPG-MFTG and plots the analytical optimal distribution to be achieved at $t=1$.
MF-MAPPO follows the prescribed optimal trajectories and converges to the equilibrium uniform distribution; 
DDPG-MFTG, however, fails to converge,
since it conditions only on the ego team’s policy, limiting its ability to respond to adversarial behaviors.
Figure~\ref{fig:crps-simplex-combined}(b) shows trajectories of MF-MAPPO policies trained with 1{,}000 agents and deployed to larger populations. 
The average population distribution converges to equilibrium, while larger populations reduce trajectory variance, consistent with Theorem~\ref{thm:extend-to-many}.

\begin{figure}[t!]
    \centering
            \includegraphics[width=\linewidth]{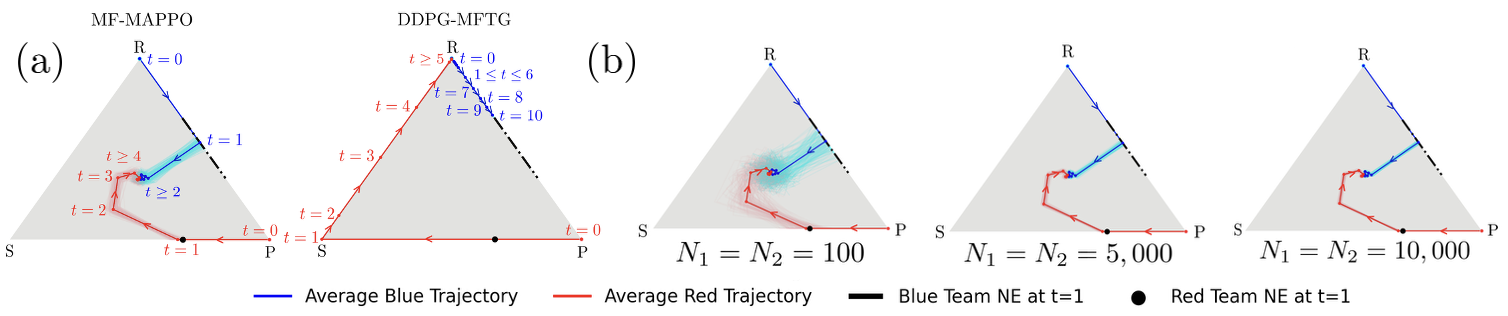}
    \caption{%
    (a) MF-MAPPO and DDPG-MFTG trajectories.
    (b) MF-MAPPO under varying team sizes.
    }
    \label{fig:crps-simplex-combined}
\end{figure}

\vspace{-0.05in}
\subsection{Emergent Heterogeneous Team Behaviors}
\vspace{-0.05in}
With convergence validated, we deploy MF-MAPPO to two more complex grid-based scenarios, whose setups are briefly described below.
All experiments use $N_1 = N_2 = 100$ agents.

\textbf{Battlefield Game.}
This game, featuring an activation--deactivation mechanism, was introduced in Section~\ref{sec: problem-formulation}. 
The reward is defined as
\(
    r_t(\mu_t, \nu_t) = \sum_{x \in \mathrm{Target}} \mu_t(x),
\)
which incentivizes the (maximizing) Blue team to reach target cells, while the (minimizing) Red team seeks to block or deactivate Blue agents. 
Additional details are provided in Appendix~\ref{AppendixC}.

\textbf{Epidemiology Game.}
In this scenario, Red agents emulate a virus that infects the Blue population, with infected Blue agents further spreading the disease. 
The Blue team aims to maximize the size of its healthy population, leading to the reward
\(
r_t(\mu_t, \nu_t) = \sum_{x \in \mathrm{Healthy}} \mu_t(x).
\)
Accordingly, Blue agents avoid infected peers and move toward hospital states for recovery, while the Red (virus) team tracks the Blue population and block access to hospitals.
See Appendix~\ref{Appendix-epidemic} for additional details.

%
%

%
%

%

We compare MF-MAPPO with baselines via head-to-head evaluation for both scenarios 
with full mean-field feedback (Table~\ref{table-4x4-eval}). 
MF-MAPPO consistently achieves superior performance:
Fixing the Red policy (row) shows that the MF-MAPPO Blue team always attains the highest performance, while fixing the Blue policy (column) shows that the MF-MAPPO Red team performs the best except against Blue DDPG-MFTG in Epidemiology%
\footnote{
MF-MAPPO adopts a more conservative policy that is robust to stronger opponents, and therefore exploits weak strategies, such as the DDPG-MFTG Blue policy, less aggressively than MAPPO-CC.}.
Figure~\ref{fig:training-architecture}(b) presents the training curves in the Battlefield scenario, where MF-MAPPO converges faster and achieves higher rewards than all baselines. 
These results highlight the advantages of combining a minimally informed critic with a shared team actor: 
the resulting networks are smaller, have input dimensions independent of the population size, and achieve better performance. 
Notably, DDPG-MFTG exhibits initial performance gains but subsequently fails to maintain and further improve, and we were unable to stabilize the algorithm despite extensive hyperparameter tuning.

\begin{table}[!b]
\centering
\caption{Performance evaluations averaged over 10 random seeds.}
\label{table-4x4-eval}
\setlength{\tabcolsep}{8pt} 
\resizebox{\linewidth}{!}{\begin{tabular}{ccccccc}
\toprule
\textbf{Battlefield}  & {\textcolor{darkblue}{MF-IPPO}} & {\textcolor{darkblue}{MAPPO-PS}} & {\textcolor{darkblue}{MAPPO-CC}} & {\textcolor{darkblue}{DDPG-MFTG}} & {\textcolor{darkblue}{MF-MAPPO (Ours)}} \\
 \textcolor{darkred}{MF-IPPO}    & 61.54 & 73.85 & 70.60 & 0.00 & \textbf{86.26} \\
 \textcolor{darkred}{MAPPO-PS}   & 55.82 & 72.78 & 71.64 & 0.00 & \textbf{86.22} \\
 \textcolor{darkred}{MAPPO-CC}   & 51.60 & 68.81 & 65.58 & 0.00 & \textbf{82.58} \\
 \textcolor{darkred}{DDPG-MFTG}  & 65.26 & 75.14 & 77.21 & 0.00 & \textbf{85.95} \\
 \textcolor{darkred}{MF-MAPPO (Ours)}   & \textbf{36.10} & \textbf{54.43} & \textbf{62.00} & 0.00 & \textbf{76.26} \\
\midrule
 \textbf{Epidemiology} &{\textcolor{darkblue}{MF-IPPO}} & {\textcolor{darkblue}{MAPPO-PS}} & {\textcolor{darkblue}{MAPPO-CC}} & {\textcolor{darkblue}{DDPG-MFTG}} & {\textcolor{darkblue}{MF-MAPPO (Ours)}} \\
 \textcolor{darkred}{MF-IPPO}    & 62.76 & 58.26 & 59.48 & 35.56 & \textbf{77.83} \\
 \textcolor{darkred}{MAPPO-PS}   & 57.95 & 57.63 & 55.41 & 29.52 & \textbf{79.78} \\
 \textcolor{darkred}{MAPPO-CC}   & 55.64 & 55.16 & 55.13 & \textbf{27.95} & \textbf{76.74} \\
\textcolor{darkred}{DDPG-MFTG}  & 99.83 & 93.99 & 88.26 &  80.22 & \textbf{107.39} \\
\textcolor{darkred}{MF-MAPPO (Ours)}  &  \textbf{54.71} & \textbf{46.38} & \textbf{47.42} & 30.50 & \textbf{71.04} \\
\bottomrule
\end{tabular}}
\end{table}

\begin{figure}[t]
    \centering
    \includegraphics[width=\linewidth]{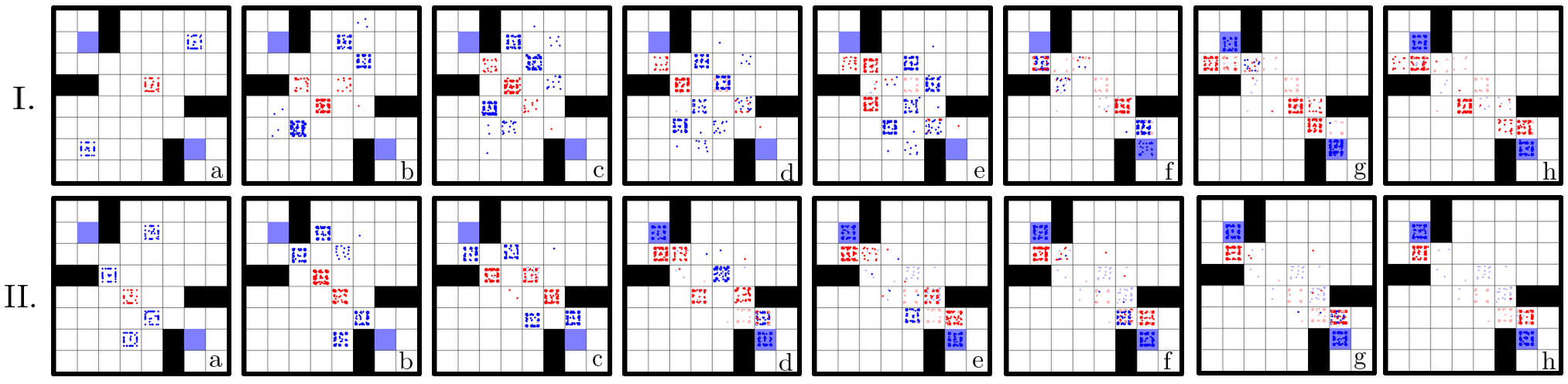}
    \caption{I. Red is concentrated; 30\%  Blue are at {the bottom, the rest are at the top} II. Blue is evenly split, Red is concentrated. }
    \label{fig:battle-comparison}
    \vspace{-0.1in}
\end{figure}

We deploy MF-MAPPO to a more challenging 8$\times$8 battlefield in~Figure~\ref{fig:battle-comparison} where the Red team faces a dilemma in deciding which target to defend, while the Blue team exploits this ambiguity. 
Since DDPG-MFTG and other homemade baselines require high computation due to their high network input dimensions (and DDPG-MFTG has a high output dimension), they are excluded from our analysis, and we qualitatively assess MF-MAPPO's performance. 
Panels I–II present sequences of MF-MAPPO movements under different initial distributions. 
In subplot I, the Blue team fakes a move towards the top-left target  and then switches towards the other target to trick the Red team. 
In subplot II, each Blue sub-team first reaches towards the closest target, and the Red team splits and tries to block both corridors. 
The highlighted Blue subteam realized the corridor would be blocked by the Red team and headed towards the bottom-right target. 
The Red team then draws part of its force blocking the upper-left target to reinforce the blockage at the bottom-right, and is able to deactivate all agents in the highlighted Blue subteam. 
These two game sequences illustrate how identical policies can generate complex, heterogeneous team behaviors, including emergent Blue agent target adaptation and Red team target-defense role distribution and reinforcement behaviors, highlighting the flexibility and richness of the mean-field approximation and identical team strategies.

\vspace{-0.05in}
\subsection{Performance under Partial Observability}
\vspace{-0.05in}

We evaluate GR-MF-MAPPO with $\epsilon$-accurate estimators in a partially observable battlefield setting where each team must estimate the opponent’s distribution (Figure~\ref{fig:map3}). 
\edit{When GR-MF-MAPPO policies are paired with D-PC estimators for both teams, the cumulative regret decreases as $R_\textrm{com}$ increases from 5 to 12 (Table~\ref{table-dpc}), aligning with Proposition~\ref{thm:estimator-eval-dyn}.
Although $R_\textrm{com}=3$ yields lower regret than 5 and 10, its substantially higher variability and wider confidence interval indicate less reliable estimation, with additional communication rounds yielding more consistent performance.}
We further empirically study the impact of communication quality (Figure~\ref{fig:map3}I). 
With $R_{\mathrm{com}} = 7$, increasing Gaussian noise variance leads to larger regret deviations; however, these remain within $5\%$, as the D-PC projection step mitigates noise accumulation from communication.
Panel II presents an ablation study of gradient regularization. 
Larger $\lambda^{\text{GR}}$ penalty imposes stronger smoothness (Lipschitz) constraints, yielding reduced trajectory deviation and thus improved robustness against estimation errors.
%
%
%
Finally, Panel III shows the estimation error of the opponent distribution over time under varying communication rounds. 
Both estimators achieve comparable total variation errors, with D-PC performing better under limited communication budgets {($R_{\mathrm{com}} < 10$)}. 

\begin{table}[!b]
\centering
\caption{Cumulative regret under varying $R_{\textrm{com}}$}
\label{table-dpc}
{\begin{tabular}{cccc}
\toprule
$R_\textrm{com}$  & {Average Regret} & {95\% Confidence Interval}  & CI Width\\
\midrule
 3   & $4.563\pm 0.621$	 & [3.021, 6.105] & 3.084\\
 5  & $4.967\pm 0.517$	 & [3.682, 6.252]	& 2.570 \\
 10  & $ 4.687\pm 0.408$	 & [3.674, 5.699] & 2.026	\\
12  & $4.530\pm 0.333$	 &  [3.703, 5.357] & 1.654	\\
\bottomrule
\end{tabular}}
\end{table}

\begin{figure}[!t]
    \centering
    \includegraphics[width=0.95\linewidth]{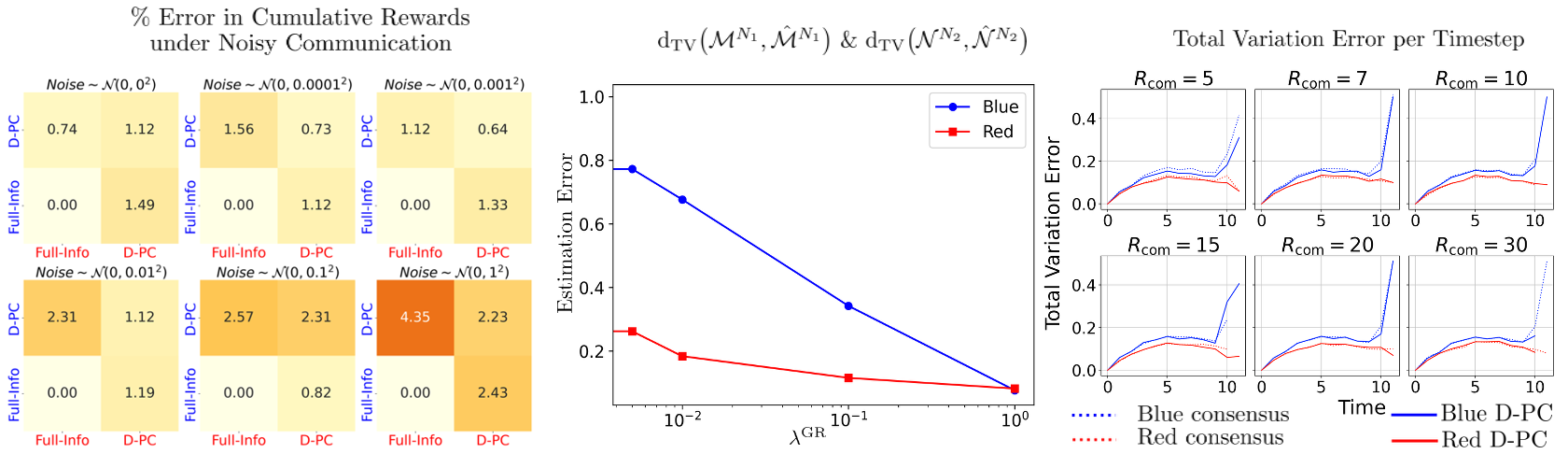} 
    \caption{
    I. Reward error under different communication noises;
    II. Effect of gradient regularization on estimation quality;
    III. Estimation error of Classical Consensus and D-PC under varying $R_{\textrm{com}}$.} 
    \label{fig:map3}
    \vspace{-0.2in}
\end{figure}

\vspace{-0.05in}
\section{Conclusion}\label{sec:conclusion}
\vspace{-0.05in}
In this work, we studied zero-sum mean-field team games and introduced MF-MAPPO, a scalable MARL framework for mixed cooperative--competitive large-population systems. 
MF-MAPPO combines a minimally informed critic with shared team policies, enabling efficient learning in finite-population simulators while scaling to thousands of agents. 
We further extended the framework to partially observable settings through gradient regularization and decentralized mean-field estimation algorithms. 
Experiments on large-scale benchmarks demonstrated that MF-MAPPO outperforms existing MARL and mean-field baselines while generating rich emergent team behaviors. 
Future work includes improving scalability in high-dimensional state spaces and extending the framework to more general communication and interaction structures.

\begin{ack}
\edit{This work was supported in part by ONR Award N00014-23-1-2308 and by ARL DCIST CRA W911NF-17-2-0181.}
\end{ack}

\bibliographystyle{unsrtnat}
\bibliography{neurips_references}
\appendix

\newpage
\appendix
\onecolumn
\setcounter{equation}{0}
\renewcommand{\theequation}{\thesection.\arabic{equation}}
\renewcommand{\thefigure}{\thesection.\arabic{figure}} 
\setcounter{figure}{0} 
\setcounter{table}{0}
\setcounter{algorithm}{0}
\renewcommand{\thealgorithm}{\thesection}
\renewcommand{\thetable}{\thesection.\arabic{table}}

\section{Related Work}
\textbf{Mean-Field Teams.} The single-team problem was explored in~\cite{arabneydi2014team}, where agents share a common team reward (MFC).
In contrast,~\cite{Mahajan2015linear} established optimality for finite-population games \emph{only} in the LQG setting, {while~\cite{sanjari2022nash} analyzed a two-team setting with continuous states and actions---unlike our finite state-space formulation that directly admits the familiar MDP-type structure.
Similarly, \textit{multi-population} MFGs (MP-MFGs) have been studied in the past} but often restrict agent dynamics and policies to be independent of both other agents and other population distributions, see \cite{perolat2021scalingmeanfieldgames} and references therein.
More recently,~\cite{guan2024zero} introduced \textit{zero-sum mean-field team games} (ZS-MFTGs), modeling large-population teams that compete while cooperating internally.
A common-information decomposition~\citep{nayyar2013decentralized} enables team-size-independent learning with identical optimal team policies {using MF feedback, unlike aforementioned MP-MFGs (open-loop MF policies).} 
{However, computing such policies numerically, especially in large state-action spaces using dynamic programming, is costly.}
%
By contrast, MF-MAPPO (Algorithm~\ref{alg:mf-mappo}) leverages shared actor–critic networks per team and uses only commonly accessible information (exact/estimated), ensuring tractability and scalability.

\textbf{Mean-Field Theory and Learning.}
Recent advances in mean-field learning span Q-function–based methods, such as MF-Q and MF-AC~\citep{yang2018mean} and DDPG-MFTG~\citep{shao2024reinforcementlearningfinitespace}, to value-function–based methods like Dec-POMFPPO~\citep{cui2023learning} (MFC only).
{While DDPG-MFTG incorporates team games, it is restricted to simple grid worlds, unlike our focus on tightly coupled collaborative–competitive domains (ZS-MFTGs). 
We adopt it as a baseline and show that MF-MAPPO consistently outperforms it in stability and performance.}  
Other related works include PMD-TD for MFGs~\citep{yardim2024exploiting} and GAN-based ECA-Net for continuous-space attack–defense games~\citep{wang2022multi}, both differing in scope and structure.
Moreover, \cite{yang2018mean, subramanian2020partially, subramanian2022multitypemeanfield, subramanian2022decentralizedmeanfieldgames} define mean fields over (neighboring) agents’ actions rather than their state distribution, yielding a rich line of research that is orthogonal to our formulation.
MF-MAPPO extends PPO~\citep{schulman2017proximal} to competitive MFTGs, using team distributions as critic inputs for scalability, a shared actor–critic per team with a single buffer for efficiency, and simultaneous team training to avoid the inefficiencies of iterative best-response methods~\citep{lanctot2017unified, smith2021iterative}.
Unlike prior MF methods that rely on infinite-population oracles~\citep{shao2024reinforcement, perolat2021scalingmeanfieldgames, carmona2021modelfreemeanfieldreinforcementlearning}, MF-MAPPO is trained directly in finite-population simulators, making it suitable for realistic deployment.
{Finally, to standardize evaluation in large-scale MFTGs, we create novel MFTG benchmark environments (Constrained Rock–Paper–Scissors, Battlefield) going beyond existing ones that either focus only on MFGs~\citep{guo2023mfglib} or omit MF coupling altogether~\citep{terry2021pettingzoo, zheng2018magent}.}

\textbf{Mean-Field Estimation and Opponent Modeling.}
{To enable reactive behavior to opponents’ actions, most existing} MF approaches assume centralized or exact knowledge of the opponent's MF, which is rarely practical.
While~\cite{cui2023learning} considers partial observability, it is limited to MFC. 
Estimation methods such as kernel density estimation~\citep{inoue2021model} or normalizing flows~\citep{perrin2021meanfieldgamesflock} struggle with discrete spaces and limited agent-visibility.
%
%
%
Communication-based methods~\citep{benjamin2025networkedcommunicationmeanfieldgames, benjamin2025networked} address {these} constraints but assume uniform estimates for unobserved states and rely on perfect multi-round communication {while being restricted to fully cooperative or competitive regimes}; noisy variants may even produce invalid distributions.
%
%
We instead propose Dynamic-Projected Consensus (D-PC), a constrained consensus algorithm that ensures validity, exponential convergence, and bounded deviations when paired with a {\textit{gradient-regularized}} MF-MAPPO {policy}.
%
Gradient regularization naturally stabilizes MF-MAPPO in partially observable MFTGs.
%
%
Experiments show that D-PC matches baseline performance and even outperforms them under limited communication, especially critical in adversarial settings requiring rapid adaptation~\citep{richards2012developing}, enhancing robustness and fault tolerance.
To our knowledge, this is the first MF estimator for opponent modeling in competitive team settings.

\begin{algorithm}[t]\label{alg:mf-mappo}
\caption{Mean-Field Multi-Agent Proximal Policy Optimization (\textbf{MF-MAPPO})}
\begin{algorithmic}
    \State \textbf{Initialize:} NN parameters $\{\theta_{\text{Blue}}, \zeta_{\text{Blue}}\}$ and $\{\theta_{\text{Red}}, \zeta_{\text{Red}}\}$; step size sequences $\{\alpha_m\}$ and $\{\beta_m\}$; entropy decay sequence $\{\omega_m\}$

    \For{$i = 1,2,\ldots  $} 
    \State $(\phi_{\theta_{\text{Blue}}^{\text{old}}}, \psi_{\theta_{\text{Red}}^{\text{old}}}) \leftarrow (\phi_{\theta_{\text{Blue}}}, \psi_{\theta_{\text{Red}}})$
        \For{$t = 0, 1, \ldots, T_{\text{rollout}}$} 
            \State Sample joint actions
            \State $u_{i,t} \sim \phi_{\theta_{\text{Blue}}^{\text{old}}}(x_{i,t}, \mu_t^{N_1}, \nu_t^{N_2})$,
            $v_{j,t} \sim \psi_{\theta_{\text{Red}}^{\text{old}}}(y_{j,t}, \mu_t^{N_1}, \nu_t^{N_2})$
            \State Step environment according to kernels $(f_t, g_t)$
            \State Collect samples $(\bfx_{t+1}^{N_1}, \bfy_{t+1}^{N_2}, \mu^{N_1}_{t+1}, \nu^{N_2}_{t+1}, \bfu_t^{N_1}, \bfv_t^{N_2}, r_t)$
        \EndFor
        
        \For{$K$ epochs} 
            \State Update $\{\theta_{\text{Blue}}, \zeta_{\text{Blue}}\}$ and $\{\theta_{\text{Red}}, \zeta_{\text{Red}}\}$ using (\ref{eq:critic network}--\ref{eq:actor network})
        \EndFor
    \EndFor  
    \State \textbf{Return: $(\phi_{\theta_{\text{Blue}}}, \psi_{\theta_{\text{Red}}})$}
\end{algorithmic}
\end{algorithm}

\section{Details of D-PC}\label{AppendixDPCDeets}


\begin{definition}[Visibility Graph]
\label{def:viz-graph}
The bipartite state-based visibility graph at time \( t \), denoted by \( \mathcal{G}_t^{\mathrm{viz}} = (\X, \Y, \mathcal{E}_t^{\mathrm{viz}}) \), is an undirected graph where (i) The vertex sets consist of the states $x\in\X$ and $y\in \Y$; (ii) An undirected edge \( (x, y)  \in\mathcal{E}_t^{\mathrm{viz}} \) exists if and only if all Blue agents at state \( x \) are mutually visible with all Red agents at state \( y \) at time \( t \); that is, every Blue agent at state \( x \) is visible to every Red agent at state \( y \) at time $t$.
\end{definition}
\edit{We further define the $|\mathcal{X}| \times |\mathcal{X}|$ matrix $M_t$  as $M_t[p, q] = 1 $ if $(x_p, x_q) \notin \mathcal{E}_t^{\mathrm{viz}}$, else $0$.
$M_t[p,q] = 1$ indicates that state $x_q$ is \textbf{not} visible from state $x_p$ at time $t$, and vice versa. 
\begin{assumption}\label{assmpt:invertible-invisible}
     $M_t$ is invertible. 
\end{assumption}}

Before describing the communication architecture, we define $\X^{o}_t \subseteq\X$ such that $\hat\mu_t(x) > 0 $ if and only if $ x\in\X^{o}_t$.
\begin{definition}[Communication Graph]
\label{def:comms-graph}
The state-based {communication graph} at time \( t \), denoted by \( \mathcal{G}_{{\rm Blue},t}^{\mathrm{com}} = (\X^{o}_t, \mathcal{E}_t^{\mathrm{com}}) \), is an undirected graph where (i) The vertex set $\X^{o}_t$ consists of the states \( x \in \X^{o}_t \); (ii) An undirected edge \( (x_p, x_q) \in \mathcal{E}_t^{\mathrm{com}} \) exists if and only if all agents at state \( x_p \) can communicate with all agents at state \( x_q \) at time \( t \).
\end{definition}

We denote the set of neighbors of vertex $v$ of $\mathcal{G}_t^{\mathrm{viz}}$ and $\mathcal{G}_{{\rm Blue},t}^{\mathrm{com}}$ respectively, as $\N_t^{\mathrm{viz}}(v)$ and $\N_{{\rm Blue},t}^{\mathrm{com}}(v)$.
$\Y^{o}_t \subseteq\Y$, \( \mathcal{G}_{{\rm Red},t}^{\mathrm{com}}\) and $\N_{{\rm Red},t}^{\mathrm{com}}(v)$ are defined similarly.
The following assumption on the topology of the communication graph will simplify the mean-field estimator design problem.

\begin{assumption}
    \label{assumpt:graph-connectedness}
    We assume that $\mathcal{G}_{{\rm Blue},t}^{\mathrm{com}}$ and $\mathcal{G}_{{\rm Red},t}^{\mathrm{com}}$ are connected for all $t$.
\end{assumption}

\edit{This assumption ensures that information from all agents contributes to the mean-field estimate. 
    Without connectivity, disconnected agents cannot communicate their states, and convergence of the estimated mean-field distribution to the true distribution is no longer guaranteed~\citep{benjamin2025networked}.
    Studying disconnected or intermittently connected communication graphs is an interesting direction for future work.}

At time $t$, we assume that each agent $i\in[N]$ observes the population distribution of its opponent's visible states $y\in \mathcal{N}_t^{\mathrm{viz}}(x_{i,t})$. 
Thus, from the perspective of agent $i$, the ED $\hat\nu^{N_2}_{i,t}$ is defined on the constrained simplex $\textstyle \hat\nu^{N_2}_{i,t} \in \R(x_{i,t}) \triangleq \{\eta \in \P(\Y)~~|~~ \eta(z) =  \hat\nu^{N_2}_t(z), ~~\forall~~ z \in \mathcal{N}_t^{\mathrm{viz}}(x_{i,t})\}.$ given by 
$\R(x_{i,t})$ combines the Red team's mean-field components known with certainty by Blue agent $i$—i.e., the observable states \( \mathcal{N}_t^{\mathrm{viz}}(x_{i,t}) \)—with those that must be estimated. 

\begin{proposition}\label{prop:convex-constraints}
For all $x\in\X$,  $\mathcal{R}(x)$ is a closed and convex set. 
\end{proposition}
It can be easily seen that the true mean-field distribution is an element of this constrained simplex for each agent, i.e., 
$\hat\nu^{N_2}_{t} \in \R(x_{i,t})$ for all $  i\in [N_1]$ and time $t$.
The following assumption imposes uniformity of the estimate of the ED across agents in a given state.
\begin{assumption}
    \label{assumpt:state-dependent-estimate}
    We assume that $ \hat\nu^{N_2}_{i,t} = \hat\nu^{N_2}_{j,t} \triangleq \hat\nu^{N_2}_{x,t} $ for all agents satisfying $x_{i,t}^{N_1}=x_{j,t}^{N_1} = x$.
\end{assumption}

\begin{definition}\label{def:proj-op}
    Let $x\in\X$. The projection operator $\Omega_{\mathcal{R}(x)}:\mathbb{R}^{|\Y|}\to \R(x)$ where $\R(x)$ is a closed and convex set is defined as $\Omega_{\mathcal{R}(x)}[\eta] = \omega^* \triangleq\arg\min_{\omega\in\R(x)}\|\eta-\omega\|_2, ~~\eta\in\mathbb{R}^{|\Y|}$, 
    where $\|\cdot\|_2$ is the standard Euclidean norm.
\end{definition}

At each communication round, D-PC performs the following two steps:
\begin{enumerate}
    \item {Weighted average consensus:}
    \begin{equation}\label{eq:avg-consensus}
    \xi_{x, t}^{\tau} = \sum_{z \in \N_{{\rm Blue},t}^{\rm com}(x)\cup x} w^{\rm Blue}_{(x,z)} \hat\nu^{\tau-1}_{z,t}, 
    \quad \sum_z w^{\rm Blue}_{(x,z)} = 1, \quad x\in\X^o.
    \end{equation}
    
    \item {Projection onto the constraint set:}
    \begin{equation}\label{eq:projection}
    \hat\nu^{\tau}_{x,t} = \Omega_{ \mathcal{R}(x)}\left[\xi_{x, t}^{\tau}\right], ~x\in\X^o.
    \end{equation}
\end{enumerate}

\begin{assumption}\label{assumpt:doubly-stochastic-weights}
    The matrix $W^{\rm Blue} \in \mathbb{R}^{|\X^o|\times |\X^o|}$ formed by the non-negative weights is symmetric and doubly-stochastic.
    Furthermore, it respects the sparsity structure of the communication graph $\G^{\mathrm{com}}$, that is, $w^{\rm Blue}_{(x, z)} > 0 $ if and only if $ {z \in \mathcal{N}^{\mathrm{com}}(x)\cup x}$. 
    %
    $W^{\rm Red}$ is constructed similarly.
\end{assumption}

To satisfy Assumption \ref{assumpt:doubly-stochastic-weights}, one may use the well-known Metropolis Matrix~\citep{xiao2006distributed}.
We define the smallest non-zero entry of $W$ as 
\begin{equation}\label{eq:theta-def}
\begin{aligned}
    \theta \triangleq \min_{(x, y)} \; & \{w_{(x, y)}>0\}. 
\end{aligned}
\end{equation}

For clarity of the D-PC pseudocode presented in Algorithm~\ref{alg:cap-dyn}, the estimation is presented from the Blue team's perspective which models the distribution of its opponent (Red team).
One can simultaneously run this estimation algorithm from the Red team's perspective.
\begin{algorithm}[ht]\label{alg:cap-dyn}
\caption{Dynamic Projected Consensus for Mean-Field Estimation (\textbf{D-PC})}
\begin{algorithmic}
    \State \textbf{Initialize:} $\hat\mu^{N_1}_{t=0} =\mu^{N_1}_{t=0}$ and $\hat\nu^{N_2}_{t=0} =\nu^{N2}_{t=0}$, identical MF-MAPPO trained Lipschitz policies $(\phi^*, \psi^*)$ and graphs $\G^{\mathrm{viz}}_{t=0}$ and $(\mathcal{G}_{{\rm Blue},t}^{\mathrm{com}}, \mathcal{G}_{{\rm Red},t}^{\mathrm{com}})$
    \For{$t= 0, \ldots, T$}
    \State Update the graphs $\G^{\mathrm{viz}}_{t}$ and $(\mathcal{G}_{{\rm Blue},t}^{\mathrm{com}}, \mathcal{G}_{{\rm Red},t}^{\mathrm{com}})$.
    \State Receive reward $r_t(\hat\mu^{N_1}_t, \hat\nu^{N_2}_t)$
    \State Set $\tau = 0$
    \State Initialize $\hat\nu^{N_2,\tau=0}_{x,t}$ ( $\hat\mu^{N_1,\tau=0}_{y,t}$) for all $ x\in\X^o_t$ ($ y\in\Y^o_t$) using $\G^{\mathrm{viz}}_{t=0}$
    \State Compute matrix $W^{\rm Blue} (W^{\rm Red})$
    \For{$\tau < R_{\mathrm{com}}$} 
    \State Communicate mean-field estimate to neighbors based on $\mathcal{G}_{{\rm Blue},t}^{\mathrm{com}}(\mathcal{G}_{{\rm Red},t}^{\mathrm{com}})$
    \State $\tau = \tau + 1$
        \For{ all $x\in \X^o_t$ ($y\in\Y^o_t$)} 
            \State Compute weighted average consensus~\eqref{eq:avg-consensus}
            \State Project result onto constraint set~\eqref{eq:projection}
        \EndFor
    \EndFor
    \State \textbf{Return: $\hat\nu^{N_2,R_{\mathrm{com}}}_{x,t}  (\hat\mu^{N_1,R_{\mathrm{com}}}_{y,t})$}
    \EndFor
\end{algorithmic}
\end{algorithm}
\begin{remark}
One could alternatively cast this problem as a Partially Observable Markov Decision Process (PO-MDP)~\citep{bernstein2002complexity}, where the environment provides arbitrary observations.
However, our focus is on policies that explicitly depend on the mean-field distributions of the two teams, for two key reasons: (a) Distinct opponent strategies induce different mean-field trajectories, and closed-loop feedback on these distributions enables an appropriate response to the strategies deployed by the opponent~\citep{guan2024zero}; and (b) Access to team-level distributions facilitates credit assignment~\citep{pignatelli2024surveytemporalcreditassignment} in MARL by allowing agents to reason about which collective distributions are optimal.
Thus, having access to MF information—particularly that of the opponent—is desirable. 
Yet, it is unrealistic to assume that the environment directly provides these distributions, which motivates the design of estimation algorithms proposed in this work.
Finally, unlike general PO-MDP formulations that often require maintaining long observation histories, our setting involves population sizes for which such history-based tracking is computationally infeasible in terms of memory.
\end{remark}

The following theorem establishes the convergence guarantees of D-PC.
\begin{restatable}{theorem}{maintheoremfour}\label{thm:estimator-eval-dyn-dpc}

\edit{Under Assumption~\ref{assmpt:invertible-invisible},} D-PC is $\mathcal{O}\!\left(e^{-cR_{\mathrm{com}}}\right)$-accurate. 
\end{restatable} 

\begin{proof}[Sketch of Proof]
The exponential convergence of D-PC at each time-step $t$ follows directly from Theorems 6-8 in~\cite{nedich2014lyapunov} as the number of communication rounds $R_\textrm{com}\to\infty$.
Therefore, for finite communication rounds, there exist constants \(c_x, c_y > 0\) and \(\rho_x, \rho_y \in (0,1)\) such that for any number of communication rounds \(R_{\mathrm{com}}\) and any underlying distributions $(\hat\mu_t^{N_1}, \hat\nu_t^{N_2})$ D-PC satisfies
\begin{align}\label{eq:estimator-eval-bound}
   \dtv{\hat\nu_{x, t}^{N_2}, \hat\nu_t^{N_2}} < c_x \rho_x^{R_{\mathrm{com}}} \triangleq \epsilon_x, \quad \text{for all } x \in \mathcal{X}^o \text{ and } t \nonumber\\
   \dtv{\hat\mu_{y, t}^{N_1}, \hat\mu_t^{N_1}} < c_y\rho_y^{R_{\mathrm{com}}} \triangleq \epsilon_y, \quad \text{for all } y \in \mathcal{Y}^o \text{ and } t.
\end{align}
We define $\epsilon = \max(\epsilon_x, \epsilon_y).$
\end{proof}
\begin{proof} We begin by defining set regularity.
    \begin{definition}[Set Regularity~\citep{nedich2014lyapunov}]\label{defn:set-regularity}
    Let $Z \subseteq \mathbb{R}^n$ be a nonempty set.
    A (finite) collection of closed convex sets $\{Y_j\}_{j=1}^J \subseteq \mathbb{R}^n$ is
regular (in Euclidean norm) with respect to the set $Z$, if there exists a constant $r\geq1$ such that 
\begin{align*}
\inf_{y\in Y}\|z-y\|_2 \leq r \max_{j}\inf_{y\in Y_j}\|z-y\|_2, z\in Z,
\end{align*}
where $Y  \triangleq\bigcap_{j=1,\ldots, J}Y_j$ is non-empty.
We refer to the scalar $r$ as the regularity constant.
When the preceding relation holds with $Z=\mathbb{R}^n$, we say that the sets $\{Y_j\}_{j=1}^J$ are uniformly regular.
\end{definition}

\begin{proposition}\label{prop:set-regularity}
    The constraint sets $\R(x)$ for all $x\in\X^o$ are uniformly regular with regularity constant $\kappa\geq 1$, i.e.,
    \begin{align*}
        \inf_{\nu\in\F}\|\omega -\nu\|_2\leq \kappa \max_{x\in\X}\inf_{\nu_x\in\R(x)}\|\omega -\nu_x\|_2,  \omega\in\mathbb{R}^{|\X|}.
    \end{align*}
\end{proposition}
\begin{proof}
    Refer to Proposition 1 from \cite{nedich2014lyapunov}.
\end{proof}

    Following Theorem 8 from \cite{nedich2014lyapunov}, we have
    \begin{align*}
        \rho_x =  \left(1 - \frac{\theta^2}{{4|\X^o|\mathrm{diam}(\G^{\mathrm{com}}_{\mathrm{Blue}})}(\kappa+1)^2}\right), \quad c_x = \frac{1}{2}\sqrt{|\Y|}\sum_{x\in\X^o}\|\hat\nu^{\tau=0}_x - \xi^{\tau=1}_x\|_2^2 \\
        \rho_y =  \left(1 - \frac{\theta^2}{{4|\Y^o|\mathrm{diam}(\G^{\mathrm{com}}_{\mathrm{Red}})}(\kappa+1)^2}\right), \quad c_y = \frac{1}{2}\sqrt{|\X|}\sum_{x\in\Y^o}\|\hat\mu^{\tau=0}_y - \xi^{\tau=1}_y\|_2^2,
    \end{align*}
    where $\kappa$ is the regularity constant as in Proposition \ref{prop:set-regularity} and {$\theta$ follows from~\eqref{eq:theta-def}}.
\end{proof}

\section{Existence of Infinite-Population Game Value}
\label{appdx-sec:game-value}

\begin{restatable}{theorem}{EGV}
    \label{thm:game-value}
    Suppose the reward function $r_t$ is concave-convex in its mean-field arguments and the dynamics are independent, i.e., for all $(\mu_t, \nu_t) \in \P(\X) \times \P(\Y)$ , 
    \begin{equation}
        \label{eqn:ind-dynamics}
        f_t(x_{t+1} \vert x_t, u_t, \mu_t, \nu_t) = \bar{f}_t(x_{t+1} \vert x_t, u_t), \quad
        g_t(y_{t+1}\vert y_{t}, v_{t}, \mu_t, \nu_t) = \bar{g}_t(y_{t+1}\vert y_{t}, v_{t}).
    \end{equation}
    Then, 
    the Nash value exists. 
\end{restatable}

To examine the game value, we first formulate an equivalent max-min optimization leveraging a reachability analysis. 
At the infinite-population limit, the mean-field dynamics is deterministic, 
and thus selecting the local policies $\pi_t$ and $\sigma_t$ at time $t$ is equivalent to selecting the desirable MFs at the next time step. 
Consequently, we examine the set of MFs that can be reached from the current MFs.

\begin{definition}
    \label{def:reachable-set}
    The Blue reachable set, starting from the MFs $\mu_t$ and $\nu_t$, is defined as the set comprising all the possible next Blue team MFs  $\mu_{t+1}$ that can be achieved by employing a local policy $\pi_t \in \Pi_t$. Formally, 
    \begin{equation}
        \label{eqn:blue-reachable-set-def}
        \RMSet{t}{} = \left\{\mu_{t+1} \;\vert\; \exists \pi_t \in \Pi_t \st \mu_{t+1} = \mu_t F_t(\mu_t,\nu_t, \pi_t)\right\}.
    \end{equation}
    Similarly, the Red reachable set is defined as 
    \begin{equation}
        \label{eqn:red-reachable-set-def}
         \RNSet{t}{} = \left\{\nu_{t+1} \; \vert \; \exists \sigma_t \in \Sigma_t \st \nu_{t+1} = \nu_t G^{}_t(\mu_t,\nu_t, \sigma_t)\right\}.
    \end{equation}
\end{definition}

We will regard the reachable sets as set-valued functions (correspondences)~\citep{freeman2008robust}. 
In this case, we write $\R_{\mu,t} : \P(\X) \times \P(\Y) \rightsquigarrow \P(\X)$,
and similarly $\R_{\nu,t} : \P(\X) \times \P(\Y) \rightsquigarrow \P(\Y)$.

We can now change the optimization domains in~\eqref{eqn:lower-value-inf} from the policy spaces to the corresponding reachable sets.
\begin{equation}
    \label{eqn:lower-j-ab-t-rset}
    \begin{aligned}
        &\lowervalue_{\infty,T}^{*}(\mu_T, \nu_T) = r_T(\mu_T, \nu_T)\\
        &\lowervalue_{\infty,t}^{*}(\mu_t, \nu_t) = r_t(\mu_t, \nu_t) + \hspace{-0.15in}
        \max_{\mu_{t+1} \in \mathcal{R}_{\mu,t}(\mu_t,\nu_t)}~ \min_{\nu_{t+1} \in \mathcal{R}_{\nu,t}(\mu_t,\nu_t)}
        \lowervalue_{\infty,t+1}^{*}(\mu_{t+1}, \nu_{t+1}), \quad t= 0,\ldots, T-1.
    \end{aligned}
\end{equation}

To examine the optimal value in the infinite-population limit, we first examine the convexity of the reachability correspondences. 
\begin{definition}[\citep{kuroiwa1996convexity}]
    Let $\X$, $\Y$ and $\Z$ be convex sets. 
    The correspondence $\Gamma: \X \times \Y \rightsquigarrow \Z$ is convex with respect to $\X$ if, for all $x_1, x_2 \in \X$, $y \in \Y$,  $z_1 \in \Gamma(x_1, y)$, $z_2 \in \Gamma(x_2, y)$, and $\lambda \in (0,1)$, there exists $z \in \Gamma(\lambda x_1 + (1-\lambda) x_2, y)$ such that 
    \begin{equation*}
        z = \lambda z_1 + (1-\lambda) z_2.
    \end{equation*}
\end{definition}

\begin{remark} \label{rmk:corespondence-convexity}
    The above definition of convexity is equivalent to  the following set inclusion
    \begin{equation}
        \label{eqn:convex-correspondence-def}
        \Gamma(\lambda x_1 + (1-\lambda) x_2, y) \supseteq \lambda \Gamma(x_1,y) + (1-\lambda) \Gamma(x_2, y),
    \end{equation}
    for all $\lambda \in (0,1)$, $x_1,x_2 \in \X$ and $y \in \Y$,
    where the Minkowski sum is defined as     
    \[
    \lambda \Gamma(x_1,y) + (1-\lambda) \Gamma(x_2, y) = \{z| z=\lambda z_1 + (1-\lambda) z_2 \, \text{ where } \, z_1 \in \Gamma(x_1, y), z_2 \in \Gamma(x_2, y)\}.
    \]
\end{remark}

\begin{definition}
    The correspondence $\Gamma: \X \times \Y \rightsquigarrow \Z$ is constant with respect to $\Y$ if
    \begin{equation}
        \label{eqn:constant-correspondence-def}
        \Gamma(x, y_1) = \Gamma(x, y_2) \quad \forall x \in \X, ~\forall y_1, y_2 \in \Y.
    \end{equation}
    
\end{definition}

We have the following concave-convex result for a max-min marginal function defined in~\eqref{eqn:max-min-marginal-def}. 

\begin{equation}
        \label{eqn:max-min-marginal-def}
        f(x,y) = \max_{p\in \Gamma(x,y)} \min_{q\in \Theta(x,y)} g(p,q)
\end{equation}

\begin{restatable}{lemma}{lcc}
    \label{lmm:concave-convex}
    Consider two compact-valued correspondences $\Gamma: \X \times \Y \rightsquigarrow  \X$ and $\Theta: \X \times \Y \rightsquigarrow \Y$. 
    Let $\Gamma$ be convex with respect to $\X$ and constant with respect to $\Y$, and let $\Theta$ be constant with respect to $\X$ and convex with respect to $\Y$.
    Let $g: \X \times \Y \to \mathbb{R}$ be concave-convex. 
    Then, the max-min marginal function $f(x, y) = \max_{p\in \Gamma(x,y)} \min_{q \in \Theta(x,y)} g(p,q)$ 
    is also concave-convex.
\end{restatable}

\begin{proof}
    Fix $y \in \Y$, and consider $x_1, x_2 \in \X$ and $\lambda \in (0,1)$. Denote $x = \lambda x_1 + (1-\lambda) x_2$. It follows that
    \begin{align*}
        f(\lambda x_1 + (1-\lambda)x_2, y ) &= \max_{p \in \Gamma(\lambda x_1 + (1-\lambda)x_2, y)} \; \min_{q\in \Theta(\lambda x_1 + (1-\lambda)x_2, y)} g(p,q)\\
        &\stackrel{\text{(i)}}{\geq} 
        \max_{p \in \lambda \Gamma(x_1, y) + (1-\lambda) \Gamma(x_2, y)} ~~ \min_{q\in \Theta(x, y)} g(p, q) \\
        &= \max_{\substack{p_1 \in \Gamma(x_1,y)\\ p_2 \in \Gamma(x_2,y)}} \min_{q\in \Theta(x, y)} g(\lambda p_1 + (1-\lambda)p_2, q)\\
        &\stackrel{\text{(ii)}}{\geq}\max_{\substack{p_1 \in \Gamma(x_1,y)\\ p_2 \in \Gamma(x_2,y)}} \min_{q\in \Theta(x,y)} (\lambda g(p_1,q) + (1-\lambda)g(p_2, q)) \\
        &\stackrel{\text{(iii)}}{\geq} \max_{\substack{p_1 \in \Gamma(x_1,y)\\ p_2 \in \Gamma(x_2,y)}}  \Big( \underbrace{\min_{q_1\in \Theta(x,y)}\lambda g(p_1,q_1)}_{A(p_1)} + \underbrace{\min_{q_2\in \Theta(x,y)}(1-\lambda)g(p_2, q_2)}_{B(p_2)} \Big) \\
        & = \lambda \max_{p_1 \in \Gamma(x_1,y)} \min_{q_1 \in \Theta(x,y)} g(p_1, q_1) + (1-\lambda) \max_{p_2 \in \Gamma(x_2,y)} \min_{q_2 \in \Theta(x,y)} g(p_2, q_2)\\
        & \stackrel{\text{(iv)}}{=} \lambda \max_{p_1 \in \Gamma(x_1,y)} \min_{q_1 \in \Theta(x_1,y)} g(p_1, q_1) + (1-\lambda) \max_{p_2 \in \Gamma(x_2,y)} \min_{q_2 \in \Theta(x_2,y)} g(p_2, q_2)\\
        &= \lambda f(x_1, y) + (1-\lambda) f(x_2, y),
    \end{align*}
    where inequality (i) holds from restricting the maximization domain using the convexity of $\Gamma$ in~\eqref{eqn:convex-correspondence-def}; 
    inequality (ii) holds from the assumption that $g$ is concave with respect to its $p$-argument; 
    inequality (iii) is the result of distributing the minimization;
    equality (iv) is due to $\Theta$ being constant with respect to $\X$.
    The above result implies the concavity of $f$ with respect to its $x$-argument.
    
    Fix $x \in \X$, and let $y_1, y_2 \in \Y$ and $\lambda \in (0,1)$. Denote $y = \lambda y_1 + (1-\lambda) y_2$. Then,
    \begin{align*}
        f(x, \lambda y_1 + (1-\lambda) y_2) &= \max_{p \in \Gamma(x,\lambda y_1 + (1-\lambda) y_2)} \; \min_{q\in \Theta(x,\lambda y_1 + (1-\lambda) y_2)} g(p,q)\\
        &\leq \max_{p\in \Gamma(x,y)} \min_{\substack{q_1\in \Theta(x,y_1) \\ q_2 \in \Theta(x, y_2)}} g(p, \lambda q_1 + (1-\lambda) q_2)\\
        &\leq \max_{p\in \Gamma(x,y)} \min_{\substack{q_1\in \Theta(x, y_1) \\ q_2 \in \Theta(x, y_2)}}(\lambda g(p,q_1) + (1-\lambda)g(p, q_2)) \\
        & \leq \lambda \max_{p_1 \in \Gamma(x, y)} \min_{q_1 \in \Theta(x,y_1)} g(p_1, q_1) + (1-\lambda) \max_{p_2 \in \Gamma(x, y)} \min_{q_2 \in \Theta(x,y_2)} g(p_2, q_2)\\
        & = \lambda \max_{p_1 \in \Gamma(x, y_1)} \min_{q_1 \in \Theta(x,y_1)} g(p_1, q_1) + (1-\lambda) \max_{p_2 \in \Gamma(x, y_2)} \min_{q_2 \in \Theta(x,y_2)} g(p_2, q_2)\\
        &= \lambda f(x, y_1) + (1-\lambda) f(x, y_2),
    \end{align*}
    which implies that $f$ is convex with respect to its $y$-argument.
\end{proof}

The following lemma shows the convexity of the reachability correspondences $\R_{\mu,t}$ and $\R_{\nu,t}$ under independent dynamics.

\begin{restatable}{lemma}{lrc}
    \label{lmm:rset-indep-dynamics}
    Under the independent dynamics in~\eqref{eqn:ind-dynamics}, the reachability correspondences satisfy
    \begin{subequations}
        \begin{align}
            \R_{\mu,t}(\mu_t, \nu_t) &= \R_{\mu,t}(\mu_t, \nu'_t),\quad \forall \mu_t \in \P(\X), ~~\forall \nu_t, \nu'_t \in \P(\Y), \label{eqn:constant-blue-reachability}
            \\
            \R_{\nu,t}(\mu_t, \nu_t) &= \R_{\nu,t}(\mu'_t, \nu_t),\quad \forall \mu_t, \mu'_t \in \P(\X), ~~ \forall \nu_t \in \P(\Y).
            \label{eqn:constant-red-reachability}
        \end{align}
    \end{subequations}
    Furthermore, for all $\lambda \in [0,1]$, 
    \begin{subequations}
        \begin{align}
            \R_{\mu,t}\big( \lambda \mu_t + (1-\lambda) \mu'_t, \nu_t \big) &= \lambda \R_{\mu,t}\big( \mu_t, \nu_t \big) + (1-\lambda) \R_{\mu,t}\big( \mu'_t, \nu_t \big),
            \quad \forall \mu_t, \mu'_t \in \P(\X), ~~ \forall \nu_t \in \P(\Y), \label{eqn:concave-blue-reachability}
            \\
            \R_{\nu,t}\big(\mu_t, \lambda \nu_t + (1-\lambda) \nu'_t \big) &= \lambda \R_{\nu,t}\big(\mu_t, \nu_t \big) + (1-\lambda) \R_{\nu,t}\big(\mu_t, \nu'_t \big),
            \quad \forall \mu_t \in \P(\X), ~~\forall \nu_t, \nu'_t \in \P(\Y). \label{eqn:concave-red-reachability}
        \end{align}
    \end{subequations}
\end{restatable}

\begin{proof}
    Due to symmetry, we only prove the results for the Blue reachability correspondence.
    The results for the Red team can be obtained through a similar argument. 
    
    Consider an arbitrary pair of distributions $\mu_t \in \P(\X)$ and $\nu_t \in \P(\Y)$. 
    If $\mu_{t+1} \in \R_{\mu,t} (\mu_t, \nu_t)$ under the independent dynamics in~\eqref{eqn:ind-dynamics},  there exists a local policy $\pi_t \in \Pi_t$, such that
    \begin{equation*}
        \mu_{t+1}(x') = \sum_{x \in \X} \Big[\sum_{u \in \U} \bar{f}_t(x'|x,u)\pi_t(u|x)\Big] \mu_t(x),
    \end{equation*}
    which is independent of $\nu_t$. 
    Consequently, 
    \begin{equation*}
        \R_{\mu,t} (\mu_t, \nu_t) = \R_{\mu,t} (\mu_t, \nu'_t), \quad \forall~ \mu_t \in \P(\X) \text{ and } \forall~\nu_t, \nu'_t \in \P(\Y).
    \end{equation*}
    
    Next, we prove the property in~\eqref{eqn:concave-blue-reachability}.
    For the $\subseteq$ direction, consider two Blue mean fields $\mu_t, \mu'_t \in \P(\X)$ and a Red mean field $\nu_t \in \P(\Y)$. 
    Let $\bar{\mu}_{t+1} \in \R_{\mu,t}(\lambda \mu_t + (1-\lambda)\mu'_t, \nu_t)$. 
    From the definition of the reachable set, there exists a local policy $\pi_t \in \Pi_t$, such that
    \begin{align*}
        \bar{\mu}_{t+1}(x') &= \sum_{x \in \X} \Big[\sum_{u \in \U} \bar{f}_t(x'|x,u) \pi_t(u|x) \Big] (\lambda \mu_t(x) + (1-\lambda)\mu_t'(x))  \\
        &=\lambda \underbrace{\sum_{x \in \X} \Big[\sum_{u \in \U} \bar{f}_t(x'|x,u) \pi_t(u|x) \Big]  \mu_t(x) }_{\mu_{t+1}(x')}
        +(1-\lambda) \underbrace{\sum_{x \in \X} \Big[\sum_{u \in \U} \bar{f}_t(x'|x,u) \pi_t(u|x) \Big]  \mu_t'(x)}_{\mu'_{t+1}(x')},
    \end{align*}
    where $\mu_{t+1}\in \R_{\mu,t} (\mu_t,\nu_t)$ and $\mu'_{t+1}\in \R_{\mu,t} (\mu'_t,\nu_t)$.
    Hence, we have 
    \[\R_{\mu,t}(\lambda \mu_t + (1-\lambda)\mu'_t, \nu_t) \subseteq \lambda \R_{\mu,t}\big( \mu_t, \nu_t \big) + (1-\lambda) \R_{\mu,t}\big( \mu'_t, \nu_t \big).\]
    
    \vspace{+5pt}
    
    Let now $\bar{\mu}_{t+1} = \lambda \mu_{t+1} + (1-\lambda) \mu_{t+1}'$, where $\mu_{t+1} \in \R_{\mu,t}(\mu_t, \nu_t)$ and $\mu_{t+1}' \in \R_{\mu,t}(\mu_t', \nu_t)$. 
    Further, define $\bar{\mu}_t = \lambda \mu_t + (1-\lambda) \mu'_t$.
    From the definition of the reachable set, there exist local policies $\pi_t$ and $\pi'_t$ such that 
    \begin{align*}
        \bar{\mu}_{t+1} &= \lambda \sum_{x \in \X} \Big[\sum_{u \in \U} \bar{f}_t(x'|x,u) \pi_t(u|x) \Big]  \mu_t(x)
        +(1-\lambda) \sum_{x \in \X} \sum_{u \in \U} \bar{f}_t(x'|x,u) \pi'_t(u|x) \Big]  \mu_t'(x) \\
        &= \sum_{x \in \X} \sum_{u \in \U} \bar{f}_t(x'|x,u) \Big[ \lambda \pi_t(u|x) \mu_t(x) + (1-\lambda) \pi'_t(u|x) \mu'_t(x)\Big] \\
        &= \sum_{x \in \X} \sum_{u \in \U} \mathds{1}_{\bar{\mu}_t(x) >0} \; \bar{f}_t(x'|x,u) \Big[ \lambda \pi_t(u|x) \mu_t(x) + (1-\lambda) \pi'_t(u|x) \mu'_t(x)\Big] \\
        & + \sum_{x \in \X} \sum_{u \in \U} \mathds{1}_{\bar{\mu}_t(x) =0}  \; \bar{f}_t(x'|x,u) \underbrace{\Big[ \lambda \pi_t(u|x) \mu_t(x) + (1-\lambda) \pi'_t(u|x) \mu'_t(x)\Big]}_{=0 \, \footnotemark} \\
        &= \sum_{x \in \X} \sum_{u \in \U} \mathds{1}_{\bar{\mu}_t(x) >0} \; \bar{f}_t(x'|x,u) \frac{ \lambda \pi_t(u|x) \mu_t(x) + (1-\lambda) \pi'_t(u|x) \mu'_t(x)}{\bar{\mu}_t(x)} \bar{\mu}_t(x) \\
        &+\sum_{x \in \X} \sum_{u \in \U} \mathds{1}_{\bar{\mu}_t(x) =0} \; \bar{f}_t(x'|x,u) \underbrace{\Big[ \lambda \pi_t(u|x) + (1-\lambda) \pi'_t(u|x) \Big] \bar{\mu}_{t}(x)}_{=0} \\
        &= \sum_{x \in \X} \left[\sum_{u \in \U} \bar{f}_t(x'|x,u) \bar{\pi}_t(u|x)\right] \bar{\mu}_t(x),
    \end{align*}
    \footnotetext{Since the condition $\bar{\mu}_t(x)=0$ together with $\lambda \in (0,1)$ implies that $\mu_t(x) =0$ and $\mu_t'(x) = 0$.}
    
    where the ``averaged" local policy $\bar{\pi}_t$ is given by
    \begin{equation*}
        \bar{\pi}_t (u|x) = \mathds{1}_{\bar{\mu}_t(x) > 0}  \frac{\lambda\pi_t(u|x)\mu_t(x) + (1-\lambda) \pi'_t(u|x) \mu'_t(x)}{\lambda \mu_t(x) + (1-\lambda) \mu'_t(x)} + \mathds{1}_{\bar{\mu}_t(x) = 0} \Big(\lambda\pi_t(u|x) + (1-\lambda) \pi'_t(u|x) \Big).
    \end{equation*}
    Consequently, we have
    \begin{equation*}
        \bar{\mu}_{t+1} \in \R_{\mu,t} (\lambda \mu_t + (1-\lambda) \mu_{t}', \nu_t).
    \end{equation*}
     Hence,  
     \[\R_{\mu,t}(\lambda \mu_t + (1-\lambda)\mu'_t, \nu_t) \supseteq \lambda \R_{\mu,t}\big( \mu_t, \nu_t \big) + (1-\lambda) \R_{\mu,t}\big( \mu'_t, \nu_t \big).\]

\end{proof}


\begin{proof}
We can now prove Theorem~\ref{thm:game-value}. The game value at time $T$ is directly the mean-field reward, which is assumed to be concave-convex.
Hence, one can provide an inductive proof leveraging Lemma~\ref{lmm:concave-convex} and Lemma~\ref{lmm:rset-indep-dynamics} to show that both the lower and upper game values are concave-convex at each time step. Note that the lower game value in~\eqref{eqn:lower-j-ab-t-rset} has the form of a max-min marginal function. 
Finally, as the reachable sets are compact and convex, one can apply the minimax theorem~\citep{owen2013game}, which guarantees that the max-min optimal value is the same as the min-max optimal value, and thus ensures the existence of the game value.
\end{proof}

\section{Results on Minimally-Informed Critics}
\label{appdx-sec:MI-critic}

\begin{proposition}\label{thm:minimal-critic}
Let $\mu^{N_1}_t$ and $\nu^{N_2}_t$ denote the EDs induced by identical Blue and Red team policies.
Then, the team reward structure admits a critic that depends only on $\mu^{N_1}_t$ and $\nu^{N_2}_t$.
Specifically, for each Blue team agent $i \in [N_1]$, the individual critic satisfies
\(
V_{i, t}^{N_1, \phi_t}(x_{i, t}, \mu^{N_1}_t, \nu^{N_2}_t)
= V^{N_1, \phi_t}_{{\rm Blue}, t}(\mu^{N_1}_t, \nu^{N_2}_t),
\)
where $V^{N_1, \phi_t}_{{\rm Blue}, t}$ is the team-level value function.
\end{proposition}

We present the proof for a single team (Blue). 
The proof can be directly extended to the opponent team (Red). 
Let the critic value function at time $t$ for each Blue team agent $i$, following policy $\phi_{t}$ be given by $V_{i, t}^{N_1, \phi_{t}}(x^{i}_t, \mu_t, \nu_t)$ where $\zeta^i$ represents its parameterization, namely,
\begin{align*}
    V_{i, \zeta_i, t}^{N_1, (\phi, \psi)}(x^{i}_t, \mu_t, \nu_t  ) = \expct{\sum_{t'=t}^\infty \gamma^{t'-t}r^{t'}(\mu_{t'}, \nu_{t'})}{\phi,\psi}.
\end{align*}
Note that RHS of the above expression is the closed form representation for the value function and is independent of $x^{i}_t$. 
Thus, for all $i \in \{1, 2,\ldots, N_1\}$ we can define
\begin{align*}
    V^{N_1, (\phi, \psi)}_{\zeta_{\text{Blue}}, t}(\mu_t, \nu_t)  \triangleq V_{i,\zeta_i, t}^{N_1, (\phi, \psi)}(x^{i}_t, \mu_t, \nu_t  ) , 
\end{align*}
where $  V^{N_1, (\phi, \psi)}_{\zeta_{\text{Blue}}, t}(\mu_t, \nu_t)  \triangleq \expct{\sum_{t'=t}^\infty \gamma^{t'-t}r^{t'}(\mu_{t'}, \nu_{t'})  | \mu_{t}, \nu_{t}}{\phi, \psi}$ and the result follows.

\section{\edit{Proof of Theorem~\ref{thm:finite-population-training}}}\label{proof:finite-population-training}
\begin{theorem}[\cite{guan2024zero}]
    \label{thm:performance-guarantees}
    The optimal identical Blue team policy $\phi_\infty^*\in\Phi$ obtained from the equivalent zero-sum coordinator game is   $\epsilon$-optimal Blue team policy. 
    Formally, for all joint states $\bfx^{N_1}$ and $\bfy^{N_2}$,
    \begin{align}
        \label{eqn:performance-bounds}
        \min_{\psi^{N_2} \in \Psi^{N_2}} 
        J^{N,\phi_\infty^*,\psi^{N_2}} (\bfx^{N_1}, \bfy^{N_2}) \geq \underline{J}^{N*}(\bfx^{N_1}, \bfy^{N_2})-\Residue, ~\textrm{where $\underline{N} = \min \{N_1, N_2\}$}.
    \end{align}
\end{theorem}


\maintheoremtwo*
\begin{proof}[Sketch of Proof]
The proof proceeds in two steps: (i) we first show that the optimal identical Blue policy obtained from the \emph{finite-$N$ system} is $\epsilon$-optimal against identical Red exploitation in the \emph{infinite-population limit}; (ii) we then invoke the approximation result in~\cite{guan2024zerosumgameslargepopulationteams} to show that this identical Blue policy remains $\epsilon$-optimal against non-identical Red policies in the original finite-population game.
\end{proof}

\begin{proof}
We begin by noting two key lemmas that will be used to prove the theorem.
The first lemma follows directly from analysis and convex optimization.
\begin{lemma}\label{lmm-f-g-max-min}
    Let $f(x, y)$ and $g(x, y)$ be two functions defined for all $x, y$ in some compact convex set. Suppose that for all $x, y$, the functions are $\epsilon$-close: $|f(x, y) - g(x, y)| < \epsilon$, which is equivalent to,
\begin{equation*}
    f(x, y) - \epsilon < g(x, y) < f(x, y) + \epsilon \label{eq:eps_bound}
\end{equation*}
Consider the max-min values of both functions $f^* = \max_x \min_y f(x, y) $ and $g^* = \max_x \min_y g(x, y)$.
Let $x_f^*$ be the maximizer that attains $f^*$. Then, the performance of the optimal policy $x_f^*$ for $f$, when evaluated under $g$, is $2\epsilon$-optimal. Specifically, $ g^* - 2\epsilon \le \min_y g(x_f^*, y) \le g^*$.
\end{lemma}

\begin{lemma}\label{lmm-intermediate-B}
   For all identical team policies $(\phi, \psi)$  we have the following inequality for all joint states $\bfx^{N_1} \in \X^{N_1}$ and $\bfy^{N_2} \in \Y^{N_2}$, 
\begin{align}
  |J^{{N},\phi,\psi} (\bfx^{N_1}, \bfy^{N_2}) - {J}_\infty^{\phi, \psi} (\mu^{N_1}, \nu^{N_2})|  \leq \Residue,
\end{align}
where $\mu^{N_1}=\empMu{\bfx^{N_1}}$, $\nu^{N_2} = \empNu{\bfy^{N_2}}$. 
\end{lemma}
\begin{proof}
    We can prove this by mathematical induction. We show that for any pair of identical policies $(\phi, \psi)$  we have the following inequality for all joint states $\bfx^{N_1} \in \X^{N_1}$ and $\bfy^{N_2} \in \Y^{N_2}$, 
\begin{align}\label{eqn:step1-inequality-e2m-edit}
  J^{{N},\phi,\psi} (\bfx^{N_1}, \bfy^{N_2}) \geq {J}_\infty^{\phi, \psi} (\mu^{N_1}, \nu^{N_2})  -\Residue,
\end{align}
where $\mu^{N_1}=\empMu{\bfx^{N_1}}$, $\nu^{N_2} = \empNu{\bfy^{N_2}}$. 
The proof is constructed based on induction.

\noindent\textit{Base case:}
    At the terminal timestep $T$, since there is no decision to be made, both value functions are equal to the terminal reward and are thus the same.
    Formally, for all $\bfx^{N_1}_T \in \bfX^{N_1}$ and $\bfy^{N_2}_T \in \bfY^{N_2}$, 
    \begin{align*}
        {J}_{T}^{{{N},\phi,\psi}}(\bfx^{N_1}_T , \bfy^{N_2}_T) = {J}_{\infty,T}^{\phi, \psi}(\mu^{N_1}_T , \nu^{N_2}_T).
    \end{align*} 
    For simplicity, we do not emphasize the correspondence between the joint states and the EDs for the rest of the proof, as it is clear from the context.

    \noindent\textit{Inductive hypothesis: }
    Assume that at $t+1$, the following holds for all joint states  $\bfx^{N_1}_{t+1} \in \bfX^{N_1}$ and $\bfy^{N_2}_{t+1} \in \bfY^{N_2}$:
    \begin{equation}
        J_{t+1}^{{{N},\phi,\psi}}(\bfx^{N_1}_{t+1}, \bfy^{N_2}_{t+1}) \geq 
        {J}_{\infty,t+1}^{\phi, \psi}(\mu^{N_1}_{t+1} , \nu^{N_2}_{t+1}) - \Residue.
    \end{equation}
    \textit{Induction: }
    At timestep $t$, consider an arbitrary pair of joint states $\bfx^{N_1}_{t} \in \bfX^{N_1}$ and $\bfy^{N_2}_{t} \in \bfY^{N_2}$, and their corresponding EDs $\mu^{N_1}_t$ and $\nu^{N_2}_t$. 
    Then, for all joint states $\bfx^{N_1}_t \in \bfX^{N_1}$ and $\bfy^{N_2}_t \in \bfY^{N_2}$, we have
    \begin{align*}
        & J_{t}^{{N},\phi,\psi}(\bfx^{N_1}_t, \bfy^{N_2}_t)\\
        &=r_t(\mu^{N_1}_t, \nu^{N_2}_t) + \mathbb{E}_{\phi, \psi}\Big[J_{t+1}^{{N},\phi,\psi}(\bfX^{N_1}_{t+1}, \bfY^{N_2}_{t+1}) \Big]\\
        &\stackrel{\text{(i)}}{\geq}  r_t(\mu^{N_1}_t, \nu^{N_2}_t) + \mathbb{E}_{\phi, \psi}\Big[J_{\infty, t+1}^{\phi, \psi}(\M^{N_1}_{t+1}, \N^{N_2}_{t+1}) \Big] - \Residue
        \\
        &=r_t(\mu^{N_1}_t, \nu^{N_2}_t)  - \Residue + \mathbb{E}_{\phi, \psi}\Big[J_{\infty, t+1}^{\phi, \psi}(\M^{N_1}_{t+1}, \N^{N_2}_{t+1})
        \\&\qquad \qquad \qquad \qquad \qquad \qquad \qquad \qquad- J_{\infty, t+1}^{\phi, \psi}(\mu^{\phi}_{t+1}, \nu^{\psi}_{t+1})  + J_{\infty, t+1}^{\phi, \psi}(\mu^{\phi}_{t+1}, \nu^{\psi}_{t+1}) \Big] \nonumber\\
        &\stackrel{\text{(ii)}}{\geq}  r_t(\mu^{N_1}_t, \nu^{N_2}_t) +
        J_{\infty, t+1}^{\phi, \psi}(\mu^{\phi}_{t+1}, \nu^{\psi}_{t+1}) -\Residue- L_{J, t+1}  \underbrace{\mathbb{E}_{\phi}\Big[\dtv{\M_{t+1}^{N_1},\mu^{\phi}_{t+1}}\Big]}_{=\mathcal{O}(\frac{1}{\sqrt{N_1}})}  
        \\
        &  \qquad \qquad \qquad \qquad \qquad \qquad \qquad \qquad\qquad \qquad -\underbrace{L_{J, t+1} \mathbb{E}_{ \psi}\Big[\dtv{ \N^{N_2}_{t+1},\nu^{\psi}_{t+1}}\Big]}_{=\mathcal{O}(\frac{1}{\sqrt{N_2}})}  
        \nonumber \\
        & \stackrel{\text{(iii)}}{\geq}  r_t(\mu^{N_1}_t, \nu^{N_2}_t) +
         J_{\infty, t+1}^{\phi, \psi}(\mu^{\phi}_{t+1}, \nu^{\psi}_{t+1})-\Residue \\
        &\stackrel{\text{(iv)}}{=}  {J}_{\infty,t}^{\phi, \psi}(\mu^{N_1}_{t} , \nu^{N_2}_{t}) -\Residue.
    \end{align*}
    For inequality (i), we used the inductive hypothesis; 
    for inequality (ii), we utilized the Lipschitz continuity of the coordinator value function where $\mu^\phi_{t+1} = \mu^{N_1}_{t}F(\mu^{N_1}_{t}, \nu^{N_2}_{t}, \phi)$ and $\nu^\psi_{t+1} = \nu^{N_2}_{t}G(\mu^{N_1}_{t}, \nu^{N_2}_{t}, \psi)$; inequality (iii) follows from~\cite{guan2024zerosumgameslargepopulationteams} and equality (iv) follows from the definition of $J_{\infty, t}^{\phi, \psi}(\mu^{N_1}_t, \nu^{N_2}_t )$ which completes the induction.
 We have that, for all joint states $\bfx^{N_1} \in \bfX^{N_1}$ and $\bfy^{N_2} \in \bfY^{N_2}$,
    \begin{align*}
        J^{{N},\phi,\psi} (\bfx^{N_1}, \bfy^{N_2}) &= 
        J_0^{{N},\phi,\psi} (\bfx^{N_1}, \bfy^{N_2}) \geq {J}_{\infty}^{\phi, \psi} (\mu^{N_1}, \nu^{N_2})  -\Residue.
    \end{align*}
Similarly, using induction we can show that for $(\phi, \psi)$  we have the following inequality for all joint states $\bfx^{N_1} \in \X^{N_1}$ and $\bfy^{N_2} \in \Y^{N_2}$, 
\begin{align}
  J^{{N},\phi,\psi} (\bfx^{N_1}, \bfy^{N_2}) \leq {J}_\infty^{\phi, \psi} (\mu^{N_1}, \nu^{N_2})  +\Residue,
\end{align}
where $\mu^{N_1}=\empMu{\bfx^{N_1}}$, $\nu^{N_2} = \empNu{\bfy^{N_2}}$ and the result follows. 
\end{proof}

We begin by presenting the key steps of the proof of Theorem~\ref{thm:finite-population-training}:

\textbf{Step 1:} We have the following definition of the lower game value for the 
finite-population ZS-MFTG:
\begin{equation}\label{eq:max-min-general-case-edit}
    \underline{J}^{N*}(\bfx^{N_1} , \bfy^{N_2})  = \max_{\mathbf{\phi}^{N_1} \in \Phi^{N_1}} \min_{\mathbf{\psi}^{N_2} \in \Psi^{N_2}} J^{N, \phi^{N_1}, \psi^{N_2}}(\bfx^{N_1}, \bfy^{N_2}).
\end{equation}
If we restrict ourselves to the set of identical team policies $\Phi \subseteq \Phi^{N_1}$ it follows immediately that
\begin{align}
\underline{J}^{N*}(\bfx^{N_1} , \bfy^{N_2})
&\geq \max_{\mathbf{\phi}^{N_1} \in \Phi} \min_{\mathbf{\psi}^{N_2} \in \Psi^{N_2}} J^{N, \phi^{N_1}, \psi^{N_2}}(\bfx^{N_1}, \bfy^{N_2})\nonumber\\
&\geq \min_{\psi^{N_2} \in \Psi^{N_2}} 
       J^{N,\tilde\phi^*,\psi^{N_2}} (\bfx^{N_1}, \bfy^{N_2}),
\end{align}
where $\tilde\phi^* $ follows the definition from the statement of the theorem.

\textbf{Step 2:} We next define, $ \tilde\psi^*_\infty = \arg\min_{\psi \in \Psi} ~ J_\infty^{\tilde\phi^*, \psi}(\mu_0, \nu_0)$ which determines the optimal Red team response in the \textit{infinite-population} domain, when the Blue team plays $\tilde\phi^*$ derived from the finite population optimization over identical policies.
For $\tilde\phi^*$  we have the following inequality for all joint states $\bfx^{N_1} \in \X^{N_1}$ and $\bfy^{N_2} \in \Y^{N_2}$, 
\begin{align}\label{eqn:step2}
   \min_{\psi^{{N}_2} \in \Psi^{{N}_2}} 
        J^{{N},\tilde\phi^*,\psi^{N_2}} (\bfx^{N_1}, \bfy^{N_2}) &\geq {J}_\infty^{\tilde\phi^*, \tilde\psi^*_\infty} (\mu^{N_1}, \nu^{N_2})  -\Residue\nonumber\\
         &=\min_{\psi\in\Psi}{J}_\infty^{\tilde\phi^*, \psi} (\mu^{N_1}, \nu^{N_2})-\Residue     
\end{align}
where $\mu^{N_1}=\empMu{\bfx^{N_1}}$, $\nu^{N_2} = \empNu{\bfy^{N_2}}$. The proof of inequality~\eqref{eqn:step2} statement is identical to Step 2 in Appendix G.

\textbf{Step 3:} Now, applying Lemma~\ref{lmm-f-g-max-min} to $f = J^{{N},\phi,\psi} (\bfx^{N_1}, \bfy^{N_2})$ and $g = {J}_\infty^{\phi, \psi} (\mu^{N_1}, \nu^{N_2})$  where $\epsilon=\Residue$  (by Lemma~\ref{lmm-intermediate-B}) and $x_f^* = \tilde\phi^*$, it follows that
\begin{align}\label{eqn:step3}
    \min_{\psi\in\Psi}{J}_\infty^{\tilde\phi^*, \psi} (\mu^{N_1}, \nu^{N_2})&\geq \max_{\phi\in\Phi}\min_{\psi\in\Psi}{J}_\infty^{\phi, \psi} (\mu^{N_1}, \nu^{N_2}) - \Residue\nonumber\\
    &\geq \underline{J}^{N*}(\bfx^{N_1} , \bfy^{N_2}) -\Residue,
\end{align}
where the second inequality follows from Lemma 6 in~\cite{guan2024zerosumgameslargepopulationteams}.
Combining~\eqref{eqn:step2} and~\eqref{eqn:step3}, we obtain the desired result.

\end{proof}

\section{Proof of Theorem~\ref{thm:convergence-of-gradients}}\label{proof:convergence-of-gradients}

\maintheoremone*

We first discuss the infinite-population solution in detail.
To reduce the complexity of team policy optimization domains in~\eqref{eq:blue-max-min}, the authors of \cite{guan2024zero} 
{examined} team behaviors under \textit{identical team policies} at the \textit{infinite-population} limit. 
It was shown that the team joint states can be represented using the team population distribution, which 
coincides with the state distribution of a \textit{typical agent} referred to as the mean-fields ($\mu_t$ and $\nu_t$ for the Blue and Red teams, respectively) and the infinite-population game can be transformed to an equivalent zero-sum game between two coordinators.
Furthermore, there is a one-to-one correspondence between infinite-population \textit{coordination} policies $(\alpha, \beta)$ and local \textit{identical} team policies $(\phi, \psi) \in \Phi \times \Psi$.
The performance of $(\phi, \psi)$ in the equivalent zero-sum coordinator game is measured by 
\begin{equation}
  J_{\infty}^{\alpha, \beta}(\mu_0, \nu_0) \equiv J_{\infty}^{\phi, \psi}(\mu_0, \nu_0) = \sum_{t=0}^{T} r_t(\mu_t, \nu_t)  ,
\end{equation}
where $\mu_t$ and $\nu_t$ follow a \textit{deterministic} dynamics~\citep{guan2024zero} similar to the state distribution propagation of a controlled Markov chain.
\cite{guan2024zero} establishes guarantees that identical team policies resulting from the solution of this equivalent zero-sum \textit{coordinator} game are still $\epsilon$-optimal for the original max-min optimization problem in~\eqref{eq:blue-max-min} where {$\epsilon = \mathcal{O}(1/\sqrt{\underline{N}})$ and $ \underline{N} = \min \{N_1, N_2\}$.}

We extend this formulation to potentially \textit{stochastic} infinite-population coordinator policies and the local finite-population policies as follows: a stochastic Blue coordination policy $\alpha \in \mathcal{A}$ induces an identical team policy $\phi \in \Phi$ according to the rule
    \begin{equation}\label{eq:equivalent-policy}
        \phi_{t}(u_t|x_t, \mu^\rho_t, \nu^\rho_t) = \int_{\pi\in\Pi}\pi_t(u_t |x_t)\alpha_t(\pi_t|\mu^\rho_t,\nu^\rho_t)d\pi \quad \forall \mu_t \in \P(\X), \; \nu_t \in \P(\Y),\; x_t \in \X \text{ and } u_t \in \U,
    \end{equation}
where $\pi_t(u_t |x_t)$ corresponds to the identical local policies for all states $x\in \X$ prescribed by the coordinator, i.e., $\pi_t \sim \alpha(\mu^\rho_t, \nu^\rho_t).$
Thus, we can now formally state Theorem~\ref{thm:convergence-of-gradients} using the notion of coordinator policies.
\begin{theorem}
    The approximate policy gradient of the infinite-population Blue (Red) team coordinator policy $\alpha$ ($\beta$) computed using local policies from the finite-population ZS-MFTG via MF-MAPPO ($\hat{J}^{N_1}(\alpha_\theta)$) uniformly tends to the true policy gradient as the population size increases, i.e., $ \|\nabla_{\theta}  J_\infty(\alpha_\theta) - \nabla_{\theta}  \hat J^{N_1}(\alpha_\theta)\|_2 \to 0$ as $(N_1, N_2)\to\infty$.
\end{theorem}

\begin{proof}[Sketch of Proof]
In MF-MAPPO, we directly train the finite-population local policies $(\phi_t, \psi_t)$.
As a result, the gradient of the coordinator policy $\nabla_{\theta} J(\alpha_\theta)$ only has action samples $u^{N_1}_{i,t}, i\in[N_1]$ from the finite-population local policies $\phi_t(\cdot\vert x^{N_1}_{i,t}, \mu^{N_1}_t, \nu^{N_2}_t)$ and not $\pi^{N_1}_t(\cdot \vert x^{N_1}_{i,t})\sim\alpha(\mu^{N_1}_t, \nu^{N_2}_t)$ itself.
Thus, our first step is to construct an approximation of $\pi^{N_1}_t$ using the obtained finite-population action samples $u^{N_1}_{i,t}, i\in[N_1]$, resulting in an approximate finite-population policy gradient expression given by $\nabla_{\theta} \hat J^{N_1}(\alpha_\theta)$.
Coupled with existing mean-field approximation results~\citep{cui2023major, guan2024zero, shao2024reinforcement} for the infinite-population limit $N_1, N_2 \to \infty$ (LLN, etc.) we show convergence.
\end{proof}

\begin{proof}
It is important to note that in MF-MAPPO we directly work with the finite-population local policies $(\phi, \psi)$, unlike prior works that train on the infinite-population coordinator policy (see Section~\ref{sec:intro}).
Consequently, the gradient of the coordinator policy $\nabla_{\theta} J(\alpha_\theta)$ computed using this finite-population approximation (denoted $\nabla_{\theta} \hat J^{N_1}(\alpha_\theta)$) only has action samples $u^{N_1}_{i,t}, i\in[N_1]$ from the finite-population local policies $\phi_t(\cdot\vert x^{N_1}_{i,t}, \mu^{N_1}_t, \nu^{N_2}_t)$ and not direct access to $\pi^{N_1}_t\sim\alpha(\mu^{N_1}_t, \nu^{N_2}_t)$.
Thus, our first step is to construct an approximation of the policy based on the obtained action samples.
We define this approximate policy from the obtained samples as follows,
\begin{equation}
        \label{eqn:apprx-local-policy}
        \hat\pi_t^{N_1} (u|x) = \left\{
        \begin{array}{ll}
             \frac{\sum_{i=1}^{N_1} \indicator{x}{X_{i,t}^{N_1}} \indicator{u}{U_{i,t}^{N_1}}}{N_1 \M^{N_1}_t(x)} \quad 
             &  \text{if } ~~ \M_t^{N_1}(x) >0, \\
             {1}/{|\U|} \quad 
             &  \text{if } ~~ \M_t^{N_1}(x) =0.
        \end{array}
        \right .
\end{equation}

We can similarly define $\hat\sigma_t^{N_2}$ for the Red team.
Using this constructed empirical policy $\hat\pi^{N_1}_t$ from the sampled actions $\u^{N_1}_t$, the policy gradient $\nabla_{\theta} \hat J^{N_1}(\alpha_\theta)$ for the finite-population ZS-MFTG is,
\begin{align}\label{eq:pol-grad-finite}
    \nabla_{\theta} \hat J^{N_1}(\alpha_\theta) = \sum_{t=0}^T \expct{Q_t^{N, \alpha, \beta}(\mu^{N_1}_t, \nu^{N_2}_t, \hat\pi^{N_1}_t, \hat\sigma^{N_2}_t)\nabla_\theta \log\alpha_\theta (\hat\pi^{N_1}_t |\mu^{N_1}_t, \nu^{N_2}_t)}{\u^{N_1}\sim\phi, \v^{N_2}\sim\psi},
\end{align}
where, 
\begin{align}\label{eq:q-func-finite}
    &Q_t^{N, \alpha, \beta}(\mu^{N_1}_t, \nu^{N_2}_t, \hat\pi^{N_1}_t, \hat\sigma^{N_2}_t) \nonumber \\&= \expct{\sum_{\tau = 0}^T r(\mu^{N_1}_\tau, \nu^{N_2}_\tau) \big\vert \mu^{N_1}_{\tau=0} = \mu^{N_1}_t, \nu^{N_2}_{\tau=0} = \nu^{N_2}_t, \pi^{N_1}_{\tau=0} = \hat\pi^{N_1}_t, \sigma^{N_2}_{\tau=0} = \hat\sigma^{N_2}_t}{\u^{N_1}\sim\phi, \v^{N_2}\sim\psi}.
\end{align}
Recall that $\u^{N_1}$ and $\v^{N_2}$ enter~\eqref{eq:pol-grad-finite} and~\eqref{eq:q-func-finite} through the construction of $\hat\pi^{N_1}_t$ (Equation F.31) and $\hat\sigma^{N_2}_t$ respectively.
Furthermore, we have the following expression for the gradient of the infinite-population coordinator game:
\begin{align}\label{eq:pol-grad-infty}
    \nabla_{\theta} J(\alpha_\theta) = \sum_{t=0}^T\expct{Q_t^{\alpha, \beta}(\mu_t, \nu_t, \pi_t, \sigma_t)\nabla_\theta \log\alpha_\theta (\pi_t | \mu_t ,\nu_t)}{\pi\sim\alpha, \sigma\sim\beta},
\end{align}
where, 
\begin{align}\label{eq:q-func-infty}
    Q_t^{\alpha, \beta}(\mu, \nu, \pi, \sigma) = \expct{\sum_{\tau = 0}^T  r(\mu_\tau, \nu_\tau) \big\vert \mu_{\tau=0} = \mu, \nu_{\tau=0} = \nu, \pi_{\tau=0} = \pi, \sigma_{\tau=0} = \sigma}{\pi\sim\alpha, \sigma\sim\beta}.
\end{align}
While~\eqref{eq:q-func-finite} and~\eqref{eq:q-func-infty} utilize the joint coordinator policies $(\alpha, \beta)$, the policy gradients~\eqref{eq:pol-grad-finite} and~\eqref{eq:pol-grad-infty} are taken for each team \textit{individually}.  
Consequently, the policy gradient analysis is done on a per-team basis. 
We focus on the policy gradient with respect to the Blue team but the analysis of the Red team's policy gradient is symmetric and can be proved by an identical approach.

The following lemma relates the expectation in~\eqref{eq:pol-grad-finite} to $(\alpha, \beta)$ using~\eqref{eq:equivalent-policy}. \edit{It shows that a stochastic coordinator policy, under shared randomness, caninduce a closed-loop policy $\phi$ that achieves the same expected performance.}

\begin{lemma}\label{lmm:equiv-expectation}
  \edit{Assume common randomness among all agents in a team when sampling the local policy $\pi$.}
Then, given a function $g : \U \to \mathbb{R}$ of the joint actions $\mathbf{u}^{N_1}_t$ such that $u^{N_1}_{i,t} \sim \phi(\cdot | x^{N_1}_{i,t}, \mu^{N_1}_t, \nu^{N_2}_t)$, or equivalently, $\mathbf{u}^{N_1}_t \sim \phi(\cdot | \mathbf{x}^{N_1}_t, \mu^{N_1}_t, \nu^{N_2}_t)$, it holds that
    \begin{align*}
        \expct{g(\mathbf{u}^{N_1}_t)}{\mathbf{u}^{N_1}_t\sim\phi} = \expct{\expct{g(\mathbf{u}^{N_1}_t)}{\mathbf{u}^{N_1}_t\sim\pi}}{\pi\sim\alpha}.
    \end{align*}
\end{lemma}
\begin{proof}
By the definition of expectation,
\begin{align*}
    \expct{g(\mathbf{u}^{N_1}_t)}{\mathbf{u}^{N_1}_t\sim\phi} = \sum_{\mathbf{u}^{N_1}_t} g(\mathbf{u}^{N_1}_t)\phi(\mathbf{u}^{N_1}_t | \mathbf{x}^{N_1}_t, \mu^{N_1}_t, \nu^{N_2}_t),
\end{align*}
where the sum is taken over all the possible joint actions that can be sampled from $\phi$.
By using the definition of $\phi$ from~\eqref{eq:equivalent-policy},
\begin{align*}
    \sum_{\mathbf{u}^{N_1}_t} g(\mathbf{u}^{N_1}_t)\phi(\mathbf{u}^{N_1}_t | \mathbf{x}^{N_1}_t, \mu^{N_1}_t, \nu^{N_2}_t) &=  \sum_{\mathbf{u}^{N_1}_t} g(\mathbf{u}^{N_1}_t)\underbrace{\int_{\pi\in\Pi}\pi_t(\mathbf{u}^{N_1}_t | \mathbf{x}^{N_1}_t)\alpha_t(\mu^{N_1}_t, \nu^{N_2}_t)d\pi}_{= \expct{\pi_t(\mathbf{u}^{N_1}_t | \mathbf{x}^{N_1}_t)}{\pi\sim\alpha}}\\
    &= \sum_{\mathbf{u}^{N_1}_t} g(\mathbf{u}^{N_1}_t) \expct{\pi_t(\mathbf{u}^{N_1}_t | \mathbf{x}^{N_1}_t)}{\pi\sim\alpha}\\
    &= \sum_{\mathbf{u}^{N_1}_t} \expct{g(\mathbf{u}^{N_1}_t) \pi_t(\mathbf{u}^{N_1}_t | \mathbf{x}^{N_1}_t)}{\pi\sim\alpha}\\
    &= \expct{\sum_{\mathbf{u}^{N_1}_t}g(\mathbf{u}^{N_1}_t)\pi_t(\mathbf{u}^{N_1}_t | \mathbf{x}^{N_1}_t)}{\pi\sim\alpha},
\end{align*}
where the last two expressions follow from linearity (and finiteness) of expectation.
It follows then, 
\begin{align*}
    \expct{g(\mathbf{u}^{N_1}_t)}{\mathbf{u}^{N_1}_t\sim\phi} = \expct{\expct{g(\mathbf{u}^{N_1}_t)}{\mathbf{u}^{N_1}_t\sim\pi}}{\pi\sim\alpha}.
\end{align*}
\end{proof}
\begin{remark}
    This shared-randomness assumption is only required for this theoretical equivalence and is not required by MF-MAPPO itself. 
In practice, MF-MAPPO directly learns the local policy $\phi$ from finite-population interactions and therefore does not require a common random seed to be shared across agents.
\end{remark}

We further assume continuity with respect to the local policies.

\begin{assumption}\label{assumpt:continuous-Q}
   The Q-functions $Q^{\alpha, \beta}$ and $Q^{N, \alpha, \beta}$ from~\eqref{eq:q-func-infty} and~\eqref{eq:q-func-finite} are continuous in inputs $(\pi, \sigma)$ and $(\hat\pi^{N_1}, \hat\sigma^{N_2})$.
\end{assumption}

It is important to note that the action-space $\Pi$ and $\Sigma$ are continuous.
Thus, although Assumption~\ref{assumpt:continuous-Q} seems restrictive, the use of function approximators to train $Q^{\alpha, \beta}$ and $Q^{N, \alpha, \beta}$ like neural networks (endowed with continuous activation functions), ensures that the resultant functions are continuous in their inputs~\citep{lillicrap2015continuous}.

The following assumption, as done in many works~\citep{cui2023learning, cui2023major}, states that the gradient of the coordinator policy is continuous and uniformly bounded.

\begin{assumption}\label{assumpt:bounded-gradient-log-prob}
    The log-gradient $\nabla\log\alpha_\theta (\pi_t | \mu_t ,\nu_t)$ is $L_{\nabla}$-Lipschitz continuous and uniformly bounded.
\end{assumption}
Furthermore, we have the following result from~\cite{guan2024zero}:

    \begin{corollary}
    \label{cor:mf-aprx}
    Let $\bfX^{N_1}_t$, $\bfY^{N_2}_t$, $\M^{N_1}_t$, and $\N^{N_2}_t$ be the joint states and the corresponding \ED{}s of a finite-population game.
    Denote the next Blue \ED{} induced by an identical Blue team policy $\phi_t \in \Phi_t$ as $\M^{N_1}_{t+1}$. 
    Then, the following holds:
    \begin{align*}
        \expct{\dtv{\M_{t+1}^{N_1}, \mu_{t+1}} \big\vert \; \bfX^{N_1}_t, \bfY^{N_2}_t}{\phi_t} \leq \frac{|\X|}{2}\sqrt{\frac{1}{N_1}},
    \end{align*}
    where 
    $\mu_{t+1} = \M^{N_1}_{t} F_t(\M_{t}^{N_1}, \N_{t}^{N_2},\phi_t)$.
    \end{corollary}

As mentioned in the main text, we assume that $r_t \in [-R_{\max}, R_{\max}]$.
Consequently, the Q-functions $Q_t^{\alpha, \beta}$ and $Q_t^{N, \alpha, \beta}$ are uniformly bounded over the finite-horizon.
We are interested in deriving an explicit result for $ \|\nabla_{\theta}  J(\alpha_\theta) - \nabla_{\theta}  \hat J^{N_1}(\alpha_\theta)\|_2 $ as the population sizes $(N_1, N_2)\to\infty$.
In particular,  using shorthand notation $\expct{\cdot}{\alpha,\beta}\triangleq\expct{\cdot}{\pi\sim\alpha, \sigma\sim\beta}$ and $\expct{\cdot}{\phi, \psi}\triangleq\expct{\cdot}{\u^{N_1}\sim\phi, \v^{N_2}\sim\psi}$)
\begin{align*}
    \|\nabla_{\theta}  J(\alpha_\theta) - \nabla_{\theta}  \hat J^{N_1}(\alpha_\theta)\|_2
    &=  \Big\|\sum_{t=0}^T \expct{Q_t^{\alpha, \beta}(\mu_t, \nu_t, \pi_t, \sigma_t)\nabla_\theta \log\alpha_\theta (\pi_t | \mu_t ,\nu_t)}{\alpha,\beta} \\
    &-\expct{Q_t^{N, \alpha, \beta}(\mu^{N_1}_t, \nu^{N_2}_t, \hat\pi^{N_1}_t, \hat\sigma^{N_2}_t)\nabla_\theta \log\alpha_\theta (\hat\pi^{N_1}_t |\mu^{N_1}_t, \nu^{N_2}_t)}{\phi, \psi} \Big\|_2.
\end{align*}
For the rest of the proof, we implicitly assume $\|\cdot\|\equiv\|\cdot\|_2$ to denote the $2$-norm.
By Lemma~\ref{lmm:equiv-expectation}, under the common randomness assumption, we can rewrite the above as
\begin{align*}
    \|\nabla_{\theta}  J(\alpha_\theta) - \nabla_{\theta}  \hat J^{N_1}(\alpha_\theta)\|
    &=  \Big\|\sum_{t=0}^T  \mathbb{E}_{\alpha, \beta}\Big[Q_t^{\alpha, \beta}(\mu_t, \nu_t, \pi_t, \sigma_t)\nabla_\theta \log\alpha_\theta (\pi_t | \mu_t ,\nu_t)\\&-\expct{Q_t^{N, \alpha, \beta}(\mu^{N_1}_t, \nu^{N_2}_t, \hat\pi^{N_1}_t, \hat\sigma^{N_2}_t)\nabla_\theta \log\alpha_\theta (\hat\pi^{N_1}_t |\mu^{N_1}_t, \nu^{N_2}_t)}{\mathbf{u}\sim\pi, \mathbf{v}\sim\sigma}\Big] \Big\|,
\end{align*}
where $\hat\pi^{N_1}_t, \hat\sigma^{N_2}_t$ depend on the random variables $\mathbf{u}^{N_1}\sim\pi, \mathbf{v}^{N_2}\sim\sigma$ respectively.
We add and subtract the term $\expct{Q_t^{\alpha, \beta}(\mu^{N_1}_t, \nu^{N_2}_t, \hat\pi^{N_1}_t, \hat\sigma^{N_2}_t)\nabla_\theta \log\alpha_\theta (\hat\pi^{N_1}_t |\mu^{N_1}_t, \nu^{N_2}_t)}{\mathbf{u}^{N_1}\sim\pi, \mathbf{v}^{N_2}\sim\sigma}$ and split the gradient terms and analyze the two terms individually.
\begin{align}\label{eq:split-grad}
    &\|\nabla_{\theta}  J(\alpha_\theta) - \nabla_{\theta}  \hat J^{N_1}(\alpha_\theta)\|\nonumber
    \\&\leq \Big\|\sum_{t=0}^T \mathbb{E}_{\alpha, \beta}\Big[Q_t^{\alpha, \beta}(\mu_t, \nu_t, \pi_t, \sigma_t)\nabla_\theta \log\alpha_\theta (\pi_t | \mu_t ,\nu_t)\nonumber\\&-\expct{Q_t^{\alpha, \beta}(\mu^{N_1}_t, \nu^{N_2}_t, \hat\pi^{N_1}_t, \hat\sigma^{N_2}_t)\nabla_\theta \log\alpha_\theta (\hat\pi^{N_1}_t |\mu^{N_1}_t, \nu^{N_2}_t)}{\mathbf{u}^{N_1}\sim\pi, \mathbf{v}^{N_2}\sim\sigma}\Big]\Big \|
    \nonumber\\&+ \Big\|\sum_{t=0}^T \expct{\Big(Q_t^{\alpha, \beta}(\mu^{N_1}_t, \nu^{N_2}_t, \hat\pi^{N_1}_t, \hat\sigma^{N_2}_t) -Q_t^{N, \alpha, \beta}(\mu^{N_1}_t, \nu^{N_2}_t, \hat\pi^{N_1}_t, \hat\sigma^{N_2}_t)\Big)\nabla_\theta \log\alpha_\theta (\hat\pi^{N_1}_t |\mu^{N_1}_t, \nu^{N_2}_t)}{\alpha, \beta, \mathbf{u}^{N_1}\sim\pi, \mathbf{v}^{N_2}\sim\sigma}\Big] \Big\|.
\end{align}

Since the first term in the above expression~\eqref{eq:split-grad} is a constant with respect to the random variables $\mathbf{u}^{N_1}, \mathbf{v}^{N_2}$ and the expectation of a constant is the value itself, we can equivalently write it as 
\begin{align}\label{eq:grad-log-prob-term-1-combined-expct}
   &\Big\|\sum_{t=0}^T \mathbb{E}_{\alpha, \beta, \mathbf{u}^{N_1}\sim\pi, \mathbf{v}^{N_2}\sim\sigma}\Big[Q_t^{\alpha, \beta}(\mu_t, \nu_t, \pi_t, \sigma_t)\nabla_\theta \log\alpha_\theta (\pi_t | \mu_t ,\nu_t)\nonumber\\&-Q_t^{\alpha, \beta}(\mu^{N_1}_t, \nu^{N_2}_t, \hat\pi^{N_1}_t, \hat\sigma^{N_2}_t)\nabla_\theta \log\alpha_\theta (\hat\pi^{N_1}_t |\mu^{N_1}_t, \nu^{N_2}_t)\Big]\Big \|,
\end{align}

Define $h_t(\mu, \nu, \pi, \sigma) \triangleq Q_t^{\alpha, \beta}(\mu, \nu, \pi, \sigma)\nabla_\theta \log\alpha_\theta (\pi| \mu ,\nu)$.
Assumptions~\ref{assumpt:continuous-Q} and~\ref{assumpt:bounded-gradient-log-prob} ensure continuity of the function $h_t(\mu, \nu, \pi, \sigma)$ with respect to all its inputs.
Using this definition of $h$, rewrite~\eqref{eq:grad-log-prob-term-1-combined-expct} as 
\begin{align}\label{eq:h-def-continuity}
    \Big\|\sum_{t=0}^T \mathbb{E}_{\alpha, \beta, \mathbf{u}\sim\pi, \mathbf{v}\sim\sigma}\Big[h_t(\mu_t, \nu_t, \pi_t, \sigma_t)- h_t(\mu^{N_1}_t, \nu^{N_2}_t, \hat\pi^{N_1}_t, \hat\sigma^{N_2}_t)\Big]\Big \|.
\end{align}
 Using the definition of the expectation operators, 
 \begin{align*}
    &\Big\|\sum_{t=0}^T \gamma^t \mathbb{E}_{\alpha, \beta, \mathbf{u}^{N_1}\sim\pi, \mathbf{v}^{N_2}\sim\sigma}\Big[h_t(\mu_t, \nu_t, \pi_t, \sigma_t)- h_t(\mu^{N_1}_t, \nu^{N_2}_t, \hat\pi^{N_1}_t, \hat\sigma^{N_2}_t)\Big] \Big\|\\ &=\Big\|\sum_{t=0}^T \gamma^t  \int_{\pi}\int_{\sigma}h_t(\mu_t, \nu_t, \pi_t, \sigma_t)\beta_t(\mu_t, \nu_t) \alpha_t(\mu_t, \nu_t)d\sigma d\pi\\ 
    &-\int_{\pi}\int_{\sigma}\Big[\sum_{\mathbf{u}^{N_1}}\sum_{\mathbf{v}^{N_2}} h_t(\mu^{N_1}_t, \nu^{N_2}_t, \hat\pi^{N_1}_t, \hat\sigma^{N_2}_t)\pi(\u^{N_1}|\x^{N_1}_t)\sigma(\v^{N_2}|\y^{N_1}_t)\Big]\beta_t(\mu_t, \nu_t) \alpha_t(\mu_t, \nu_t)d\sigma d\pi  \Big\|
\end{align*}
We already have, by weak LLN that $\dtv{\mu^{N_1}_t,\mu_t} \leq\mathcal{O}(1/\sqrt{N_1})$ and  $\dtv{\nu^{N_2}_t,\nu_t} \leq\mathcal{O}(1/\sqrt{N_2})$, with the error going to zero as the population size becomes large, i.e., $N_1, N_2 \to \infty$.
A similar argument using weak LLN can be made for $(\pi^{N_1}_t, \sigma^{N_2}_t) \to (\pi_t, \sigma_t)$ as $N_1, N_2 \to \infty$.
Furthermore, as the population size becomes large, the empirical policies constructed $(\hat\pi^{N_1}_t, \hat\sigma^{N_2}_t) \to (\pi^{N_1}_t, \sigma^{N_2}_t)$ because the number of samples of $\u^{N_1}_t$ and $\v^{N_2}_t$ scale with $(N_1, N_2)$.\footnote{Note that as $N_1\to\infty$, $\hat\pi^{N_1}_t(\cdot | x) \neq  \pi^{N_1}_t(\cdot | x) $ for states $x$ such that $\mu(x) = 0$ by the very definition of the approximation~\eqref{eqn:apprx-local-policy}. However, it is easy to check that the mean-field trajectories remain the same and the unoccupied states have no role to play in the evolution.} 
Thus, it follows that $(\hat\pi^{N_1}_t, \hat\sigma^{N_2}_t) \to (\pi_t, \sigma_t)$ as the team sizes become large and, coupled with the continuity of $h$ with respect to $\hat\pi^{N_1}_t, \hat\sigma^{N_2}_t$, the following simplification holds:
\begin{align*}
    \sum_{\mathbf{u}}\sum_{\mathbf{v}} h_t(\mu^{N_1}_t, \nu^{N_2}_t, \hat\pi^{N_1}_t, \hat\sigma^{N_2}_t)\pi(\u|\x^{N_1}_t)\sigma(\v|\y^{N_1}_t) &\to \sum_{\mathbf{u}}\sum_{\mathbf{v}} h_t(\mu_t, \nu_t, \pi_t, \sigma_t)\pi(\u|\x^{N_1}_t)\sigma(\v|\y^{N_1}_t)  \\ &= h_t(\mu_t, \nu_t, \pi_t, \sigma_t)\sum_{\mathbf{u}} \pi(\u^{N_1}|\x^{N_1}_t)\sum_{\mathbf{v}}\sigma(\v^{N_2}|\y^{N_1}_t)  \\&=
    h_t(\mu_t, \nu_t, \pi_t, \sigma_t).
\end{align*}

For the second term in~\eqref{eq:split-grad}, we first focus on the term 
\begin{align*}
    \Big(Q_t^{\alpha, \beta}(\mu^{N_1}_t, \nu^{N_2}_t, \hat\pi^{N_1}_t, \hat\sigma^{N_2}_t)-Q_t^{N, \alpha, \beta}(\mu^{N_1}_t, \nu^{N_2}_t, \hat\pi^{N_1}_t, \hat\sigma^{N_2}_t)\Big).
\end{align*}
The expectation in~\eqref{eq:q-func-finite} is with respect to $\u^{N_1}\sim\phi, \v^{N_2}\sim\psi$.
It follows from Lemma~\ref{lmm:equiv-expectation}, we can rewrite~\eqref{eq:q-func-finite}as
\begin{align*}
    &Q_t^{N, \alpha, \beta}(\mu^{N_1}_t, \nu^{N_2}_t, \hat\pi^{N_1}_t, \hat\sigma^{N_2}_t) \nonumber \\&= \expct{\sum_{\tau = 0}^T r(\mu^{N_1}_\tau, \nu^{N_2}_\tau) \big\vert \mu^{N_1}_{0} = \mu^{N_1}_t, \nu^{N_2}_{0} = \nu^{N_2}_t, \pi^{N_1}_{0} = \hat\pi^{N_1}_t, \sigma^{N_2}_{0} = \hat\sigma^{N_2}_t}{\pi\sim\alpha, \sigma\sim\beta,\u^{N_1}\sim\pi, \v^{N_2}\sim\sigma }
\end{align*}
On the other hand, the expectation in~\eqref{eq:q-func-infty} is with respect to $\alpha, \beta$.
Using a similar argument as~\eqref{eq:grad-log-prob-term-1-combined-expct}, we can write~\eqref{eq:q-func-infty} as
\begin{align*}
    &Q_t^{\alpha, \beta}(\mu^{N_1}_t, \nu^{N_2}_t, \hat\pi^{N_1}_t, \hat\sigma^{N_2}_t) \\&= \expct{\sum_{\tau = 0}^T r(\mu_\tau, \nu_\tau) \big\vert \mu_{0} = \mu^{N_1}_t, \nu_{0} =  \nu^{N_2}_t, \pi_{0} = \hat\pi^{N_1}_t, \sigma_{0} = \hat\sigma^{N_2}_t}{\pi\sim\alpha, \sigma\sim\beta,\u^{N_1}\sim\pi, \v^{N_2}\sim\sigma  }.
\end{align*}

We apply the following lemma along with Assumption~\ref{assumpt:bounded-gradient-log-prob} to establish convergence of the second term in~\eqref{eq:split-grad}.
\begin{lemma}
Under Assumption~\ref{assmpt:lipschitiz-model} and Lipschitz policies~\ref{assumpt:bounded-gradient-log-prob}, $Q^{N, \alpha, \beta}\to Q^{\alpha, \beta}$ as $N_1, N_2\to\infty$. 
\end{lemma}
\begin{proof}
    Assumption~\ref{assmpt:lipschitiz-model} allows the $Q_t^{N, \alpha, \beta}, Q_t^{\alpha, \beta}$ to have compact support.
    Thus, it suffices to prove pointwise convergence~\citep{cui2023major} for each $\tau$, i.e.,
    \begin{align}\label{eq:q-con-pointwise}
        &\expct{ r(\mu^{N_1}_\tau, \nu^{N_2}_\tau) \big\vert \mu^{N_1}_{0} = \mu^{N_1}_t, \nu^{N_2}_{0} = \nu^{N_2}_t, \pi^{N_1}_{0} = \hat\pi^{N_1}_t, \sigma^{N_2}_{0} = \hat\sigma^{N_2}_t}{} \to \nonumber\\
        &\expct{r(\mu_\tau, \nu_\tau) \big\vert \mu_{0} = \mu^{N_1}_t, \nu_{0} =  \nu^{N_2}_t, \pi_{0} = \hat\pi^{N_1}_t, \sigma_{0} = \hat\sigma^{N_2}_t}{}
    \end{align}
    as $N_1, N_2\to\infty$, where the expectation is taken with respect to $\alpha, \beta,\u\sim\pi, \v\sim\sigma$ (dropped for readability).
    We follow an induction approach similar to the one presented in~\cite{cui2023major}.
    It trivially holds at $\tau = 0$.
    Let us assume that~\eqref{eq:q-con-pointwise} holds at time step $\tau$.
    For subsequent $\tau + 1$,
    \begin{align*}
        &\Big\|\expct{ r(\mu_{\tau+1}, \nu_{\tau+1}) \big\vert \mu_{0} = \mu^{N_1}_t, \nu_{0} =  \nu^{N_2}_t, \pi_{0} = \hat\pi^{N_1}_t, \sigma_{0} = \hat\sigma^{N_2}_t}{} \\
        &-\expct{ r(\mu^{N_1}_{\tau+1}, \nu^{N_2}_{\tau+1}) \big\vert \mu^{N_1}_{0} = \mu^{N_1}_t, \nu^{N_2}_{0} = \nu^{N_2}_t, \pi^{N_1}_{0} = \hat\pi^{N_1}_t, \sigma^{N_2}_{0} = \hat\sigma^{N_2}_t}{} \Big\| 
        \\
        &\leq\Big\|\expct{ r\Big(\mu_{\tau}F(\mu_{\tau}, \nu_{\tau}, \pi_\tau), \nu_{\tau}G(\mu_{\tau}, \nu_{\tau}, \sigma_\tau)\Big) \big\vert \mu_{0} = \mu^{N_1}_t, \nu_{0} =  \nu^{N_2}_t, \pi_{0} = \hat\pi^{N_1}_t, \sigma_{0} = \hat\sigma^{N_2}_t}{} \\
        &-\expct{r\Big(\mu^{N_1}_{\tau}F(\mu^{N_1}_{\tau}, \nu^{N_2}_{\tau}, \phi_\tau), \nu^{N_2}_{\tau}G(\mu^{N_1}_{\tau}, \nu^{N_2}_{\tau}, \psi_\tau)\Big) \big\vert \mu^{N_1}_{0} = \mu^{N_1}_t, \nu^{N_2}_{0} = \nu^{N_2}_t, \pi^{N_1}_{0} = \hat\pi^{N_1}_t, \sigma^{N_2}_{0} = \hat\sigma^{N_2}_t}{} \Big\| \\
        &+ \Big\|\expct{r\Big(\mu^{N_1}_{\tau}F(\mu^{N_1}_{\tau}, \nu^{N_2}_{\tau}, \phi_\tau), \nu^{N_2}_{\tau}G(\mu^{N_1}_{\tau}, \nu^{N_2}_{\tau}, \psi_\tau)\Big) \big\vert \mu^{N_1}_{0} = \mu^{N_1}_t, \nu^{N_2}_{0} = \nu^{N_2}_t, \pi^{N_1}_{0} = \hat\pi^{N_1}_t, \sigma^{N_2}_{0} = \hat\sigma^{N_2}_t}{} \\
        &-\expct{r(\mu^{N_1}_{\tau+1}, \nu^{N_2}_{\tau+1}) \big\vert \mu^{N_1}_{0} = \mu^{N_1}_t, \nu^{N_2}_{0} = \nu^{N_2}_t, \pi^{N_1}_{0} = \hat\pi^{N_1}_t, \sigma^{N_2}_{0} = \hat\sigma^{N_2}_t}{}\Big \|  
    \end{align*}

The convergence of the first term follows from the weak LLN~\citep{cui2023major} and a similar argument as used for~\eqref{eq:h-def-continuity}.
By using the Lipshcitzness of the reward function (Assumption~\ref{assmpt:lipschitiz-model}) and Corollary~\ref{cor:mf-aprx}, the second term in the expression is bounded by $\Residue$, where $\underline{N} = \min(N_1, N_2)$ and goes to zero as $(N_1, N_2) \to \infty.$
\end{proof}

It follows from the convergence of the two individual terms in~\ref{eq:split-grad} that $ \|\nabla_{\theta}  J(\alpha_\theta) - \nabla_{\theta}  \hat J^{N_1}(\alpha_\theta)\| \to 0$ as $(N_1, N_2)\to\infty$
\end{proof}

\section{Proof of Theorem~\ref{thm:extend-to-many}}\label{proof:extend-to-many}
\maintheoremthree*
\begin{proof}[Sketch of Proof]
The key idea of the proof is to exploit the fact that the equivalent infinite-population games for $\G_1$ and $\G_2$ are the same as $N_1/N_2  = \barN_1/\barN_2$~\citep{guan2024zero}.
   We first compute the performance of $\phi^*$, the optimal identical team policy derived for $\G_1$ using MF-MAPPO (finite-population training), when applied to the \textit{infinite-population} coordinator setting.
   We then compare the performance of the aforementioned $\phi^*$ for $\G_2$ with the \textit{infinite-population} coordinator.
   Then, equating the common equivalent coordinators for $\G_1$ and $\G_2$, the result follows.
\end{proof}

\begin{proof}

Consider game $\G_1$.
Let us restrict ourselves to the set of identical team policies $\Phi \subseteq \Phi^{N_1}$ and suppose that  $\phi^*$ is the optimal identical Blue team policy obtained from the $(N_1, N_2)$ finite population game via MF-MAPPO.
From Theorem~\ref{thm:finite-population-training}, we already know that
\begin{align*}
    \min_{\psi^{N_2} \in \Psi^{N_2}} 
       J^{N,\phi^*,\psi^{N_2}} (\bfx^{N_1}, \bfy^{N_2}) 
       \geq \underline{J}^{N*}(\bfx^{N_1}, \bfy^{N_2})-\Residue.
\end{align*}
Denote $J^{N,\phi^*,\psi^{N_2*}} (\bfx^{N_1}, \bfy^{N_2})  \triangleq \min_{\psi^{N_2} \in \Psi^{N_2}} 
       J^{N,\phi^*,\psi^{N_2}} (\bfx^{N_1}, \bfy^{N_2})$,
where $\psi^{N_2*}$ is the potentially non-identical policy that solves the optimization problem.

Now, for $\G_2$ we want to show that
\begin{align}\label{eq:g2-game-desired-result}
     \min_{\psi^{\bar{N}_2} \in \Psi^{\bar{N}_2}} 
        J^{\bar{N},\phi^*,\psi^{\barN_2}} (\bfx^{\barN_1}, \bfy^{\barN_2}) \geq \underline{J}^{\barN*}(\bfx^{\barN_1}, \bfy^{\barN_2})-\Residue,
\end{align}
where $\phi^*$ corresponds to the identical policy derived from $\G_1$.

We know that the $\frac{N_1}{N_2} = \frac{\barN_1}{\barN_2}$.
Thus, the infinite-population coordinator games remain identical for both $\G_1$ and $\G_2$.
%
%
We prove the theorem in the following two steps.

\textbf{Step 1.} We first show that for $\phi^*$ from $\G_1$ we have the following inequality for all joint states $\bfx^{N_1} \in \X^{N_1}$ and $\bfy^{N_2} \in \Y^{N_2}$, 
\begin{align}\label{eq:step1-inequality-e2m}
     {J}_\infty^{\phi^*, \tilde\psi^*_\infty}(\mu^{N_1}, \nu^{N_2})\geq  J^{N,\phi^*,\psi^{N_2*}} (\bfx^{N_1}, \bfy^{N_2}) - \Residue,
\end{align}
where $\mu^{N_1}=\empMu{\bfx^{N_1}}$ and $\nu^{N_2} = \empNu{\bfy^{N_2}}$ and $\tilde\psi^*_\infty\in\Psi$ is given by,
\begin{align}\label{eq:phi-star-inf-opt}
    {J}_\infty^{\phi^*, \tilde\psi^*_\infty} (\mu_0, \nu_0) = \min_{\psi \in \Psi} ~ J_\infty^{\phi^*, \psi}(\mu_0, \nu_0).
\end{align}
The above equation~\eqref{eq:phi-star-inf-opt} determines the optimal Red team response in the \textit{infinite-population} domain, when the Blue team plays $\phi^*$ derived from $\G_1$.
%
We prove this through an inductive argument. 

    \textit{Base case:}
    At the terminal timestep $T$, the two value functions are the same. Thus, we have, for all joint states $\bfx^{N_1}_T \in \bfX^{N_1}$ and $\bfy^{N_2}_T \in \bfY^{N_2}$, that
    \begin{equation*}
        {J}_{\infty,T}^{\phi^*, \tilde\psi^*_\infty}(\mu^{N_1}_T , \nu^{N_1}_T)
         =  J^{N,\phi^*,\psi^{N_2*}} _T(\bfx^{N_1}_T , \bfy^{N_2}_T) = r_t (\mu^{N_1}_T , \nu^{N_1}_T).
    \end{equation*}

    \textit{Inductive hypothesis:}
    Assume that, at time step $t+1$, the following holds for all $\bfx^{N_1}_{t+1} \in \bfX^{N_1}$ and $\bfy^{N_2}_{t+1} \in \bfY^{N_2}$,
    \begin{equation}
        {J}_{\infty,t+1}^{\phi^*, \tilde\psi^*_\infty}(\mu^{N_1}_{t+1} , \nu^{N_1}_{t+1}) \geq 
         J^{N,\phi^*,\psi^{N_2*}} _{t+1}(\bfx^{N_1}_{t+1} , \bfy^{N_2}_{t+1}) - \Residue.
    \end{equation}
    \textit{Induction:}
    Consider arbitrary $\bfx^{N_1}_{t} \in \bfX^{N_1}$ and $\bfy^{N_2}_{t} \in \bfY^{N_2}$.
    
    For simplicity, we do not emphasize the correspondence between the joint states and the \ED{}s for the rest of the proof, as it is clear from the context.
    For the identical team policies $(\phi_t^*,\psi_t) \in \Phi\times\Psi$, at time step $t$, denote 
    \begin{align}\label{eq:determinstic-extend-to-many-proof}
        \mu^{\phi^*}_{t+1} \triangleq \mu^{N_1}_t F(\mu^{N_1}_t, \nu^{N_2}_t, \phi^*_t),
        \qquad
        \nu^{\psi}_{t+1} \triangleq \nu^{N_2}_t G(\mu^{N_1}_t, \nu^{N_2}_t, \psi_t).
    \end{align}
    For notational simplicity, we drop the conditions $\bfX_t^{N_1} =\bfx^{N_1}_t$ and $\bfY_t^{N_2}=\bfy^{N_2}_t$ in the following derivations.
    Then, we have
    \begin{align*}
        &  J^{N,\phi^*,\psi^{N_2*}}_{t}(\bfx^{N_1}_{t} , \bfy^{N_2}_{t}) \\
        &=  \min_{\psi^{N_2}_t \in \Psi^{N_2}_t} r_{t} (\mu^{N_1}_t, \nu^{N_2}_t) + \mathbb{E}_{\phi^{*}_t,\psi_t^{N_2}} \Big[J^{N,\phi^*,\psi^{N_2*}}_{t+1} (\bfX^{N_1}_{t+1} , \bfY^{N_2}_{t+1}) \Big]\\
        & \stackrel{\text{(i)}}{\leq}    r_{t} (\mu^{N_1}_t, \nu^{N_2}_t) + \min_{\psi_t^{N_2} \in \Psi^{N_2}_t}\mathbb{E}_{\phi^{*}_t,\psi_t^{N_2}} \Big[J_{\infty, t+1}^{\phi^*, \tilde\psi^*_\infty}(\M^{N_1}_{t+1}, \N^{N_2}_{t+1}) \Big]  +\Residue \\
        & \stackrel{\text{(ii)}}{\leq}  r_{t} (\mu^{N_1}_t, \nu^{N_2}_t) + \min_{\psi_t \in \Psi_t} \;\; \mathbb{E}_{\phi^{*}_t,\psi_t} \Big[J_{\infty, t+1}^{\phi^*, \tilde\psi^*_\infty}(\M^{N_1}_{t+1}, \N^{N_2}_{t+1}) \Big]  +\Residue\\
        &\stackrel{\text{(iii)}}{=}  r_{t} (\mu^{N_1}_t, \nu^{N_2}_t) +\Residue
        + \min_{\psi_t \in \Psi_t}\;\; \mathbb{E}_{\phi^{*}_t,\psi_t} \Big[J_{\infty, t+1}^{\phi^*, \tilde\psi^*_\infty}(\M^{N_1}_{t+1}, \N^{N_2}_{t+1})- J_{\infty, t+1}^{\phi^*, \tilde\psi^*_\infty}(\mu^{\phi^*}_{t+1}, \nu^{\psi_t}_{t+1}) \\&\qquad \qquad \qquad \qquad \qquad \qquad \qquad \qquad \qquad \qquad \qquad \qquad    + J_{\infty, t+1}^{\phi^*, \tilde\psi^*_\infty}(\mu^{\phi^*}_{t+1}, \nu^{\psi_t}_{t+1})\Big]\\
        &\stackrel{\text{(iv)}}{\leq}  r_{t} (\mu^{N_1}_t, \nu^{N_2}_t)  +\Residue
        + \min_{\psi_t \in \Psi_t}\;\; \mathbb{E}_{\phi^{*}_t,\psi_t} \Big[L_{J, t+1}\dtv{\M^{N_1}_{t+1}, \mu^{\phi^*}_{t+1}}  \\&\qquad \qquad \qquad \qquad \qquad \qquad \qquad \qquad   + L_{J, t+1}\dtv{\N^{N_2}_{t+1}, \nu^{\psi_t}_{t+1}} + J_{\infty, t+1}^{\phi^*, \tilde\psi^*_\infty}(\mu^{\phi^*}_{t+1}, \nu^{\psi_t}_{t+1})\Big]\\
        &\stackrel{\text{(v)}}{=}  r_{t} (\mu^{N_1}_t, \nu^{N_2}_t) +\Residue
        + \min_{\psi_t \in \Psi_t} J_{\infty, t+1}^{\phi^*, \tilde\psi^*_\infty}(\mu^{\phi^*}_{t+1}, \nu^{\psi_t}_{t+1})+  L_{J, t+1}\mathbb{E}_{\psi_t} \Big[ \dtv{\N^{N_2}_{t+1}, \nu^{\psi_t}_{t+1}} \Big] \\&\qquad \qquad \qquad \qquad \qquad \qquad \qquad \qquad \qquad \qquad \qquad \qquad    + L_{J, t+1}\mathbb{E}_{\phi^*_t} \Big[\dtv{\M^{N_1}_{t+1}, \mu^{\phi^*}_{t+1}} \Big]\\
        &\stackrel{\text{(vi)}}{\leq}  r_{t} (\mu^{N_1}_t, \nu^{N_2}_t)+ \min_{\psi_t \in \Psi_t} J_{\infty, t+1}^{\phi^*, \tilde\psi^*_\infty}(\mu^{\phi^*}_{t+1}, \nu^{\psi_t}_{t+1}) +\Residue\\
        &\stackrel{\text{(vii)}}{=}   r_{t} (\mu^{N_1}_t, \nu^{N_2}_t)+ \min_{\nu_{t+1}\in \mathds{R}_{\nu,t}(\mu^{N_1}_t, \nu^{N_2}_t)} J_{\infty, t+1}^{\phi^*, \tilde\psi^*_\infty}(\mu^{\phi^*}_{t+1}, \nu^{\psi_t}_{t+1}) +\Residue\\
        &\stackrel{\text{(viii)}}{=} J_{\infty, t}^{\phi^*, \tilde\psi^*_\infty}(\mu^{N_1}_t, \nu^{N_2}_t ) +\Residue
    \end{align*}
    For inequality (i), we used the inductive hypothesis; 
    for inequality (ii), we reduced the optimization domain of the Red team to the set of \textit{identical} team policies;
    inequality (iv) is a result of the Lipschitz continuity of the value function~\citep{guan2024zero};
    for equality (v), we use the fact that the mean-fields $(\mu^{\phi^*}_{t+1}, \nu^{\psi_t}_{t+1})$ are induced deterministically from the distributions at time $t$;
    inequality (vi) holds as a consequence of Corollary~\ref{cor:mf-aprx}; (vii) converts the optimization domain from the policy space $\psi_t\in\Psi_t$ to the corresponding reachable
set $\mathds{R}_{\nu,t}(\mu^{N_1}_t, \nu^{N_2}_t) \triangleq \{\nu_{t+1} \vert  \exists \psi_t \in \Psi_t \st \! \nu_{t+1} = \nu_t G^{}_t(\mu_t,\nu_t, \psi_t)\}$ following~\citep{guan2024zero}; (vii) follows from the definition of $J_{\infty, t}^{\phi^*, \tilde\psi^*_\infty}(\mu^{N_1}_t, \nu^{N_2}_t )$.
    Thus, $J^{N,\phi^*,\psi^{N_2*}}(\bfx^{N_1}, \bfy^{N_2}) = J^{N,\phi^*,\psi^{N_2*}}_0(\bfx^{N_1}, \bfy^{N_2})$ and~\eqref{eq:step1-inequality-e2m} follows.
    
\textbf{Step 2.} We show that for $\phi^*$ from $\G_1$ we have the following inequality for all joint states $\bfx^{\barN_1} \in \X^{\barN_1}$ and $\bfy^{\barN_2} \in \Y^{\barN_2}$ in $\G_2$, 
\begin{align}\label{eqn:step2-inequality-e2m}
   \min_{\psi^{\bar{N}_2} \in \Psi^{\bar{N}_2}} 
        J^{\bar{N},\phi^*,\psi^{\barN_2}} (\bfx^{\barN_1}, \bfy^{\barN_2}) \geq {J}_\infty^{\phi^*, \tilde\psi^*_\infty} (\mu^{\barN_1}, \nu^{\barN_2})  -\ResidueNew,
\end{align}
where $\mu^{\barN_1}=\empMu{\bfx^{\barN_1}}$, $\nu^{\barN_2} = \empNu{\bfy^{\barN_2}}$ and $\bar{\subN} =\min(\barN_1, \barN_2)$.  The proof is constructed based on induction.
    Fix an arbitrary Red team policy $\psi^{\barN_2} \in \Psi^{\barN_2}$.
    
    \textit{Base case:}
    At the terminal timestep $T$, since there is no decision to be made, both value functions are equal to the terminal reward and are thus the same.
    Formally, for all $\bfx^{\barN_1}_T \in \bfX^{\barN_1}$ and $\bfy^{\barN_2}_T \in \bfY^{\barN_2}$,
    \begin{align*}
        {J}_{T}^{{\bar{N},\phi^*,\psi^{\barN_2}}}(\bfx^{\barN_1}_T , \bfy^{\barN_2}_T) = {J}_{\infty,T}^{\phi^*, \tilde\psi^*_\infty}(\mu^{N_1}_T , \nu^{N_1}_T)
         =  r_t (\mu^{\barN_1}_T , \nu^{\barN_1}_T),
    \end{align*} 
    where $\mu^{\barN_1}_T = \empMu{\bfx^{\barN_1}_T}$ and $\nu^{\barN_2}_T = \empNu{\bfy^{\barN_2}_T}$.
    For simplicity, we do not emphasize the correspondence between the joint states and the \ED{}s for the rest of the proof, as it is clear from the context.
    
    \textit{Inductive hypothesis: }
    Assume that at $t+1$, the following holds for all joint states  $\bfx^{\barN_1}_{t+1} \in \bfX^{\barN_1}$ and $\bfy^{\barN_2}_{t+1} \in \bfY^{\barN_2}$:
    \begin{equation}
        J_{t+1}^{{\bar{N},\phi^*,\psi^{\barN_2}}}(\bfx^{\barN_1}_{t+1}, \bfy^{N_2}_{t+1}) \geq 
        {J}_{\infty,t+1}^{\phi^*, \tilde\psi^*_\infty}(\mu^{\barN_1}_{t+1} , \nu^{\barN_1}_{t+1}) - \ResidueNew.
    \end{equation}
    
    \textit{Induction: }
    At timestep $t$, consider an arbitrary pair of joint states $\bfx^{\barN_1}_{t} \in \bfX^{\barN_1}$ and $\bfy^{\barN_2}_{t} \in \bfY^{\barN_2}$, and their corresponding \ED{}s $\mu^{\barN_1}_t$ and $\nu^{\barN_2}_t$.

    For the identical team policy $\phi^* \in \Phi$,  denote $\mu^{\phi^*}_{t+1})$ from~\eqref{eq:determinstic-extend-to-many-proof}.
    Furthermore, from Theorem 1 in~\cite{guan2024zero} there exists a $\nu^{\psi^{\barN_2}_t}_{\apprx, t+1} $ for the Red team policy $\psi_t^{N_2}$ such that
    \begin{equation}
        \label{eqn:red-apprx-nu}
        \expct{\dtv{\N^{\barN_2}_{t+1},\nu^{\psi^{\barN_2}_t}_{\apprx,t+1}}\big \vert \bfX_t^{\barN_1} =\bfx^{\barN_1}_t, \bfY_t^{\barN_2}=\bfy^{\barN_2}_t }{\psi_t^{\barN_2}} \leq \frac{|\Y|}{2}\sqrt{\frac{1}{\barN_2}}.
    \end{equation}
     
    Then, for all joint states $\bfx^{\barN_1}_t \in \bfX^{\barN_1}$ and $\bfy^{\barN_2}_t \in \bfY^{\barN_2}$, we have
    \begin{align*}
        & J_{t}^{\bar{N},\phi^*,\psi^{\barN_2}}(\bfx^{\barN_1}_t, \bfy^{\barN_2}_t)\\
        &=r_t(\mu^{\barN_1}_t, \nu^{\barN_2}_t) + \mathbb{E}_{\phi^*, \psi^{\barN_2}}\Big[J_{t+1}^{\bar{N},\phi^*,\psi^{\barN_2}}(\bfX^{\barN_1}_{t+1}, \bfY^{\barN_2}_{t+1}) \Big]\\
        &\stackrel{\text{(i)}}{\geq}  r_t(\mu^{\barN_1}_t, \nu^{\barN_2}_t) + \mathbb{E}_{\phi^*, \psi^{\barN_2}}\Big[J_{\infty, t+1}^{\phi^*, \tilde\psi^*_\infty}(\M^{\barN_1}_{t+1}, \N^{\barN_2}_{t+1}) \Big] - \ResidueNew
        \\
        &=r_t(\mu^{\barN_1}_t, \nu^{\barN_2}_t)  - \ResidueNew + \mathbb{E}_{\phi^*, \psi^{\barN_2}}\Big[J_{\infty, t+1}^{\phi^*, \tilde\psi^*_\infty}(\M^{\barN_1}_{t+1}, \N^{\barN_2}_{t+1})
        \\&\qquad \qquad \qquad \qquad \qquad \qquad \qquad \qquad- J_{\infty, t+1}^{\phi^*, \tilde\psi^*_\infty}(\mu^{\phi^*}_{t+1}, \nu^{\psi^{\barN_2}_t}_{\apprx,t+1})  + J_{\infty, t+1}^{\phi^*, \tilde\psi^*_\infty}(\mu^{\phi^*}_{t+1}, \nu^{\psi^{\barN_2}_t}_{\apprx,t+1}) \Big] \nonumber\\
        &\stackrel{\text{(ii)}}{\geq}  r_t(\mu^{\barN_1}_t, \nu^{\barN_2}_t) +
        J_{\infty, t+1}^{\phi^*, \tilde\psi^*_\infty}(\mu^{\phi^*}_{t+1}, \nu^{\psi^{\barN_2}_t}_{\apprx,t+1}) -\ResidueNew- L_{J, t+1}  \underbrace{\mathbb{E}_{\phi^*}\Big[\dtv{\M_{t+1}^{\barN_1},\mu^{\phi^*}_{t+1}}\Big]}_{=\mathcal{O}(\frac{1}{\sqrt{\barN_1}}) \text{ due to Corollary~\ref{cor:mf-aprx}}}  
        \\
        &  \qquad \qquad \qquad \qquad \qquad \qquad \qquad \qquad\qquad \qquad -\underbrace{L_{J, t+1} \mathbb{E}_{ \psi^{N_2}}\Big[\dtv{ \N^{N_2}_{t+1},\nu^{\psi^{\barN_2}_t}_{\apprx,t+1}}\Big]}_{=\mathcal{O}(\frac{1}{\sqrt{\barN_2}}) \text{ due to\eqref{eqn:red-apprx-nu}}}  
        \nonumber \\
        & \stackrel{\text{(iii)}}{\geq}  r_t(\mu^{\barN_1}_t, \nu^{\barN_2}_t) +
         J_{\infty, t+1}^{\phi^*, \tilde\psi^*_\infty}(\mu^{\phi^*}_{t+1}, \nu^{\psi^{\barN_2}_t}_{\apprx,t+1})-\ResidueNew \\
        &\stackrel{\text{(iv)}}{\geq}  r_t(\mu^{\barN_1}_t, \nu^{\barN_2}_t) +
         \min_{\nu_{t+1}\in \mathds{R}_{\nu,t}(\mu^{\barN_1}_t, \nu^{\barN_2}_t)} J_{\infty, t+1}^{\phi^*, \tilde\psi^*_\infty}(\mu^{\phi^*}_{t+1}, \nu_{t+1})-\ResidueNew \\
        &\stackrel{\text{(v)}}{=}  {J}_{\infty,t}^{\phi^*, \tilde\psi^*_\infty}(\mu^{\barN_1}_{t} , \nu^{\barN_2}_{t}) -\ResidueNew.
    \end{align*}
    For inequality (i), we used the inductive hypothesis; 
    for inequality (ii), we utilized the Lipschitz continuity of the coordinator value function~\citep{guan2024zero} using mean-field distributions $(\mu^{\phi^*}_{t+1}, \nu^{\psi^{\barN_2}_t}_{\apprx,t+1})$ from~\eqref{eq:determinstic-extend-to-many-proof} and the deterministic transitions from $t$ to $t+1$;
    inequality (iv) follows from the definition of the $\min$ operator; and equality (v) follows from the definition of $J_{\infty, t}^{\phi^*, \tilde\psi^*_\infty}(\mu^{N_1}_t, \nu^{N_2}_t )$ which completes the induction.

    Since the Red team policy $\psi^{\barN_2} \in \Psi^{\barN_2}$ is arbitrary, we have that, for all joint states $\bfx^{\barN_1} \in \bfX^{\barN_1}$ and $\bfy^{\barN_2} \in \bfY^{\barN_2}$,
    \begin{align*}
        \min_{\psi^{\bar{N}_2} \in \Psi^{\bar{N}_2}} 
        J^{\bar{N},\phi^*,\psi^{\barN_2}} (\bfx^{\barN_1}, \bfy^{\barN_2}) &= \min_{\psi^{\bar{N}_2} \in \Psi^{\bar{N}_2}} 
        J_0^{\bar{N},\phi^*,\psi^{\barN_2}} (\bfx^{\barN_1}, \bfy^{\barN_2})\\& \geq \underline{J}_{\infty}^{\phi^*, \tilde\psi^*_\infty} (\mu^{\barN_1}, \nu^{\barN_2})  -\ResidueNew.
    \end{align*}
     
 %

\textbf{Step 3.}  Combining~\eqref{eq:step1-inequality-e2m} from Step 1 with Theorem~\ref{thm:finite-population-training}, we have for all joint states $\bfx^{N_1} \in \X^{N_1}$ and $\bfy^{N_2} \in \Y^{N_2}$,
\begin{align}
        \label{eqn:cor-apprx-err-bound}
        {J}_\infty^{\phi^*, \tilde\psi^*_\infty} (\mu^{N_1} , \nu^{N_1}) &\geq J^{N,\phi^*,\psi^{N_2*}}(\bfx^{N_1}, \bfy^{N_2})  - \Residue\nonumber\\&\geq  \lowervalue^{N*}(\bfx^{N_1} , \bfy^{N_2} )-\Residue,
    \end{align}
    where $\mu^{N_1}=\empMu{\bfx^{N_1}}$ and $\nu^{N_2} = \empNu{\bfy^{N_2}}$. 
    %
    

%

Assuming that the initial distributions are the same for $\G_1$ and $\G_2$, i.e., $\mu^{N_1}=\mu^{\barN_1}$ and $\nu^{N_2} = \nu^{\barN_2}$ we have the following sequence of inequalities
\begin{align*}
    \min_{\psi^{\bar{N}_2} \in \Psi^{\bar{N}_2}} 
        J^{\bar{N},\phi^*,\psi^{\barN_2}} (\bfx^{\barN_1}, \bfy^{\barN_2}) &\stackrel{\text{(i)}}{\geq} {J}_\infty^{\phi^*, \tilde\psi^*_\infty} (\mu^{\barN_1}, \nu^{\barN_2})  -\ResidueNew \\
        &\stackrel{\text{(ii)}}{\geq} \lowervalue^{N*}(\bfx^{N_1} , \bfy^{N_2} )-\Residue - \ResidueNew\\
        &\stackrel{\text{(iii)}}{\geq}  \underline{J}_\infty^{\phi^*_\infty, \psi^*_\infty} (\mu^{N_1} , \nu^{N_1}) - \Residue -\ResidueNew\\
        &\stackrel{\text{(iv)}}{\geq}  \underline{J}^{\barN*}(\bfx^{\barN_1}, \bfy^{\barN_2})  - \Residue -\ResidueNew
\end{align*}
where ${J}_\infty^{\phi^*_\infty, \psi^*_\infty} (\mu_0, \nu_0) = \max_{\phi \in \Phi}\min_{\psi \in \Psi} ~ J_\infty^{\phi, \psi}(\mu_0, \nu_0)$.
Note that (i) follows from Step 2~\eqref{eqn:step2-inequality-e2m}; (ii) is a consequence of~\eqref{eqn:cor-apprx-err-bound} and inequalities (iii) and (iv) follow from Lemmas 5 (applied on game $\G_1$) and 6 (applied on game $\G_2$) in~\cite{guan2024zerosumgameslargepopulationteams} respectively.

\edit{By expanding the terms $\Residue$ and $\ResidueNew$ from~\cite{guan2024zerosumgameslargepopulationteams}, we can fully characterize the transfer bound across population sizes. 
In particular, for the $(\barN_1,\barN_2)$-population game, $\phi^*$ satisfies
\begin{align}
     \min_{\psi^{\barN_2} \in \Psi^{\barN_2}} 
        J^{\barN,\phi^*,\psi^{\barN_2}} (\bfx^{\barN_1}, \bfy^{\barN_2}) &\geq \underline{J}^{\barN*}(\bfx^{\barN_1}, \bfy^{\barN_2})\\
        &
        - \sum_{t=0}^TL_{J, t}\Big(\frac{|\X|}{2}\sqrt{\frac{1}{N_1}} + \frac{|\Y|}{2}\sqrt{\frac{1}{N_2}}\Big)  
        - \sum_{t=0}^TL_{J, t}\Big(\frac{|\X|}{2}\sqrt{\frac{1}{\barN_1}} + \frac{|\Y|}{2}\sqrt{\frac{1}{\barN_2}}\Big),
\end{align}
where $L_{J,t}$ denotes the Lipschitz constant of the time-$t$ value function.
Note that $\sum_{t=0}^TL_{J, t}\Big(\frac{|\X|}{2}\sqrt{\frac{1}{N_1}} + \frac{|\Y|}{2}\sqrt{\frac{1}{N_2}}\Big) = \Residue = \epsilon$, where $\underline{N} = \min(N_1, N_2)$ and similarly, we can define $\ResidueNew =\bar\epsilon$.
 Thus, the identical policy pair learned in the smaller finite-population game transfers to the larger game with a vanishing finite-population error, and the corresponding optimality gap is bounded by $\epsilon+\bar\epsilon$.
Combined with the fact that $\min(\barN_1, \barN_2) > \min(N_1, N_2)$, we obtain  $\epsilon$-optimality for the game $\mathcal{G}_{\barN}$.}
 \begin{remark}
     Theorem~\ref{thm:extend-to-many} applies to the \textit{exact optimizers} of the finite-population optimization restricted to identical team policies.
The suboptimality bound comprises two components: the first captures the approximation error of the identical policy $\phi^*$ trained in the $(N_1,N_2)$-population game  (Theorem~\ref{thm:finite-population-training}); the second captures the stochasticity mismatch between training and deployment due to differing population sizes.
Thus, for a fixed deployment population, increasing the training population size can improve the guarantee by reducing the first component, while training beyond the target size offers diminishing returns.
For simplicity, we assume matching initial empirical distributions. 
The result can be extended to mismatched initial distributions using a similar argument to that in the original mean-field paper~\citep{huang2006large}.
 \end{remark}
\end{proof}

\section{Proof of Proposition~\ref{prop:lipschitz-proposition}}
\mainpropositionlipschitz*
\begin{proof}
We prove the proposition in two parts.

\textbf{Step 1.} We first show that $\textstyle \left\|\nabla_{\eta} \log \phi(u \mid x, \mu, \nu)\right\|_2 \leq \frac{L_\phi}{2|\U|}$ implies
\begin{align}\label{eq:lcp-prop-temp}
    \sum_{u}|\log\phi_t\left(u|x, \hat\mu,\hat\nu\right) - 
    \log\phi_t\left(u|x, \hat\mu',\hat\nu'\right)| &\leq L_\phi\left(\dtv{\hat\mu, \hat\mu'} + \dtv{\hat\nu, \hat\nu'}\right) ~~\forall~~ x\in\X.
\end{align}
By the Definition of $\eta$, we can denote $\log\phi_t\left(u|x, \eta\right) \triangleq \log\phi_t\left(u|x, \hat\mu,\hat\nu\right) $.
By the Mean-Value Theorem and the bounded gradient from Step 1 we have for all $x\in\X$, $u\in\U$, $\eta\in\P(\X)\times\P(\Y)$,
\begin{align*}
    \vert\log\phi_t\left(u|x, \eta\right) - 
    \log\phi_t\left(u|x, \eta'\right)\vert &\leq \sup_{\eta\in \P(\X)\times\P(\Y)}\left\|\nabla_{\eta} \log \phi(u \mid x, \eta)\right\|_2\left\|\eta - \eta'\right\|_2\\
    &\leq \sup_{\eta\in \P(\X)\times\P(\Y)}\left\|\nabla_{\eta} \log \phi(u \mid x, \eta)\right\|_2\left\|\eta - \eta'\right\|_1\\
    &\leq \frac{L_\phi}{2|\U|}\cdot 2\left(\dtv{\hat\mu, \hat\mu'} + \dtv{\hat\nu, \hat\nu'}\right) \\
    &= \frac{L_\phi}{|\U|}\left(\dtv{\hat\mu, \hat\mu'} + \dtv{\hat\nu, \hat\nu'}\right).
\end{align*}
Summing over all actions for any given $x\in\X$,
\begin{align*}
    \sum_u\vert\log\phi_t\left(u|x, \eta\right) - 
    \log\phi_t\left(u|x, \eta'\right)\vert &\leq \sum_u\frac{L_\phi}{|\U|}\left(\dtv{\hat\mu, \hat\mu'} + \dtv{\hat\nu, \hat\nu'}\right)\\
    &\leq L_\phi\left(\dtv{\hat\mu, \hat\mu'} + \dtv{\hat\nu, \hat\nu'}\right) \sum_u\frac{1}{|\U|}\\
    &\leq L_\phi\left(\dtv{\hat\mu, \hat\mu'} + \dtv{\hat\nu, \hat\nu'}\right),
\end{align*}
completing the proof of Part 1.

\textbf{Step 2.} Using Step 1, we show that~\eqref{eq:lcp-prop-temp} implies
\begin{align}\label{eq:mf-lcp-alter}
    \sum_{u}|\phi_t\left(u|x, \hat\mu,\hat\nu\right) - 
    \phi_t\left(u|x, \hat\mu',\hat\nu'\right)|&\leq L_\phi\left(\dtv{\hat\mu, \hat\mu'} + \dtv{\hat\nu, \hat\nu'}\right) ~~\forall~~ x\in\X.
\end{align}
Define $h(z) \triangleq e^z$ and values $ a\triangleq\log\phi_t\left(u|x, \hat\mu',\hat\nu'\right)$ and $b\triangleq\log\phi_t\left(u|x, \hat\mu,\hat\nu\right)$. 
It follows from the Mean-Value Theorem that for all $x\in\X$ and $u\in\U$,
\begin{align*}
    \vert h(b)-h(a)\vert &= \vert h'(c)\vert\vert b-a\vert,  ~\min(a, b) \leq c \leq \max(a, b)\\
    \vert\phi_t\left(u|x, \hat\mu,\hat\nu\right)- \phi_t\left(u|x, \hat\mu',\hat\nu'\right)\vert&= e^c\vert\log\phi_t\left(u|x, \hat\mu,\hat\nu\right)-\log\phi_t\left(u|x, \hat\mu',\hat\nu'\right)\vert\\
     &\leq e^{\max(a, b)}\vert\log\phi_t\left(u|x, \hat\mu,\hat\nu\right)-\log\phi_t\left(u|x, \hat\mu',\hat\nu'\right)\vert.
\end{align*}
Note that $e^a = \phi_t\left(u|x, \hat\mu',\hat\nu'\right) \leq 1$ and $e^b = \phi_t\left(u|x, \hat\mu,\hat\nu\right) \leq 1$.
Consequently $e^{\max(a, b)}\leq 1$ and we have for all $x\in\X$,
\begin{align*}
    \vert\phi_t\left(u|x, \hat\mu,\hat\nu\right)- \phi_t\left(u|x, \hat\mu',\hat\nu'\right)\vert &\leq \vert\log\phi_t\left(u|x, \hat\mu,\hat\nu\right)-\log\phi_t\left(u|x, \hat\mu',\hat\nu'\right)\vert\\
     \sum_u\vert\phi_t\left(u|x, \hat\mu,\hat\nu\right)- \phi_t\left(u|x, \hat\mu',\hat\nu'\right)\vert &\leq \sum_u\vert\log\phi_t\left(u|x, \hat\mu,\hat\nu\right)-\log\phi_t\left(u|x, \hat\mu',\hat\nu'\right)\vert\\
     \sum_u\vert\phi_t\left(u|x, \hat\mu,\hat\nu\right)- \phi_t\left(u|x, \hat\mu',\hat\nu'\right)\vert &\leq L_\phi\left(\dtv{\hat\mu, \hat\mu'} + \dtv{\hat\nu, \hat\nu'}\right),
\end{align*}
where the last inequality follows from~\eqref{eq:lcp-prop-temp}
\end{proof}

\section{Proof of Proposition~\ref{thm:estimator-eval-dyn}}
\label{appdx-sec:estimator-thm}
\mainpropthree*
\begin{proof}[Sketch of Proof]
  MFs induced by identical team policies in an infinite-population game closely approximate the \ED{}s induced in the corresponding finite-population game.
  Combining this consequence with Lipschitz continuity in the dynamics, rewards (Assumption~\ref{assmpt:lipschitiz-model}) and policies (Proposition~\ref{prop:lipschitz-proposition}), we prove the statement.
\end{proof}

\begin{proof}
For the performance guarantee, 
\begin{align}\label{eqn:estimator-eval}
\Delta{J}(\phi^*_t, \psi^*_t) = 
\mathbb{E}_{{\phi^*, \psi^*}}
\Bigg[\Big|\sum_{t=0}^{T} r_t(\M^{N_1}_t, \N^{N_2}_t) - \sum_{t=0}^{T} r_t(\hat\M^{N_1}_t, \hat\N^{N_2}_t)\Big|  \Bigg],
\end{align}
we have,
\begin{align*}
\Delta{J}(\phi^*_t, \psi^*_t) = \mathbb{E}_{{\phi^*, \psi^*}}
    \Big[\left\vert\sum_{t=0}^{T} r_t(\M^{N_1}_t, \N^{N_2}_t) - \sum_{t=0}^{T}r_t(\hat\M^{N_1}_t, \hat\N^{N_2}_t) \right\vert \Big \vert \bfx_0^{N_1}=\hat\bfx_0^{N_1}, \bfy^{N_2}_0=\hat\bfy^{N_2}_0\! \Big].
\end{align*}

From Lipschitz rewards (Assumption~\ref{assmpt:lipschitiz-model}), we have,
\begin{align}\label{eq:simplify-equality-eval}
&\Delta{J}(\phi^*_t, \psi^*_t; \Gamma) \nonumber\\&= \mathbb{E}_{{\phi^*, \psi^*}}
    \Big[\left \vert\sum_{t=0}^{T} r_t(\M^{N_1}_t, \N^{N_2}_t) - \sum_{t=0}^{T}r_t(\hat\M^{N_1}_t, \hat\N^{N_2}_t)\right\vert\Big \vert \bfx_0^{N_1}=\hat\bfx_0^{N_1}, \bfy^{N_2}_0=\hat\bfy^{N_2}_0\! \Big]\nonumber\\
    &\leq \mathbb{E}_{{\phi^*, \psi^*}}
    \Big[\sum_{t=0}^{T} L_r\left(\dtv{\M^{N_1}_t, \hat\M^{N_1}_t}  + \dtv{\N^{N_2}_t, \hat\N^{N_2}_t}\right) \Big \vert \bfx_0^{N_1}=\hat\bfx_0^{N_1}, \bfy^{N_2}_0=\hat\bfy^{N_2}_0\! \Big]\nonumber\\
    &\leq L_r\sum_{t=0}^{T}\mathbb{E}_{{\phi^*, \psi^*}}
    \Big[ \dtv{\M^{N_1}_t, \hat\M^{N_1}_t}  + \dtv{\N^{N_2}_t, \hat\N^{N_2}_t} \Big \vert \bfx_0^{N_1}=\hat\bfx_0^{N_1}, \bfy^{N_2}_0=\hat\bfy^{N_2}_0\! \Big],
\end{align}
where the last inequality follows from the linearity of expectations.

\textbf{Step 1.} We first show the following for all time steps $t$, 
\begin{align*}
    &\mathbb{E}_{{\phi^*, \psi^*}}
    \Big[ \dtv{\M^{N_1}_{t+1}, \hat\M^{N_1}_{t+1}}  + \dtv{\N^{N_2}_{t+1}, \hat\N^{N_2}_{t+1}} \Big \vert \bfx_t^{N_1}, \hat\bfx_t^{N_1}, \bfy^{N_2}_t,\hat\bfy^{N_2}_t\! \Big]\\ &\leq \kappa_1 + \kappa_2\left (\dtv{\M^{N_1}_{t}, \hat\M^{N_1}_{t}}  + \dtv{\N^{N_2}_{t}, \hat\N^{N_2}_{t}}\right),
\end{align*}
where $\kappa_1 \triangleq  \frac{1}{2}(L_\phi + L_\psi)\epsilon + \Residue $ and $\kappa_2 \triangleq (1+L_f + \frac{1}{2}L_\phi + \frac{1}{2}L_\psi)$.

By the linearity property of expectations,
\begin{align*}
    &\mathbb{E}_{{\phi^*, \psi^*}}
    \Big[ \dtv{\M^{N_1}_{t+1}, \hat\M^{N_1}_{t+1}}  + \dtv{\N^{N_2}_{t+1}, \hat\N^{N_2}_{t+1}} \Big \vert \bfx_t^{N_1}, \hat\bfx_t^{N_1}, \bfy^{N_2}_t,\hat\bfy^{N_2}_t\! \Big] \\
    &=\mathbb{E}_{{\phi^*}}
    \Big[ \dtv{\M^{N_1}_{t+1}, \hat\M^{N_1}_{t+1}}  \Big \vert \bfx_t^{N_1}, \hat\bfx_t^{N_1}, \bfy^{N_2}_t,\hat\bfy^{N_2}_t\! \Big] 
    + \mathbb{E}_{{\psi^*}}
    \Big[ \dtv{\N^{N_2}_{t+1}, \hat\N^{N_2}_{t+1}} \Big \vert \bfx_t^{N_1}, \hat\bfx_t^{N_1}, \bfy^{N_2}_t,\hat\bfy^{N_2}_t\! \Big]
\end{align*}
We now bound the Blue team's mean-field, i.e., the first term on the RHS. 
The Red team can be bounded similarly (second term).
\begin{align}\label{eq:blue-splits-estimator-eval}
    &\mathbb{E}_{{\phi^*}}
    \Big[ \dtv{\M^{N_1}_{t+1}, \hat\M^{N_1}_{t+1}}  \Big \vert \bfx_t^{N_1}, \hat\bfx_t^{N_1}, \bfy^{N_2}_t,\hat\bfy^{N_2}_t\! \Big]  \nonumber\\
    &\leq\mathbb{E}_{{\phi^*}}
    \Big[ \dtv{\M^{N_1}_{t+1}, \M_{t+1}} \Big \vert \bfx_t^{N_1}, \hat\bfx_t^{N_1}, \bfy^{N_2}_t,\hat\bfy^{N_2}_t\! \Big]\nonumber\\
    &+ \mathbb{E}_{{\phi^*}}
    \Big[ \dtv{\M_{t+1}, \hat\M_{t+1}} \Big \vert \bfx_t^{N_1}, \hat\bfx_t^{N_1}, \bfy^{N_2}_t,\hat\bfy^{N_2}_t\! \Big]\nonumber\\
    &+ \mathbb{E}_{{\phi^*}}
    \Big[ \dtv{\hat\M^{N_1}_{t+1}, \hat\M_{t+1}} \Big \vert \bfx_t^{N_1}, \hat\bfx_t^{N_1}, \bfy^{N_2}_t,\hat\bfy^{N_2}_t\! \Big],
\end{align}
where $\M_{t+1} = \M^{N_1}_{t}F(\M^{N_1}_{t}, \N^{N_2}_{t}, \phi^*)$ and $\hat\M_{t+1} = \hat\M^{N_1}_{t}F(\hat\M^{N_1}_{t}, \hat\N^{N_2}_{t}, \phi^*, \Gamma_{\textrm{D-PC}})$ are shorthand notations for the next induced mean-field from the \textit{infinite}-population \textit{deterministic} dynamics~\citep{guan2024zero} with and without the estimator $\Gamma_{\textrm{D-PC}}$. 
From Corollary~\ref{cor:mf-aprx} in~\cite{guan2024zero}, we have 
\begin{align*}
    \mathbb{E}_{{\phi^*}}
    \Big[ \dtv{\M^{N_1}_{t+1}, \M_{t+1}} \Big \vert \bfx_t^{N_1}, \hat\bfx_t^{N_1}, \bfy^{N_2}_t,\hat\bfy^{N_2}_t\! \Big] &\leq \Residue, \\
    \mathbb{E}_{{\phi^*}}
    \Big[ \dtv{\hat\M^{N_1}_{t+1}, \hat\M_{t+1}} \Big \vert \bfx_t^{N_1}, \hat\bfx_t^{N_1}, \bfy^{N_2}_t,\hat\bfy^{N_2}_t\! \Big]&\leq \Residue.
\end{align*}
Thus, we are left to simplify 

\begin{align*}
    &\mathbb{E}_{{\phi^*}}\Big[ \dtv{\M_{t+1}, \hat\M_{t+1}} \Big \vert \bfx_t^{N_1}, \hat\bfx_t^{N_1}, \bfy^{N_2}_t,\hat\bfy^{N_2}_t\! \Big] \\&= \dtv{\M^{N_1}_{t}F(\M^{N_1}_{t}, \N^{N_2}_{t}, \phi^*),  \hat\M^{N_1}_{t}F(\hat\M^{N_1}_{t}, \hat\N^{N_2}_{t}, \phi^*, \Gamma_{\textrm{D-PC}})},
\end{align*}
where the expectation vanishes due to \textit{deterministic} transitions under identical policies $(\phi^*, \psi^*)$.
Now, 
\begin{align}\label{eq:I1-I2}
    &2\dtv{\M^{N_1}_{t}F(\M^{N_1}_{t}, \N^{N_2}_{t}, \phi^*),  \hat\M^{N_1}_{t}F(\hat\M^{N_1}_{t}, \hat\N^{N_2}_{t}, \phi^*, \Gamma_{\textrm{D-PC}})} \\
    &=\sum_{x'\in\X}\Big\vert \sum_{x\in \X}\sum_{u \in \U} f_t(x'|x, u, \M^{N_1}_{t}, \N^{N_2}_{t})\phi^*_t(u|x,  \M^{N_1}_{t}, \N^{N_2}_{t}) \M^{N_1}_{t}(x) \nonumber \\&-\sum_{x \in \X} \sum_{u \in \U} f_t(x'|x, u,\hat\M^{N_1}_{t}, \hat\N^{N_2}_{t})\phi^*_t(u|x, \hat\M^{N_1}_{t}, \hat\N^{N_2}_{x, t}) \hat\M^{N_1}_{t}(x)\Big\vert \nonumber \\
    &\leq \sum_{x'\in\X}\sum_{x\in \X}\sum_{u \in \U} \Big\vert  f_t(x'|x, u, \M^{N_1}_{t}, \N^{N_2}_{t})\phi^*_t(u|x,  \M^{N_1}_{t}, \N^{N_2}_{t}) \M^{N_1}_{t}(x) \nonumber \\&- f_t(x'|x, u,\hat\M^{N_1}_{t}, \hat\N^{N_2}_{t})\phi^*_t(u|x, \hat\M^{N_1}_{t}, \hat\N^{N_2}_{x, t}) \hat\M^{N_1}_{t}(x)\Big\vert\nonumber .
\end{align}
Firstly, we add and subtract $f_t(x'|x, u,\hat\M^{N_1}_{t}, \hat\N^{N_2}_{t})\phi^*_t(u|x, \hat\M^{N_1}_{t}, \hat\N^{N_2}_{t}) \hat\M^{N_1}_{t}(x)$ to split absolute value as
\begin{align*}
     2\dtv{\M^{N_1}_{t}F(\M^{N_1}_{t}, \N^{N_2}_{t}, \phi^*),  \hat\M^{N_1}_{t}F(\hat\M^{N_1}_{t}, \hat\N^{N_2}_{t}, \phi^*, \Gamma_{\textrm{D-PC}})} \leq \sum_{x'\in\X}\sum_{x\in \X}\sum_{u \in \U}(I_1+ I_2),\\
\end{align*}
where, 
\begin{align*}
    I_1 &\triangleq \Big\vert  f_t(x'|x, u, \M^{N_1}_{t}, \N^{N_2}_{t})\phi^*_t(u|x,  \M^{N_1}_{t}, \N^{N_2}_{t}) \M^{N_1}_{t}(x) -f_t(x'|x, u,\hat\M^{N_1}_{t}, \hat\N^{N_2}_{t})\phi^*_t(u|x, \hat\M^{N_1}_{t}, \hat\N^{N_2}_{t}) \hat\M^{N_1}_{t}(x)\Big\vert \\
    I_2 &\triangleq \Big\vert {f_t(x'|x, u,\hat\M^{N_1}_{t}, \hat\N^{N_2}_{t})\phi^*_t(u|x, \hat\M^{N_1}_{t}, \hat\N^{N_2}_{t}) \hat\M^{N_1}_{t}(x)  -f_t(x'|x, u,\hat\M^{N_1}_{t}, \hat\N^{N_2}_{t})\phi^*_t(u|x, \hat\M^{N_1}_{t}, \hat\N^{N_2}_{x, t}) \hat\M^{N_1}_{t}(x)\Big\vert}.
\end{align*}
Similar to \eqref{eq:mf-lcp-alter}, we can simplify $I_2$ using  Proposition~\ref{prop:lipschitz-proposition} and obtain

\begin{align*}
    \sum_{x'\in\X}\sum_{x\in \X}\sum_{u \in \U}I_2 < L_\phi\epsilon.
\end{align*}

We now individually bound the terms in $I_1$ using the definition of total variation distance, the identity
\begin{align*}
|abc -a'b'c'| \leq |a-a'|bc + a'|b-b'|c +a'b'|c-c'|, ~a,b,c\geq 0,
\end{align*}
Assumption~\ref{assmpt:lipschitiz-model} and Proposition~\ref{prop:lipschitz-proposition} as follows:

\begin{align*}
    \sum_{x'\in\X}\sum_{x\in \X}\sum_{u \in \U}I_1\leq 2\dtv{\M^{N_1}_{t}, \hat \M^{N_1}_{t}} + (L_\phi+L_f)\left(\dtv{\M^{N_1}_{t}, \hat \M^{N_1}_{t}} + \dtv{\N^{N_2}_{t}, \hat \N^{N_2}_{t}}\right) 
\end{align*}
%

Thus we can substitute the expressions for $I_1$ and $I_2$ back in \eqref{eq:I1-I2} and rewrite~\eqref{eq:blue-splits-estimator-eval} to get 
\begin{align}\label{eq:blue-step-1-final}
     &\mathbb{E}_{{\phi^*}}
    \Big[ \dtv{\M^{N_1}_{t+1}, \hat\M^{N_1}_{t+1}}  \Big \vert \bfx_t^{N_1}, \hat\bfx_t^{N_1}, \bfy^{N_2}_t,\hat\bfy^{N_2}_t\! \Big] \nonumber \\
    &<\dtv{\M^{N_1}_{t}, \hat \M^{N_1}_{t}}\nonumber  \\&+ \frac{1}{2}(L_\phi+L_f)\left(\dtv{\M^{N_1}_{t}, \hat \M^{N_1}_{t}} + \dtv{\N^{N_2}_{t}, \hat \N^{N_2}_{t}}\right) + \frac{1}{2}L_\phi\epsilon + \Residue
\end{align}
By symmetry,
\begin{align}\label{eq:red-step-1-final}
     &\mathbb{E}_{{\psi^*}}
    \Big[ \dtv{\N^{N_2}_{t+1}, \hat\N^{N_2}_{t+1}}  \Big \vert \bfx_t^{N_1}, \hat\bfx_t^{N_1}, \bfy^{N_2}_t,\hat\bfy^{N_2}_t\! \Big]\nonumber  \\
    &<\dtv{\N^{N_2}_{t}, \hat \N^{N_2}_{t}}\nonumber  \\&+ \frac{1}{2}(L_\psi+L_f)\left(\dtv{\M^{N_1}_{t}, \hat \M^{N_1}_{t}} + \dtv{\N^{N_2}_{t}, \hat \N^{N_2}_{t}}\right) + \frac{1}{2}L_\psi\epsilon + \Residue
\end{align}
Adding \eqref{eq:blue-step-1-final} and \eqref{eq:red-step-1-final} we obtain,
\begin{align*}
    &\mathbb{E}_{{\phi^*, \psi^*}}
    \Big[ \dtv{\M^{N_1}_{t+1}, \hat\M^{N_1}_{t+1}}  + \dtv{\N^{N_2}_{t+1}, \hat\N^{N_2}_{t+1}} \Big \vert \bfx_t^{N_1}, \hat\bfx_t^{N_1}, \bfy^{N_2}_t,\hat\bfy^{N_2}_t\! \Big] \\
    &< (1+L_f + \frac{1}{2}L_\phi + \frac{1}{2}L_\psi)\left(\dtv{\M^{N_1}_{t}, \hat \M^{N_1}_{t}} + \dtv{\N^{N_2}_{t}, \hat \N^{N_2}_{t}}\right) + \frac{1}{2}(L_\phi + L_\psi)\epsilon + \Residue\\
    &= \kappa_1 + \kappa_2\left(\dtv{\M^{N_1}_{t}(x), \hat \M^{N_1}_{t}(x)} + \dtv{\N^{N_2}_{t}, \hat \N^{N_2}_{t}}\right).
\end{align*}

\textbf{Step 2.} Define
\begin{align*}
    a_t \triangleq \mathbb{E}_{{\phi^*, \psi^*}}
    \Big[ \dtv{\M^{N_1}_{t}, \hat\M^{N_1}_{t}}  + \dtv{\N^{N_2}_{t}, \hat\N^{N_2}_{t}} \Big \vert \bfx_0^{N_1}=\hat\bfx_0^{N_1}, \bfy^{N_2}_0=\hat\bfy^{N_2}_0\! \Big]
\end{align*}
At $t=0$, $\M^{N_1}_{0} = \hat\M^{N_1}_{0} ~~\text{and}~~ \N^{N_2}_{0}= \hat\N^{N_2}_{0}$ as we begin the fully observable and estimated scenarios under the same initial joint states, i.e., $\bfx_0^{N_1}=\hat\bfx_0^{N_1}, \bfy^{N_2}_0=\hat\bfy^{N_2}_0$.
Thus, $a_0 = 0$.
We proceed to show that $a_{t+1} \leq \kappa_1 + \kappa_2a_t$ where $\kappa_1$ and $\kappa_2$ are the same as Step 1.
By the law of iterated expectations, 
\begin{align*}
    &\mathbb{E}_{{\phi^*, \psi^*}}
    \Big[ \dtv{\M^{N_1}_{t+1}, \hat\M^{N_1}_{t+1}}  + \dtv{\N^{N_2}_{t+1}, \hat\N^{N_2}_{t+1}} \Big \vert \bfx_0^{N_1}=\hat\bfx_0^{N_1}, \bfy^{N_2}_0=\hat\bfy^{N_2}_0\! \Big] 
    \\&= \mathbb{E}_{{\phi^*, \psi^*}}
    \left[\Big[ \mathbb{E}_{{\phi^*, \psi^*}}\dtv{\M^{N_1}_{t+1}, \hat\M^{N_1}_{t+1}}  + \dtv{\N^{N_2}_{t+1}, \hat\N^{N_2}_{t+1}} \Big \vert \bfx_t^{N_1}, \hat\bfx_t^{N_1}, \bfy^{N_2}_t,\hat\bfy^{N_2}_t\! \Big]\Big \vert \bfx_0^{N_1}=\hat\bfx_0^{N_1}, \bfy^{N_2}_0=\hat\bfy^{N_2}_0\! \right]\\
    &< \mathbb{E}_{{\phi^*, \psi^*}}
    \left[\kappa_1 + \kappa_2\left (\dtv{\M^{N_1}_{t}, \hat\M^{N_1}_{t}}  + \dtv{\N^{N_2}_{t}, \hat\N^{N_2}_{t}}\right) \vert \bfx_0^{N_1}=\hat\bfx_0^{N_1}, \bfy^{N_2}_0=\hat\bfy^{N_2}_0\! \right]\\
    &= \kappa_1 + \kappa_2\mathbb{E}_{{\phi^*, \psi^*}}
    \left[\dtv{\M^{N_1}_{t}, \hat\M^{N_1}_{t}}  + \dtv{\N^{N_2}_{t}, \hat\N^{N_2}_{t}} \vert \bfx_0^{N_1}=\hat\bfx_0^{N_1}, \bfy^{N_2}_0=\hat\bfy^{N_2}_0\! \right] \\
    &= \kappa_1 + \kappa_2 a_t,
\end{align*}
where we use the result obtained from Step 1 along with the linearity property of expectations.
Using this new notation,~\eqref{eq:simplify-equality-eval} simplifies to
\begin{align*}
    \Delta{J}(\phi^*_t, \psi^*_t) \leq L_r\sum_{t=0}^{T}a_t
\end{align*}
We have $a_0 = 0$, $a_1 < \kappa_1$, $a_2 < \kappa_1 + \kappa_2a_1$ or $a_2 < \kappa_1(\kappa_2 + 1)$, $a_3 < \kappa_2a_2 + \kappa_1$ or $a_3 < \kappa_1(\kappa_2^2 + \kappa_2 + 1) $ and so on.
This can be written compactly for all $t>0$ as 
\begin{align*}
    a_{t} < \kappa_1\sum_{\tau=0}^{t-1}\kappa_2^\tau = \kappa_1 \frac{\kappa_2^t-1}{\kappa_2-1},
\end{align*}
as $\kappa_2 =  (1+L_f + \frac{1}{2}L_\phi + \frac{1}{2}L_\psi) \neq 1$.
This is a geometric sum and thus we have
\begin{align*}
    \Delta{J}(\phi^*_t, \psi^*_t)&\leq L_r\sum_{t=0}^{T}a_t\\
    &=L_r\sum_{t=1}^{T}a_t\\
    &<  L_r\sum_{t=1}^{T} \kappa_1 \frac{\kappa_2^t-1}{\kappa_2-1}\\
    &=   \frac{L_r\kappa_1}{\kappa_2-1}\sum_{t=1}^{T}\kappa_2^t-1\\
    &=  \frac{L_r\kappa_1}{\kappa_2-1} \left(\kappa_2\frac{\kappa_2^T - 1}{\kappa_2 - 1} - T\right)\\
    &= \underbrace{\frac{1}{2}\frac{L_r(L_\phi + L_\psi)}{\kappa_2-1} \left(\kappa_2\frac{\kappa_2^T - 1}{\kappa_2 - 1} - T\right)}_{K}\epsilon + \Residue\frac{L_r}{\kappa_2-1} \left(\kappa_2\frac{\kappa_2^T - 1}{\kappa_2 - 1} - T\right)\\
    &= K\epsilon + \Residue.
\end{align*}
\end{proof}

\section{RPS and \texorpdfstring{\lowercase{c}}{c}RPS Setup}\label{AppendixB}

\subsection{State Space}
We have three states in this representation of the game: rock, paper and scissors. 
We denote this state space as $\mathcal{S} =  \{\texttt{R,P,S}\}$. 
The empirical distribution of the Blue team is denoted by ${\mu} \in \mathcal{P(S)}$ and that of the Red team is denoted by ${\nu} \in \mathcal{P(S)}$. 
Since we have three states for each team, both EDs lie in a three-dimensional simplex denoted by $\mathcal{P(S)}$.

\subsection{Action Space}
\subsubsection{RPS}
At each state, we define three actions denoted by $\mathcal{A} = \{\texttt{CW}, \texttt{CCW}, \texttt{Stay}\}$. 
These actions represent the ability of the agents to move from one state to the other in the following fashion:
\begin{enumerate}
    \item \texttt{CW} denotes a clockwise cyclic action from one state to the other, i.e., from $\texttt{R} \rightarrow \texttt{P}$, $\texttt{P} \rightarrow \texttt{S}$, $\texttt{S} \rightarrow \texttt{R}$.
    \item \texttt{CCW} denotes a counterclockwise cyclic movement, i.e., from $\texttt{R} \rightarrow \texttt{S}$, $\texttt{S} \rightarrow \texttt{P}$, $\texttt{P} \rightarrow \texttt{R}$.
    \item $\texttt{Stay}$ denotes the idle action (remain in the same state as before).
\end{enumerate}

\subsubsection{cRPS}
At each state we have two actions denoted by $\mathcal{A} = \{\texttt{CW},  \texttt{Stay}\}$. 
These actions represent the ability of the agents to move from one state to the other in the following fashion:
\begin{enumerate}
    \item \texttt{CW} denotes a clockwise cyclic action from one state to the other, i.e., from $\texttt{R} \rightarrow \texttt{P}$, $\texttt{P} \rightarrow \texttt{S}$, $\texttt{S} \rightarrow \texttt{R}$.
    \item $\texttt{Stay}$ denotes the idle action (remain in the same state as before).
\end{enumerate}
Thus, we cannot directly jump from \texttt{R}  to \texttt{S} within a single step, but must go via \texttt{P}. 
Mathematically,  \texttt{S} does not lie in the reachable set of \texttt{R}. 
The reachable set $\mathcal{R}(s)$ for each state $s$ at a given time step under this modified action space is as follows
\begin{align*}
    \mathcal{R}(\texttt{R}) = \{\texttt{R}, \texttt{P}\}, \quad
    \mathcal{R}(\texttt{P}) = \{\texttt{P}, \texttt{S}\}, \quad
    \mathcal{R}(\texttt{S}) = \{\texttt{S}, \texttt{R}\}. 
\end{align*}

The state-action space of both RPS and cRPS are presented in Figure~\ref{fig:rps-crps-state-space}.

\begin{figure}[b]
    \centering
    \includegraphics[width=0.6\linewidth]{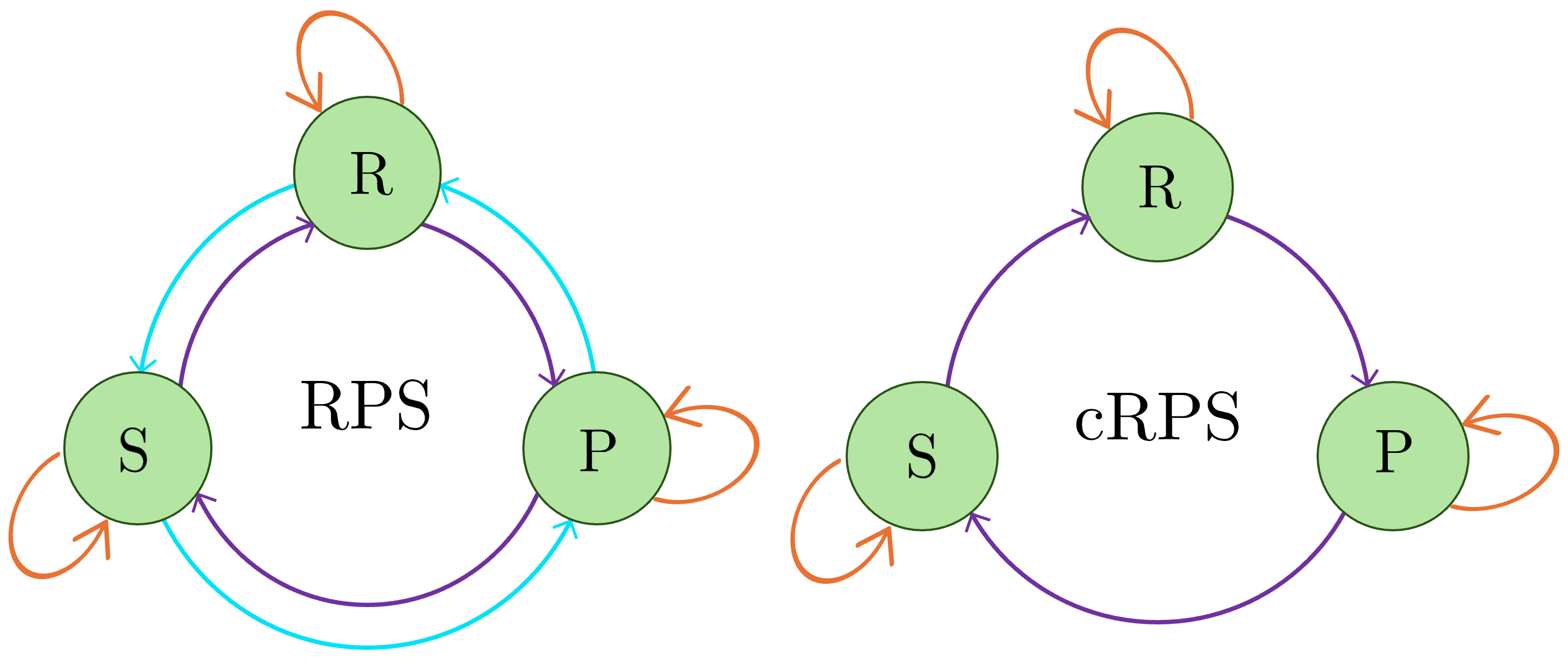}
    \caption{State-action spaces for RPS and cRPS.}
    \label{fig:rps-crps-state-space}
\end{figure}

\subsection{Dynamics and Transition Probabilities}

For both RPS and cRPS, we consider deterministic transitions $\mathcal{T}(s, a, s')$, which implies that given a state-action pair $(s, a)$, the agent reaches a unique next state $s'$ with certainty (no distribution over the reachable states). 
Thus, for state $\texttt{R}$ and action $\texttt{Stay}$, the transition function $\mathcal{T}:\mathcal{S}\times\mathcal{A}\times \mathcal{S}\rightarrow \{0, 1\}$ is given as: 
\begin{align*}
    \mathcal{T}(\texttt{R}, \texttt{Stay}, \texttt{R}) = 1,\quad
    \mathcal{T}(\texttt{R}, \texttt{Stay}, \texttt{P}) = 0,\quad
    \mathcal{T}(\texttt{R}, \texttt{Stay}, \texttt{S}) = 0.
\end{align*}
This implies that an agent in the state \texttt{R}  upon taking the \texttt{Stay}  action remains in state \texttt{R}. 
Similarly, 
\begin{align*}
    \mathcal{T}(\texttt{R}, \texttt{CW}, \texttt{R}) = 0,\quad
    \mathcal{T}(\texttt{R}, \texttt{CW}, \texttt{P}) = 1,\quad
    \mathcal{T}(\texttt{R}, \texttt{CW}, \texttt{S}) = 0.
\end{align*}
represent the method to transition from \texttt{R}  to \texttt{P} .

\subsection{Reward Structure}

In a two player RPS game, the reward matrix for Player 1 is defined as:
\begin{align*}
    A = \begin{bmatrix}
        \begin{array}{rrr}
         0 & -1 &  1 \\
         1 &  0 & -1 \\
        -1 &  1 &  0
        \end{array}
    \end{bmatrix}.
\end{align*}
We extend this two-player framework to the multi-agent team game formulation. 
Define the pairwise reward for an agent at state $x\in \mathcal{S}$ within the Blue team and at state $y \in \mathcal{S}$ from the Red team as 
\begin{equation*}
    r(x, y) \triangleq A_{xy},
\end{equation*}
where $A_{xy}$ represents the element from the reward matrix $A$ corresponding to the states $x$ (row player) and $y$ (column player). 
In lieu of the zero-sum structure, the reward for the agent at $y$ with respect to $x$ becomes $-r(x, y)$. 
Thus, for each player $x_i \in \mathcal{S} \text{ and } i=1, 2,\ldots, N_1$ in the Blue team and $y_j  \in \mathcal{S}  \text{ and } j=1, 2,\ldots, N_2$ in the Red team, the reward for the Blue team can be defined as
\begin{equation*}
R_\mathrm{Blue}(\mathbf{x}, \mathbf{y}) = \frac{1}{N_1} \sum_{i=1}^{N_1} \underbrace{\Big[\frac{1}{N_2} \sum_{j=1}^{N_2} r(x_i, y_j) \Big]}_{\text{reward for agent $i$}}.
\end{equation*}

Rewriting the term inside the square brackets as 
\begin{align}\label{eq:reward-simplification-rps}
    \frac{1}{N_2}\sum_{j=1}^{N_2} r(x_i, y_j) &=\frac{1}{N_2}\sum_{y\in \mathcal{S}} \sum_{j=1}^{N_2} r(x_i, y) \mathbf{1}_{y_j=y}\nonumber \\
    &= A_{x_is_0}\sum_{j=1}^{N_2}\frac{1}{N_2}\mathbf{1}_{y_j=s_0} + A_{x_is_1}\sum_{j=1}^{N_2}\frac{1}{N_2}\mathbf{1}_{y_j=s_1} \nonumber\\ &+ A_{x_is_2}\sum_{j=1}^{N_2}\frac{1}{N_2}\mathbf{1}_{y_j=s_2}\nonumber\\
    &= A_{x_is_0}{\nu}(s_0)+A_{x_is_1}{\nu}(s_1)+A_{x_is_2}{\nu}(s_2)\nonumber\\
    &= A(x_i){\nu},
\end{align}
where $A(x_i)$ is the row of the reward matrix corresponding to state $x_i$. 
Using \eqref{eq:reward-simplification-rps}, the total Blue reward can be expressed as 
\begin{align*}
    R_\mathrm{Blue}(\mathbf{x}, \mathbf{y}) &= \frac{1}{N_1} \sum_{i=1}^{N_1}  A(x_i){\nu}\\
    &= \Bigl( \frac{1}{N_1}\sum_{x\in \mathcal{S}} \sum_{i=1}^{N_1} A(x_i) \mathbf{1}_{x_i=x} \Bigr){\nu}\\
    &= \Bigl(A(s_0){\mu}(s_0)+A(s_1){\mu}(s_1)+A(s_2){\mu}(s_2)\Bigr){\nu}\\
    &= {\mu}\t A{\nu}.
\end{align*}

\begin{restatable}{proposition}{mainproposition}
    \label{prop:two-step-convergence}
    With initial conditions ${\mu}_{0} = [1, 0, 0]\t$ (equivalent to $[\mu(\texttt{R)}, \mu(\texttt{P)}, \mu(\texttt{S)}]$) and ${\nu}_{0} = [0, 1, 0]\t$, all mean-field optimal trajectories satisfy ${\mu}_t^* = {\nu}_t^* = [\frac{1}{3}, \frac{1}{3}, \frac{1}{3}]\t$ for all $t\geq2$, and ${\mu}^*_{1} = [0, 1-\delta, \delta]\t$ where $\delta \in [\frac{1}{3}, 
    \frac{2}{3}]$ and ${\nu}^*_{1} = [0, \frac{2}{3}, \frac{1}{3}]\t$.
    Furthermore, the unique game value is given by~$-1/3$.
\end{restatable} 

\begin{proof}
    For the constrained RPS game under the stated initial condition, we cannot obtain the target distribution $\begin{bmatrix} \frac{1}{3} & \frac{1}{3} & \frac{1}{3} \end{bmatrix}\t$ after a single time step but this may be possible for $t\geq 2$. 
To this end, consider the following candidate trajectory respecting the transition dynamics:
\[
\begin{aligned}
    {\mu}_0\t &= \begin{bmatrix} 1  & 0 & 0 \end{bmatrix}, \\
    {\mu}_1\t &= \begin{bmatrix} 1 - x_1 & x_1 & 0 \end{bmatrix}, \\
    {\mu}_2\t &= \begin{bmatrix} 1 - x_1 - x_2 & x_1 + x_2 - x_3 & x_3 \end{bmatrix}, \\
    {\mu}_t\t &= \begin{bmatrix} \frac{1}{3} & \frac{1}{3} & \frac{1}{3} \end{bmatrix}, \quad \forall t > 2.
\end{aligned}
\]
\[
\begin{aligned}
    {\nu}_0\t &= \begin{bmatrix} 0 & 1 & 0 \end{bmatrix}, \\
    {\nu}_1\t &= \begin{bmatrix} 0 & 1-y_1 & y_1 \end{bmatrix}, \\
    {\nu}_2\t &= \begin{bmatrix} y_3 & 1-y_1-y_2 & y_1+y_2-y_3 \end{bmatrix}, \\
    {\nu}_t\t &= \begin{bmatrix} \frac{1}{3} & \frac{1}{3} & \frac{1}{3} \end{bmatrix}, \quad \forall t > 2.
\end{aligned}
\]
In order to respect the simplex structure for ${\mu}_t$ and ${\nu}_t$, we have the following constraints at all times:
\[
\begin{aligned}
    0 \leq x_1, x_2, x_3 \leq 1,\quad
    x_2 \leq 1-x_1, \quad
    x_3 \leq x_1.
\end{aligned}
\]
Similarly,
\[
\begin{aligned}
    0 \leq y_1, y_2, y_3 \leq 1,\quad
    y_2 \leq 1-y_1, \quad
    y_3 \leq y_1.
\end{aligned}
\]
For the distribution at $t = 2$ to be $\begin{bmatrix} \frac{1}{3} & \frac{1}{3} & \frac{1}{3} \end{bmatrix}\t$ for both teams, we get the additional constraints
\[
\begin{aligned}
    x_3 = y_3 = \frac{1}{3},\\
    x_1+x_2-x_3 = y_1+y_2-y_3 = \frac{1}{3}, \\
    \Rightarrow x_1+x_2 = y_1+y_2 = \frac{2}{3},
\end{aligned}
\]
which implies that
\[
x_1, y_1 \leq \frac{2}{3},
\]
since $x_2, y_2 \geq 0$.
The constraints now take the form

\begin{align}
    \frac{1}{3} &\leq x_1 \leq \frac{2}{3},
\label{eq:c1x1}
\end{align}
and similarly,
\begin{align}
    \frac{1}{3} &\leq y_1 \leq \frac{2}{3}.
    \label{eq:c1y1}
\end{align}
The objective function for cRPS is given by
\begin{align}\label{eq:crps-obj}
    & J^{N, \phi, \psi} \big(\mu_0, \nu_0 \big) =
    \mathbb{E}_{{\phi, \psi}} \!
    \Big[\sum_{t=1}^{T}\mu_t\t A\nu_t \Big \vert \mu_0, \nu_0\Big],
\end{align}
which leads to the optimization problem
\begin{equation}
\begin{array}{rrclcl}
\displaystyle \max_{{\phi}^t}\min_{{\psi}^t} & \multicolumn{3}{l}{J^{N_1, N_2, {\phi}^t, {\psi}^t}({\mu}_0, {\nu}_0) = {\mu}_1\t A{\nu}_1}.
\end{array}
\label{eq: crps-ss-obj}
\end{equation}
Substituting  $\mu_1\t = [1-x_1, x_1, 0]$ and $\nu_1\t = [0, 1-y_1, y_1]$ results in the following expression for the maximizing Blue team:
\begin{align*}
\max_{{\phi}^t} \quad J^{N_1, N_2, {\phi}^t, {\psi}^t}({\mu}_0, {\nu}_0) &= x_1 + 2y_1 - 3x_1y_1 - 1 \nonumber \\
& = x_1(1-3y_1) + (2y_1-1).
\end{align*}

Since this equation is linear in $x_1$, the solution to the maximization problem subject to the constraint \eqref{eq:c1x1}
is
\begin{align}
    x_1 = \frac{1}{3}, \quad y_1>\frac{1}{3} \label{eq:soln1x1},\\
    x_1 = \frac{2}{3}, \quad y_1<\frac{1}{3}\label{eq:soln2x1},\\
    x_1 \in [\frac{1}{3}, 
    \frac{2}{3}], \quad y_1 = \frac{1}{3}\label{eq:soln3x1}.
\end{align}

Following the same approach for the minimizing Red team, we get the following objective,
\begin{align*}
\min_{{\psi}^t} \quad J^{N_1, N_2, {\phi}^t, {\psi}^t}({\mu}_0, {\nu}_0) &= x_1 + 2y_1 - 3x_1y_1 - 1 \nonumber \\
& = y_1(2-3x_1) + (x_1-1),
\end{align*}
subject to the constraint \eqref{eq:c1y1}, with the solution being:
\begin{align}
    y_1 = \frac{1}{3}, \quad x_1<\frac{2}{3}\label{eq:soln1y1},\\
    y_1 = \frac{2}{3}, \quad x_1>\frac{2}{3}\label{eq:soln2y1},\\
    y_1 \in [\frac{1}{3}, 
    \frac{2}{3}], \quad x_1 = \frac{2}{3}\label{eq:soln3y1}.
\end{align}

The Constraint \eqref{eq:c1x1} ensures that \eqref{eq:soln2y1} cannot hold, while constraint \eqref{eq:c1y1} similarly prevents \eqref{eq:soln2x1} from holding.

Consider now the case when  $y_1 > \frac{1}{3}$. 
From \eqref{eq:soln1x1} it follows that $x_1 = \frac{1}{3}$. 
Conversely, if the Blue team commits to a distribution with $x_1 = \frac{1}{3}$, the Red team's best response given by \eqref{eq:soln1y1} gives $y_1 = \frac{1}{3}$, resulting in an incentive for the Red team to deviate from $y_1 > \frac{1}{3}$.
Thus, \eqref{eq:soln1x1} does not constitute an optimal solution.
Following a similar argument, it can be shown that \eqref{eq:soln3y1} is not an optimal solution either, as illustrated below.

Assume that $x_1 = \frac{2}{3}$. From \eqref{eq:soln3y1}, $y_1 \in [\frac{1}{3}, 
    \frac{2}{3}]$. 
Now, if the Red team announces that it will deploy the distribution $y_1 \in [\frac{1}{3}, 
    \frac{2}{3}]$, the Blue team's response for $x_1$ follows from \eqref{eq:soln1x1} and \eqref{eq:soln3x1}. 
    We have already established that \eqref{eq:soln1x1} is not an optimal solution. 
    This implies that $x_1 \in [\frac{1}{3}, 
    \frac{2}{3}]$ can be a possible response to the Red team. 
    However, it violates \eqref{eq:soln3y1}, where $x_1 = \frac{2}{3}$ follows from strict equality.
    Thus, \eqref{eq:soln3y1} does not constitute an optimal solution as the Blue team has an incentive to deviate.

Now, suppose the Blue team announces a distribution where $ x_1 \in [\frac{1}{3}, 
    \frac{2}{3}]$.  
    In this case, the Red team's optimal response, derived from \eqref{eq:soln1y1} and \eqref{eq:soln3x1}, is $y_1 = \frac{1}{3}$.
    Conversely, if the Red team announces that its distribution will be $y_1 = \frac{1}{3}$, the Blue team will still follow $ x_1 \in [\frac{1}{3}, 
    \frac{2}{3}]$. 
     Since neither team has an incentive to deviate from these distributions, they form an optimal trajectory.
    Thus, the solution to the bilinear optimization problem for two-time step convergence takes the form: 
\begin{align}
    {\mu}^*_1 = \begin{bmatrix}
        1-x_1 \\x_1\\0
    \end{bmatrix}  \quad\text{and} \quad {\nu}^*_1 = 
    \begin{bmatrix}
        0 \\[3pt]
        \frac{2}{3}\\[3pt]   
        \frac{1}{3} \end{bmatrix},
\end{align}
such that $x_1 \in [\frac{1}{3}, 
    \frac{2}{3}]$, leading to a game value of $-\frac{1}{3}$. This establishes the distribution at $t=1$ and confirms the existence of a two-time step optimal trajectory, thereby proving the first part of the proposition.

    Now note the following:
    \begin{enumerate}
        \item 
        The original objective function \eqref{eq:crps-obj} can be expressed in a bilinear form (similar to the expressions for $\mu_0, \mu_1, \mu_2$ using $x_1, x_2, x_3$). This makes it concave in the first argument and convex in the second argument.
        
        \item  
        The mean-fields $\mu$ and $\nu$ lie on a simplex and are hence, compact and convex. 
    \end{enumerate}
    Thus, by the generalized version of von Neumann's minimax theorem~\citep{sorin2002first}, we conclude that the game value is unique, proving the second part of the proposition \footnote{Note: The optimal infinite horizon trajectory itself need not be unique (we have shown that $x_1$ can take a range of values).}.
\end{proof}

\subsection{Implementation Details and Hyperparameters}

The state distributions are represented as arrays that are concatenated together to form the global observation. 
This becomes the input to the critic network, which consists of a single hidden layer of 64 neurons and two \textsf{tanh} activation functions. 
The output is a single value that is equal to the estimated value function.
On the other hand, the actor-network consists of a single MLP layer of 64 neurons that is concatenated with the local agent observation. 
Additionally, the \textsf{logits} are converted to a probability distribution through a softmax layer. 
The dimension scales with $|\mathcal{A}|$. 
Both the actor and critic networks are initialized using orthogonal initialization \citep{shengyi2022the37implementation}.

The single-stage RPS game is trained for 5,000 time steps with the actor and critic learning rates set to 0.0005 and 0.001, respectively, which remain constant throughout training. 
The networks are updated using the ADAM optimizer~\citep{Kingma2014AdamAM}
every 50 time steps for 10 epochs and a PPO clip value of 0.1.
The entropy is decayed from 0.01 to 0.001 geometrically. 
We use an episode length of 1 after which the rewards are bootstrapped. 

%
Moreover, since we have a single ``team'' buffer and the input/output dimensions are small, we do not use a mini-batch based update.
For cRPS we use an episode length of 10 after which the rewards are bootstrapped.
cRPS 
is trained using 200,000 time steps (=20,000 episodes) and is updated every 100 time steps. 
The algorithm was trained on a single NVIDIA GeForce RTX 3070 GPU and the training times are given in Tables~\ref{table:rps} and \ref{table:crps-eval}.

\section{Battlefield Setup}\label{AppendixC}
\subsection{State and Action Space}
We consider a large-scale two-team (Blue and Red) ZS-MFTG on an $n\times n$ grid world.
The state of the $i^{th}$ Blue agent is defined as the pair $x_i = (p_i^x, s_i^x)$ where $p_i^x \in \S_{\text{position}}$ denotes the position of the agent in the grid world and $s_i^x\in \S_{\text{status}} = \{0, 1\}$ defines the status of the agent: 0 being inactive and 1 being active.
Similarly, we define the state of the Red agent as $y_i = (p_i^y, s_i^y)$. 
The state spaces for the Blue and Red teams are denoted by $\X = \Y = \S_{\text{position}} \times \S_{\text{status}} $, respectively. 
The mean-fields of the Blue ($\mu$) and Red ($\nu$) teams are distributions over the joint position and status space, i.e., ${\mu}, {\nu} \in \P(\S_{\text{position}} \times \S_{\text{status}})$. 
The action spaces are given by $\U = \V =\{\texttt{Up},\texttt{Down},\texttt{Left},\texttt{Right},\texttt{Stay}\}$ for both teams, representing discrete movements in the grid world. 
The learned identical team policy assigns actions based on an agent's local position and status, as well as the observed mean-fields of both teams. 
In the following subsections, we elaborate on the weakly coupled transition dynamics and reward structure introduced in the game, followed by a detailed discussion of the training procedure and network architecture for MF-MAPPO in this example.

\subsection{Interaction Between Agents}

The transitions between states for agents belonging to both teams are characterized by their dynamics. 
These dynamics are probabilistic and depend on interactions among agents and are weakly coupled through their mean-field distributions.
The weak coupling dynamics is keeping in line with the assumption in \cite{guan2024zero}.

An agent at a given grid cell can be deactivated by the opponent team with a nonzero probability if the empirical mean-field of the opponent team at the grid cell supersedes that of the agent's own team. 
Similarly, a deactivated agent can be revived if the empirical mean-field of the agent's team is greater than the opponent's. 
This is referred to as numerical advantage.
The total transition probability from state $(p, s)$ to state $(p', s')$ by taking an action $a$ is given by 
\begin{align*}
     \mathbb{P}\big((p', s') \,|\, (p, s), a, \mu, \nu \big)
    = \mathbb{P}\big(p'\,|\,(p', s'), a \big)  \, 
    \mathbb{P}\big(s' \,|\, (p, s), \mu, \nu \big),
\end{align*}
where the first term on the right-hand side corresponds to the deterministic position transition when the agent is active. 
The second term corresponding to the status transition is given by
\begin{align*}
\mathbb{P} \big( 0 \mid (p,1), \mu, \nu \big) &= \operatorname{clip}_{[0, 1]} \big( \alpha_x (\nu(p) - \mu(p)) \big), \\
\mathbb{P} \big( 1 \mid (p,1), \mu, \nu \big) &= 1 -  \mathbb{P} \big( 0 \mid (p,1), \mu, \nu \big),
\end{align*}
and
\begin{align*}
    \mathbb{P}\big(1 \,|\, (p, 0), \mu, \nu \big) &= \operatorname{clip}_{[0, 1]}\big(\beta_x (\mu(p) - \nu(p)) \big), \\
    \mathbb{P}\big(0 \,|\, (p, 0), \mu, \nu \big) &= 1-\mathbb{P}\big(1 \,|\, (p, 0), \mu, \nu \big),
\end{align*}
where  $\nu(p)-\mu(p)$ is the Red team's numerical advantage over the Blue team at $p$.
Similarly, the Blue team's numerical advantage over Red is given by $\mu(p)-\nu(p)$.
$\alpha_x$ and $\beta_x$ are tuning parameters to control the Red team's deactivation power and Blue team's reactivation power respectively. 
The Red team, being the defending team, is given a slight advantage in terms of higher deactivation power. 
This enables the possibility of capturing Blue team agents. 
However, to avoid degeneracy, the Red team agents are not allowed to enter the target.
For our experiments, we assume $\alpha_x = 15$, $\alpha_y = 5$, and $\beta_x = \beta_y = 0.$

\subsection{Reward Structure}

The team rewards only depend on the mean-fields of the two teams.
For the battlefield scenario, the Blue team agents receive a positive reward corresponding to the fraction of agents that reach the target alive.
%
This is a one-time reward that depends on the change in the fraction of the population of the agents at the target, i.e., if ${\mu}_{t}|_{\text{target}} = {\mu}_{t+1}|_{\text{target}}$, then the team does not receive any positive reward. 
Each agent in the team receives an identical ``team reward.'' 
The reward function is mathematically formulated as
\begin{align*}
 R_{\text{Blue},{t+1}}(\mu, \nu) = \kappa \, \Delta\mu_{t+1}|_{\text{target}},
\end{align*}
where,
\begin{align*}
    \Delta\mu_{t+1}|_{\text{target}} &=     \mu_{t+1}(p^x = \text{Target}, s^x = 1) - \mu_{t}(p^x = \text{Target}, s^x = 1).
\end{align*}

We have chosen
$\kappa = 100$ in our simulations (heavier emphasis on reaching the target). 
The Red team's reward is the negative of the Blue team since we have a zero-sum game.
Each team aims to maximize its own expected reward.

\subsection{Implementation and Hyperparameters}

The state distribution for a grid world of size $n\times n$ is represented as a 
three-dimensional array of size $(2, n, n)$ for each team. 
The first layer depicts the mean-field of the agents over an $n\times n$ grid that are alive and active, while the second layer gives information about the team's deactivated population. 
Each team's distribution is then concatenated together to form the global observation.
This becomes the common information that is the input to the critic network which in our case is of size $(4, n, n)$ as we have two teams. 
Both neural networks consist of two main parts: a convolutional block and a fully connected block. 

For the critic, the first CNN layer is the input layer that takes the 4 channels and outputs 32 channels, with a kernel size of 3x3, stride of 1, and padding of 1. 
Followed by \textsf{ReLU} activation, we have a hidden layer that takes 32 channels and outputs 64 channels, with the same kernel size, stride, and padding. 
Lastly, after another \textsf{ReLU} activation, we have the output layer that takes 64 channels and outputs 64 channels, again with the same kernel size, stride, and padding. 
After another \textsf{ReLU} layer, the output of the CNN is passed through an MLP. 
Namely, a fully connected (dense) layer takes the flattened output of the convolutional block and reduces it to 128 units. 
Between the input and the output layers, we have a single \textsf{tanh} activation function.

On the other hand, the input to the actor-network is split into two CNN blocks: one to process the common information and one to process the local information. 
The local information channel is an array of size $(1, n, n)$  that locates the position of the agent with value +1 if it is active and -1 if it has been deactivated.
This local information is passed through a single CNN layer that outputs 16 channels with a kernel size of 3x3, stride of 1, and padding of 1 while the common information is passed through two such layers with the output of 32 channels.
Both outputs are then followed by a \textsf{ReLU} activation function and the latent representation of the common information combined with the local agent observation is then passed through an MLP architecture.

A fully connected (dense) layer takes the flattened output of the convolutional block and reduces it to 512 units. 
We have a single hidden layer that reduces the dimension further to 128 and then the output \textsf{logits}. 
The layers are separated by the \textsf{tanh} activation functions. 
Finally, the \textsf{logits} are converted to a probability distribution through a 
\textsf{softmax} layer.
Both the actor and critic networks are initialized using orthogonal initialization \citep{shengyi2022the37implementation}.
The architectures of the shared-team actor and minimally-informed critic networks for this example are shown in Figures \ref{fig:actor-network-battlefield} and \ref{fig:critic-network-battlefield} respectively.

 \begin{figure}[ht]
    \centering 
    \includegraphics[width=0.8\linewidth]{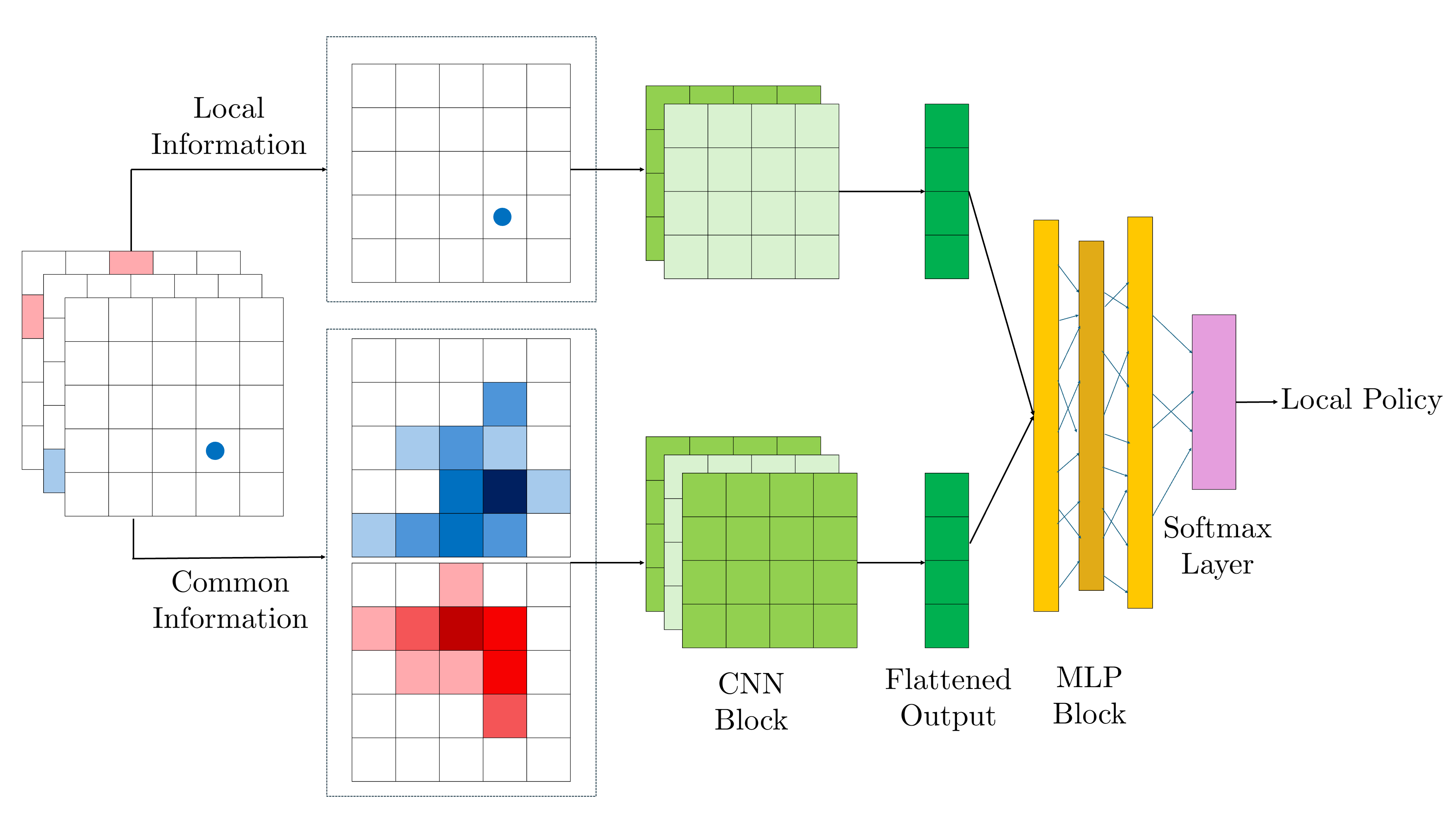}
    \caption{MF-MAPPO: Shared-team actor for battlefield}
    \label{fig:actor-network-battlefield}
\end{figure}

\begin{figure}[ht]
    \centering
    \includegraphics[width=0.8\linewidth]{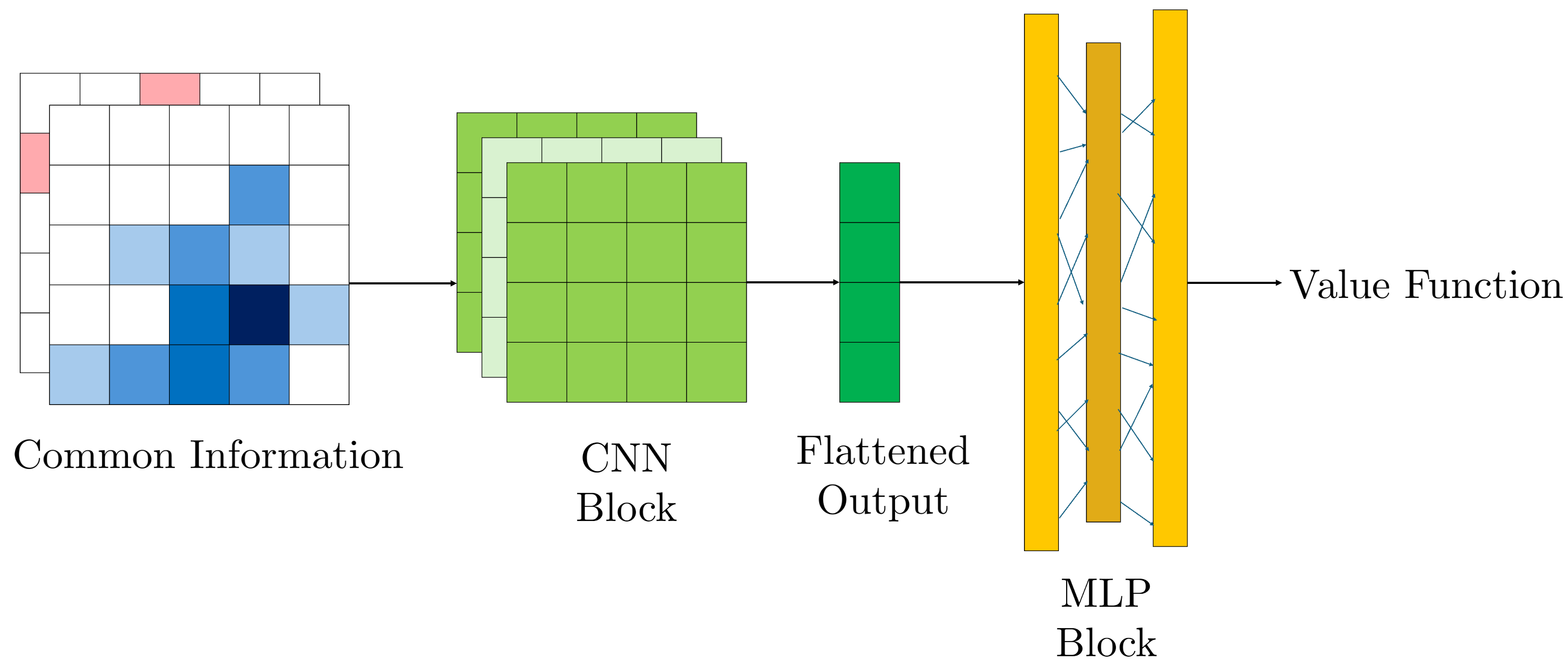}
    \caption{MF-MAPPO: Minimally-informed critic for battlefield}
    \label{fig:critic-network-battlefield}
\end{figure}

 All maps are trained using a single NVIDIA GeForce RTX 3070 GPU. 
 The actor and critic learning rates are set to 0.0005 and 0.001 and both decay geometrically by a factor of 0.999. 
 The networks are updated using the ADAM optimizer~\citep{Kingma2014AdamAM} with two mini-batches for 10 epochs and a PPO clip value of 0.1. 
 The entropy coefficient is initialized to 0.01 and decays with a factor of 0.995. 

 Maps 1 and 2 which are $4\times4$ grid worlds are trained for $5\times 10^6$ and $4.5\times 10^6$ time steps, respectively, and in both cases, the episode length is 20 time steps and the update frequency is every 500 time steps. 
 The total training period is about one day. 
 On the other hand, Map 3 being $8\times8$ in dimension, has an episode length of 64, is trained for $9\times 10^6$ time steps and its network is updated every 1,000 time steps. 
 The total training period is approximately three days.

For evaluating D-PC, we employ MF-MAPPO to learn Lipschitz policies under a fully observable mean-field regime, and use these policies to evaluate the performance of our dynamic estimation algorithm, D-PC.
We consider a time-invariant, edgeless visibility graph, wherein agents have access only to the mean-field distribution of their opponent corresponding to their current cell.
For the communication graph, we assume a subgrid based topology.
In the subgrid communication topology, an agent located at state \(x\) can exchange information with its neighboring states arranged in a \(3\times 3\) grid centered around \(x\).
The communication graph is naturally time-varying since the number of states occupied by the agents need not be fixed.
In accordance with Assumption~\ref{assumpt:graph-connectedness} and the theoretical requirements of both D-PC and benchmark algorithms, we explicitly impose the constraint that the communication graph is connected at all time steps.

At each time step $t$, the initial Red team's estimate for the D-PC algorithm by the Blue team (i.e., at $\tau = 0$) for every state \(x\) is taken as the projection of the estimate from the previous time step \(t-1\) onto the constraint set \(\R(x)\), that is,
\[
\hat\nu^{N_2, \tau=0}_{x, t} = \Omega_{\R(x)}\left[\hat\nu^{N_2, R_{\mathrm{com}}}_{x, t-1}\right],  \quad x \in \mathcal{X}^o_t.
\]
At $t=0$ and $\tau=0$, we assume $\hat\nu^{N_2, \tau=0}_{x, t=0} \sim \textrm{unif}\left(\R(x)\right) $ for all $ x \in \mathcal{X}^o_t.$

\section{Epidemiology Setup}\label{Appendix-epidemic}
\subsection{State and Action Space}
We consider a large-scale two-team (Blue and Red) ZS-MFTG on an $n\times n$ grid world where the Blue team represents a population trying to prevent itself from a spreading virus, represented by the Red team.
The state of the $i^{th}$ Blue agent is defined as the pair $x_i = (p_i^x, s_i^x)$ where $p_i^x \in \S_{\text{position}}$ denotes the position of the agent in the grid world and $s_i^x\in \S_{\text{status}} = \{0, 1\}$ defines the status of the agent: 0 being unhealthy and 1 being healthy.
The state of the Red agent as $y_i = p_i^y$, denoting the position of the Red agent.
The state spaces for the Blue and Red teams are denoted by $\X = \S_{\text{position}} \times \S_{\text{status}} $ and $\Y = \S_{\text{position}}$, respectively. 
The mean-fields of the Blue ($\mu$) and Red ($\nu$) teams are distributions over the joint position and status space, i.e., ${\mu} \in \P(\S_{\text{position}} \times \S_{\text{status}})$ while ${\nu} \in \P(\S_{\text{position}})$. 
The action spaces are given by $\U = \V =\{\texttt{Up},\texttt{Down},\texttt{Left},\texttt{Right},\texttt{Stay}\}$ for both teams, representing discrete movements in the grid world. 
The learned identical team policy assigns actions based on an agent's local position and status, as well as the observed mean-fields of both teams. 

\subsection{Interaction Between Agents and Reward Structure}

We consider weakly coupled dynamics where the Blue team agent at a cell can be infected by a non-zero probability that is proportional to the Red team population as well as the infected Blue team population at that cell.
For the Blue agents, the total transition probability from state $(p, s)$ to state $(p', s')$ by taking an action $a$ is given by 
\begin{align*}
     \mathbb{P}\big((p', s') \,|\, (p, s), a, \mu, \nu \big)
    = \mathbb{P}\big(p'\,|\,(p', s'), a \big)  \, 
    \mathbb{P}\big(s' \,|\, (p, s), \mu, \nu \big),
\end{align*}
where the first term on the right-hand side corresponds to the deterministic position transition. 
The second term corresponding to the status transition is given by
\begin{align*}
\mathbb{P} \big( 0 \mid (p,1), \mu, \nu \big) &= \alpha (\nu(p) + \mu(p, s=0))  \\
\mathbb{P}\big(1 \,|\, (p, 0), \mu, \nu \big) &= \beta_x (\mu(p) - \nu(p)) ,
\end{align*}
where $\alpha$ is a tuning parameter to control the rate of infection.
The Blue team agents can go to the hospital to heal and become healthy with probability $\beta$.
The Red team on the other hand moves around the grid trying to infect as much of the population.
However, to avoid degeneracy, the Red team agents are not allowed to enter the hospital.
The reward depends on the proportion of healthy Blue agents, i.e., $\sum_{s=1}\mu(p, s)$, thereby incentivizing Blue agents to decongest and if infected, heal at the hospital and Red agents to surround the hospital, ready for contaminating the Blue team.
Thus, this example, unlike the battlefield, exhibits asymmetry in the dynamics between the Red and Blue teams.
For our experiments, we assume $\alpha= 15$, and $\beta=5.$

\subsection{Implementation and Hyperparameters}

The MF-MAPPO architecture resembles the Battlefield networks introduced in~\ref{AppendixC}.
However, the Red team mean-field has a size $(1, n,n)$, while Blue continues having $(2, n, n)$.
 The policy was trained using a single NVIDIA GeForce RTX 3070 GPU. 
 Due to asymmetry in the dynamics, we separately tune the hyperparameters for both teams.
 For the Blue team the actor and critic learning rates are set to 0.0001 and 0.0005 with 3 epochs of training.
 On the other hand the Red team network learns with 0.0005 for the actor and 0.001 for the critic, with updates over 10 epochs.
 The entropy coefficient is initialized to 0.01 and both the entropy and the learning rates for both teams decay geometrically throughout the training period.
 The total training period is about half a day. 
 While Table~\ref{table-4x4-eval} quantitatively evaluates MF-MAPPO, a visualization of the team behaviors is presented in~\ref{fig:virus-image}.

\begin{figure}[t]
    \centering
    \includegraphics[width=\linewidth]{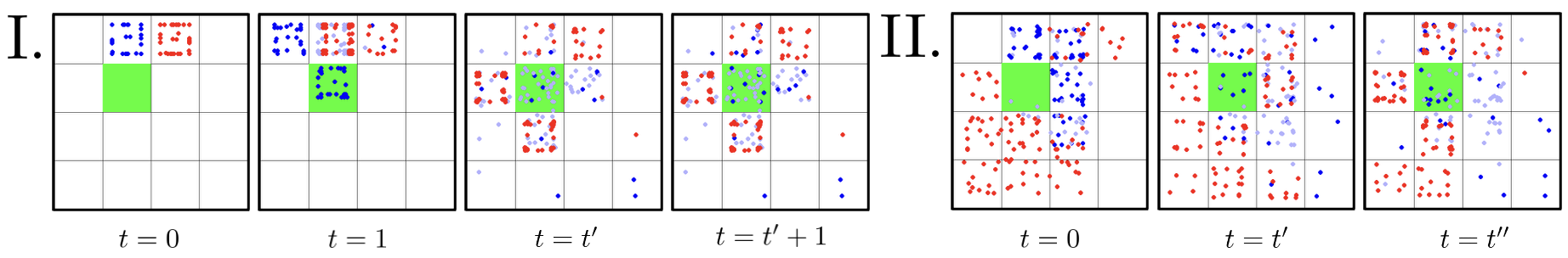}
    \caption{Epidemiology via MF-MAPPO: In both cases, Red team concentrates around the hospital making it difficult for healing Blue team agents to decongest while some Blue agents successfully maintain distance.}
    \label{fig:virus-image}
\end{figure}

\section{Additional Results}\label{appendix-additional-results}
In this section, we present additional analyses and simulation results.
In particular, we evaluate the exploitability of policies trained using MF-MAPPO and investigate the algorithm's sample efficiency in the battlefield scenario.
Furthermore, we present results for the standard Rock-Paper-Scissors game and additional experiments on the zero-sum battlefield game, both using MF-MAPPO.
%
%

\subsection{Exploitability Evaluation}\label{subsec:explt}

\edit{For cRPS, the Blue and Red policies exhibit very low exploitability ($0.002$ and $0.004$, respectively), closely matching the equilibrium performance of $-0.331$.
Together with the learned mean-field trajectories matching the analytical equilibrium in Proposition~\ref{prop:two-step-convergence}, this provides empirical evidence that MF-MAPPO converges toward the equilibrium.}

\edit{For Battlefield, the Blue policy has relatively low exploitability ($4.735$), while the Red policy is substantially more exploitable ($26.114$), compared with an equilibrium performance of $76.26$.
Although our training losses indicate convergence of the MF-MAPPO updates, this does not necessarily imply convergence to an equilibrium.
One possible explanation is that, unlike cRPS, the Battlefield game may not admit a game value and therefore may not possess an equilibrium toward which MF-MAPPO can converge.
In addition, the asymmetric structure of the environment  where Red is prohibited from entering the target region may make the Red policy more constrained and sensitive to best-response exploitation.
Approximation error, entropy regularization, and limited exploration may also contribute, and we therefore refrain from attributing the higher Red-team exploitability to any single cause, given the limited experiment we performed. 
Nevertheless, the Red policy consistently outperforms the baselines in cross-play (Table~\ref{table-4x4-eval}), providing complementary evidence that it is comparatively strong despite not being fully non-exploitable.}

\edit{Establishing broader conditions (beyond independent-dynamics) for game-value existence in ZS-MFTGs and understanding practical learning convergence beyond the settings considered here are important directions for future work.}

\subsection{Rock-Paper-Scissors (RPS)}\label{subsec:rps}
We first extend the two-player Rock-Paper-Scissors (RPS) game to a game played between two populations as described in Appendix~\ref{AppendixB}. 
The Nash equilibrium for this population-based RPS game is the uniform population distribution $[1/3, 1/3, 1/3]$ over the 3 states~\citep{raghavan1994zero}. 

We compare MF-MAPPO with DDPG-MFTG~\citep{shao2024reinforcementlearningfinitespace} based on the training time, average test rewards and attainment of the computed Nash distributions for $N_1=N_2=1,000$ agents. 
We exclude MADDPG~\citep{lowe2017multi} from our comparison, as it scales poorly to hundreds or thousands of agents {due to its reliance on all agents’ local and global observations and actions as inputs to its critic networks.}
We include the training curve for cRPS in Figure~\ref{fig:crps-training} for reference.

From the learning curves in Figures~\ref{fig:rps-training} and~\ref{fig:crps-training} one can see that the DDPG-MFTG algorithm failed to converge to the analytical game value of zero, while MF-MAPPO almost immediately attained the Nash game value.
This corroborates with the results presented for cRPS (Table~\ref{table:crps-eval}) in the main text. 
However, as shown in Table~\ref{table:rps}, MF-MAPPO does take slightly longer to train since, unlike DDPG-MFTG, {MF-MAPPO} avoids mini-batch training, following \cite{yu2022surprising}.

We tested the learned policy with a fixed initial distribution $\mu_{t=0} = [1, 0, 0]\t$ and $\nu_{t=0} = [0, 1, 0]\t$, and the resulting trajectories are  visualized in Figure~\ref{fig:rps-simplex-combined}.
All simulations were run for 150 instances.
The trajectories of the Blue and Red team ED are depicted in cyan and pink, respectively, alongside the mean trajectory. 
The randomness in these trajectories arises from the finite-population approximation under a stochastic optimal policy, resulting in stochastic EDs. 
As shown in Figures~\ref{fig:rps-simplex-combined},  DDPG-MFTG diverges from the equilibrium whereas MF-MAPPO converges immediately.

\begin{figure}[!b]
    \centering
    \includegraphics[width=0.5\linewidth]{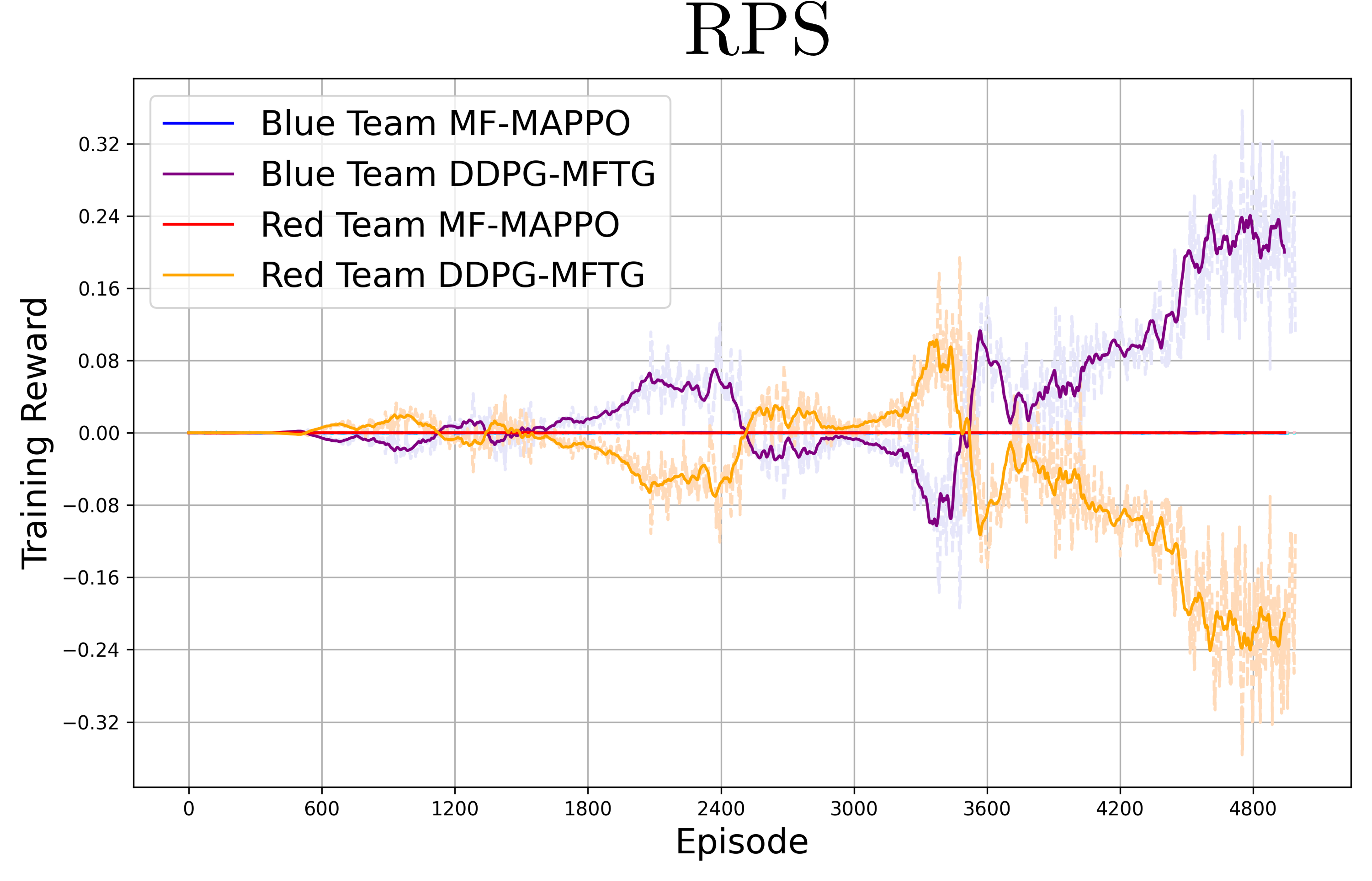}
    \caption{Training curve for RPS.}
    \label{fig:rps-training}
\end{figure}

\begin{figure}[!b]
    \centering
    \includegraphics[width=0.8\linewidth]{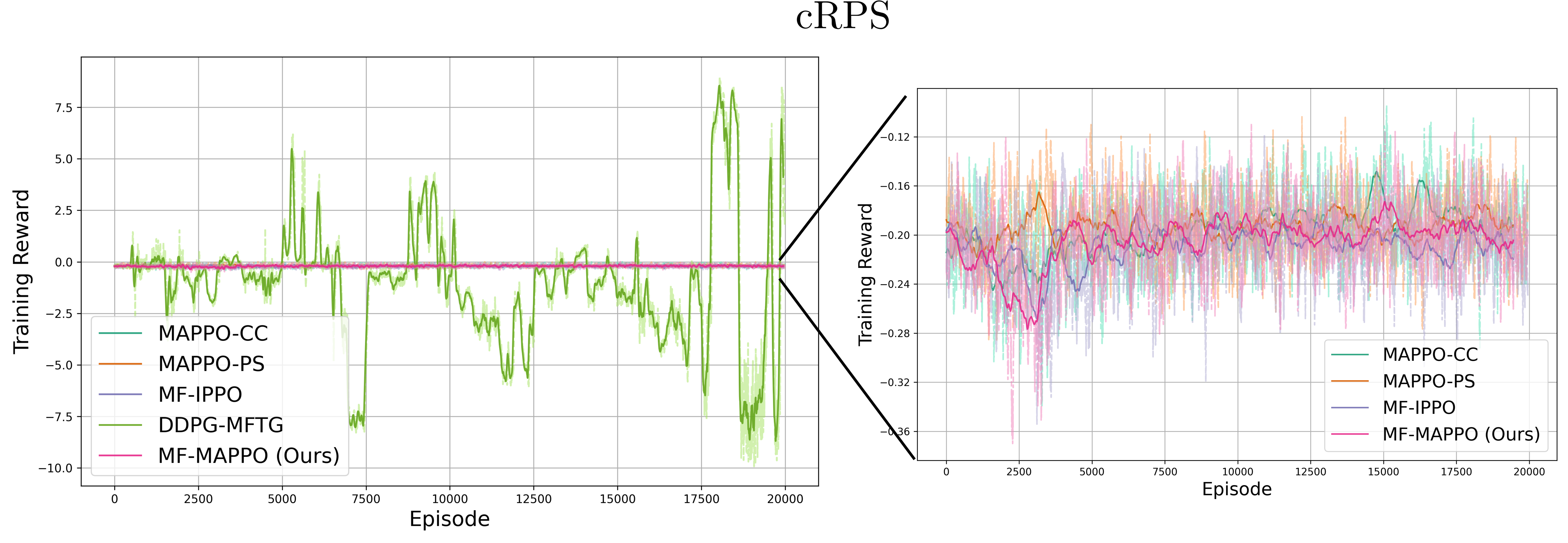}
    \caption{Training curve for cRPS.}
    \label{fig:crps-training}
\end{figure}

\begin{table}[b!]
\centering
\caption{Performance comparison for RPS}
\begin{tabular}{cccc} 
 \toprule
 Approach  & Training Time & Average Reward & NE Attained? \\ 
 \midrule
 MF-MAPPO  & 5min 17s & 0.0 & \cmark \\ 
 DDPG-MFTG  & 1min 34s & 0.334 & \xmark \\ 
 \bottomrule
\end{tabular}
\label{table:rps}
\end{table}


\begin{figure}[t!]
    \centering
    \includegraphics[width=0.6\textwidth]{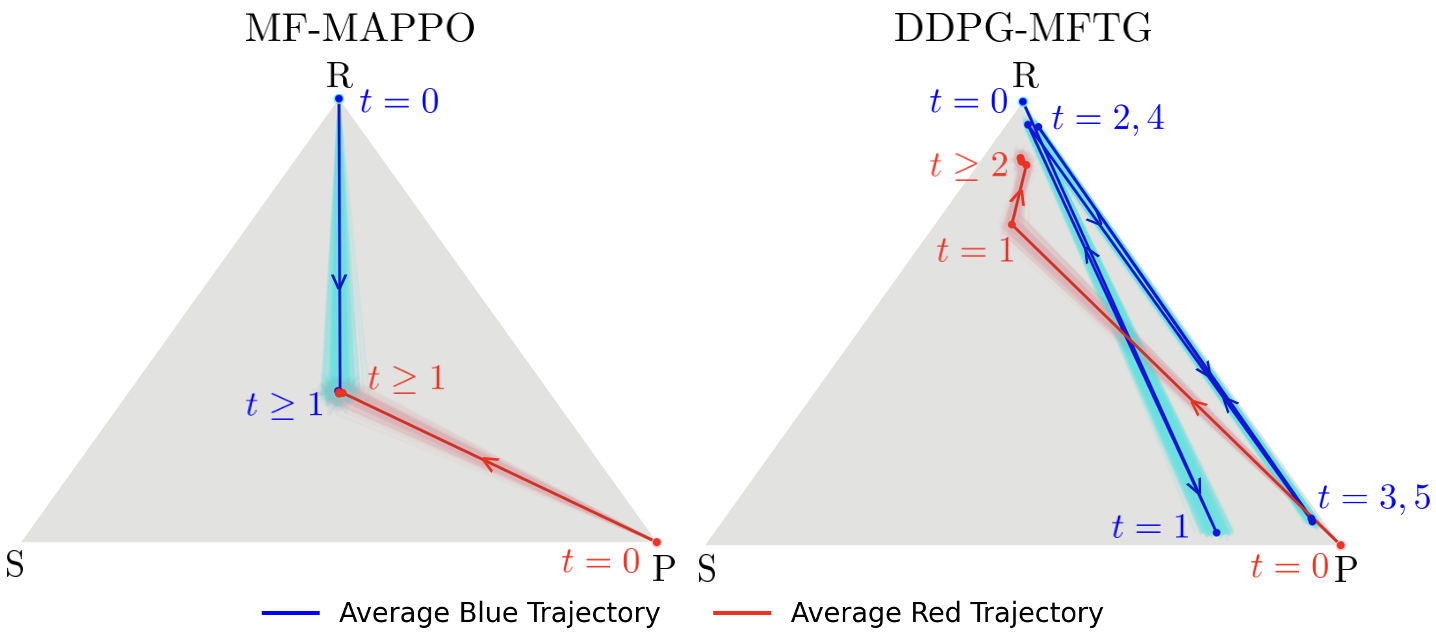}
    \caption{ED trajectories induced by learned team policies on the state distribution simplex. Mean trajectories are averaged based on the 150 runs from fixed initialization $\mu_{t=0} = [1, 0, 0]\t$ and $\nu_{t=0} = [0, 1, 0]\t$; $N_1=N_2=1,000$. }
    \label{fig:rps-simplex-combined}
\end{figure}

\subsection{Battlefield}

\subsubsection{Validation Cases for MF-MAPPO}
The following subsection qualitatively discusses the battlefield game for different map layouts.  
For these results, both teams are deploying policies trained using MF-MAPPO.

\textbf{Map A.}
The first map is a simple $4\times 4$ grid world with a single target that we use to validate our algorithm. 
The target is partially blocked by an obstacle, see Figure~\ref{fig:map1IC3}.
For the initial condition in Figure~\ref{fig:map1IC3}, the Blue team is initially split into two equal groups. 
The Blue team decides to merge the two sub-groups of agents into a single group.
With this formation, the Red team has zero numerical advantage over the Blue team when they encounter in (g), resulting in all Blue agents safely arriving at the target. 
In comparison, if the two Blue subgroups do not merge but move toward the target one at a time, it will lead to 50\% of the Blue team population being deactivated (first subgroup), followed by the remaining 50\% (second subgroup).
This demonstrates how the observation of mean-field distributions guides rational decision-making.

\begin{figure}[ht]
    \centering
    \includegraphics[width=0.4\linewidth]{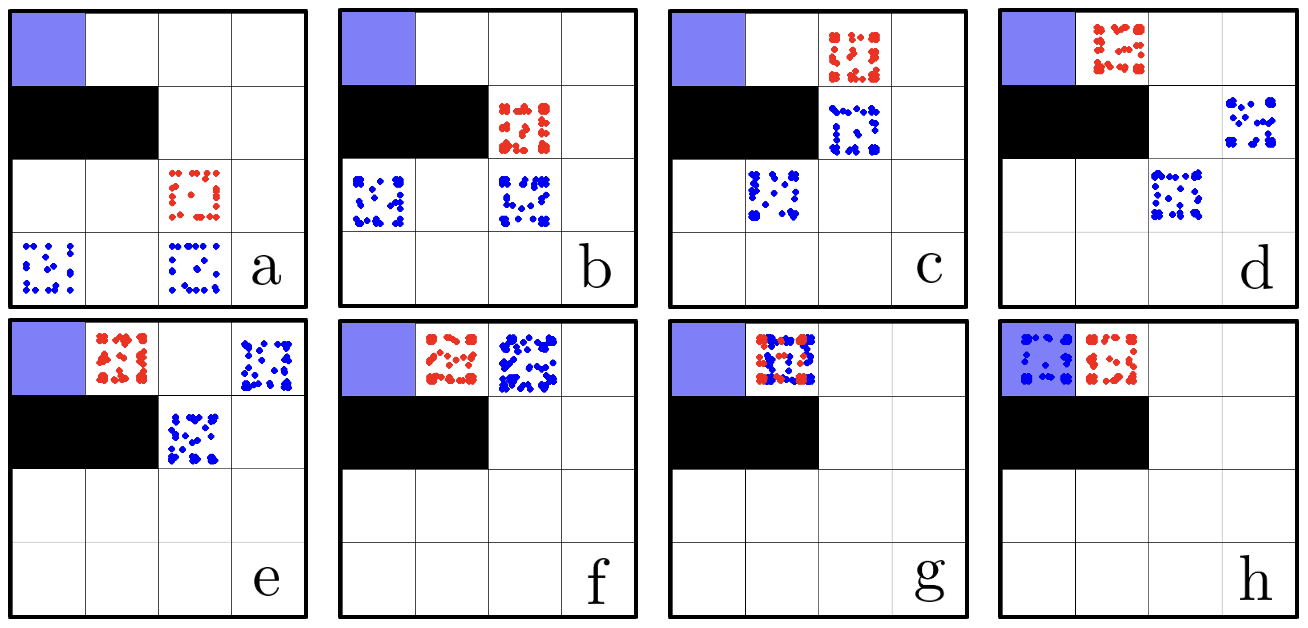}
    \caption{Red is concentrated; Blue is evenly split.}
    \label{fig:map1IC3}
\end{figure}

\textbf{Map B.}
This map is identical to the one presented in Section \ref{sec:numerical-exp}. 
In the first scenario (Figure \ref{fig:map2IC2}), 70\% of the Blue agents start at cell [2, 2], and 30\% at cell [1, 1], while the Red team is evenly split between cells [0, 1] and [3, 3]. 
Half of the Red team at [0, 1] successfully blocks the 30\% Blue agent group from entering the left corridor due to its numerical advantage, which forces the Blue agents to opt for the right corridor. 
At the same time, the larger Blue group with 70\% of the population utilized their numerical advantage over the half Red team at the top right and deactivated all the Red agents as shown in (c) and reached the target at time step (d). 
This allowed the smaller 30\% group to follow through the same corridor without losing agents. 
\begin{figure}[htb]
    \centering
    \includegraphics[width=0.6\linewidth]{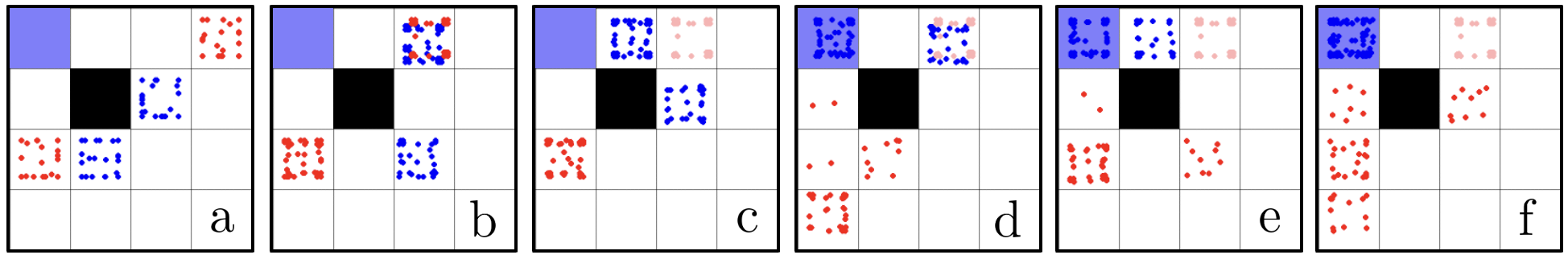}
    \caption{Red is evenly split; 30\% Blue are at $[1, 1]$ and 70\% are at $[2, 2]$.}
    \label{fig:map2IC2}
\end{figure}
In the second scenario (Figure~\ref{fig:map2IC3}), with the Red team evenly distributed at the corners, 30\% of the Blue agents start at cell [2, 2] and 70\% at cell [3, 1]. 
The Red team's numerical advantage at [3, 3] forces the Blue agents to move around and regroup (Figures \ref{fig:map2IC3}(b)-\ref{fig:map2IC3}(f)). 
Once united, the Blue team’s numerical advantage forces the Red subgroup at [0,1] to disperse to avoid deactivation, allowing Blue to reach the target. 

\begin{figure}[ht]
    \centering
    \includegraphics[width=0.4\linewidth]{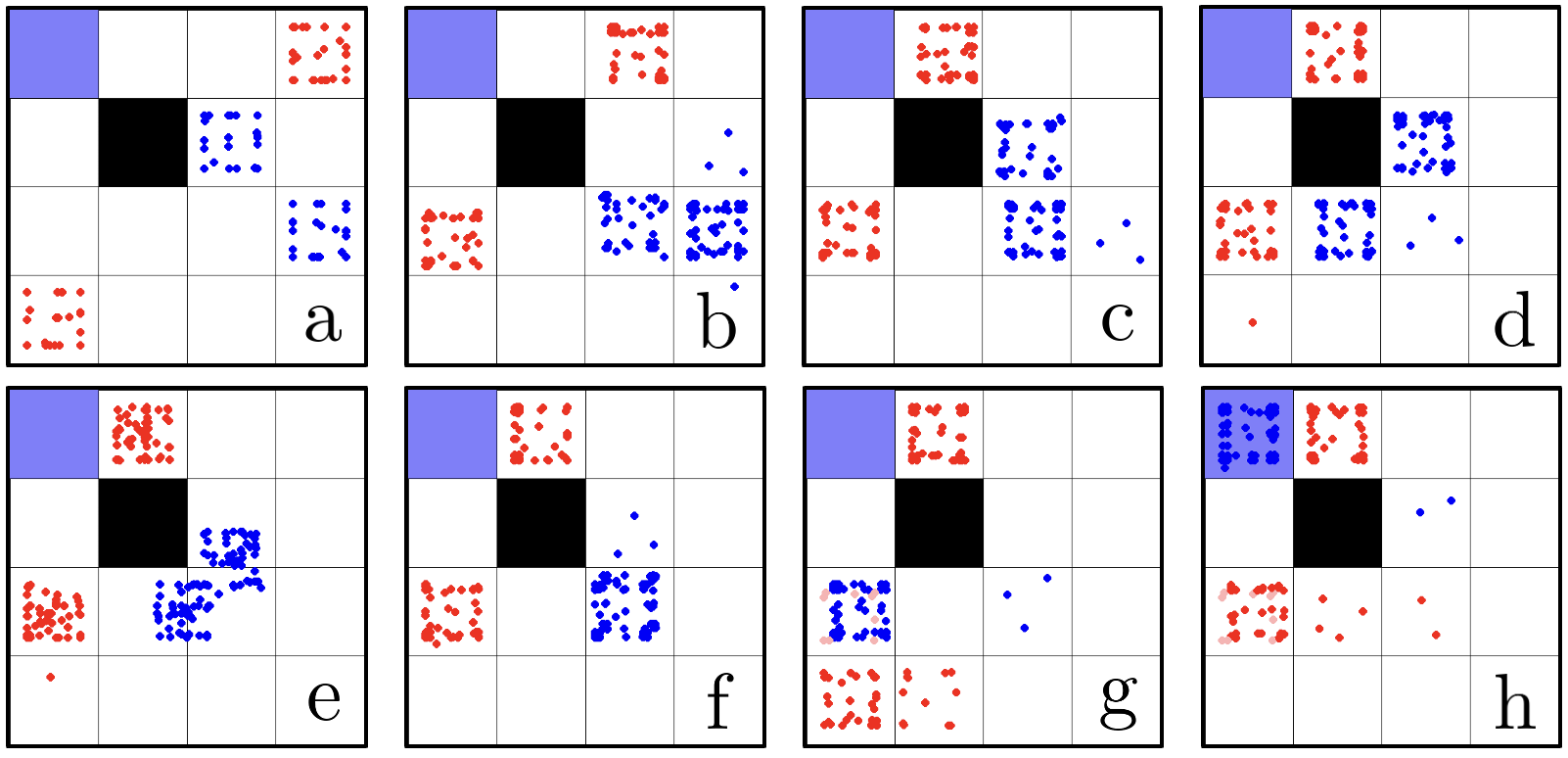}
    \caption{Red is evenly split; 70\% Blue are at $[3, 1]$ and 30\% are at $[2, 2]$.}
    \label{fig:map2IC3}
\end{figure}

\textbf{Map C.}
In Figure~\ref{fig:extend-to-many-battle}, we deploy the MF-MAPPO policies trained with $N_1=N_2=100$ (from the main text) to larger populations of $\barN_1 = \barN_2 = 1500$ and $\barN_1 = \barN_2 = 1000$ using the same initial condition as Figure~\ref{fig:battle-comparison}I and II respectively. The teams continue to achieve their objectives as established in Theorem~\ref{thm:extend-to-many}, exhibiting behaviors similar to $N_1=N_2=100$ setting.
Section~\ref{sec:numerical-exp} presented scenarios featuring structured initial configurations for both teams over an $8\times 8$ grid.
It is important to emphasize, however, that the algorithm is trained on a diverse set of initial conditions for a given map, ranging from agents concentrated within a few selected cells to agents distributed randomly across the grid world.
The following examples demonstrate that the teams are able to accomplish their objectives even in scenarios where agents are dispersed across the environment rather than clustered into just 1-2 subgroups using the same policy from Section~\ref{sec:numerical-exp}.
In Figure~\ref{fig:randomblue}, the Blue team is initialized randomly, and local subgroups of agents emerge and coordinate to reach the target.
This behavior is particularly pronounced near the upper target, where a greater numerical advantage facilitates successful coalition formation (Figures \ref{fig:randomblue}(c)–(e)).
\begin{figure}[htb]
    \centering
    \includegraphics[width=\linewidth]{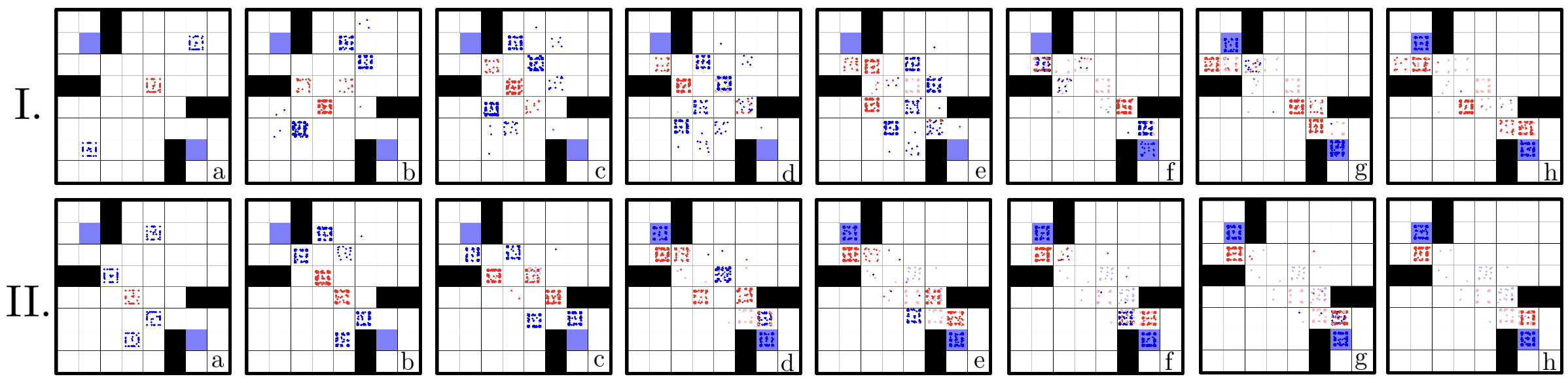}
    \caption{I. 1500 agents in each team where Red is concentrated; 30\%  Blue are at the bottom, rest are at the top II. 1000 agents in each team with Blue evenly split; Red is concentrated.}
    \label{fig:extend-to-many-battle}
\end{figure}

\begin{figure}[htb]
    \centering
    \includegraphics[width=0.75\linewidth]{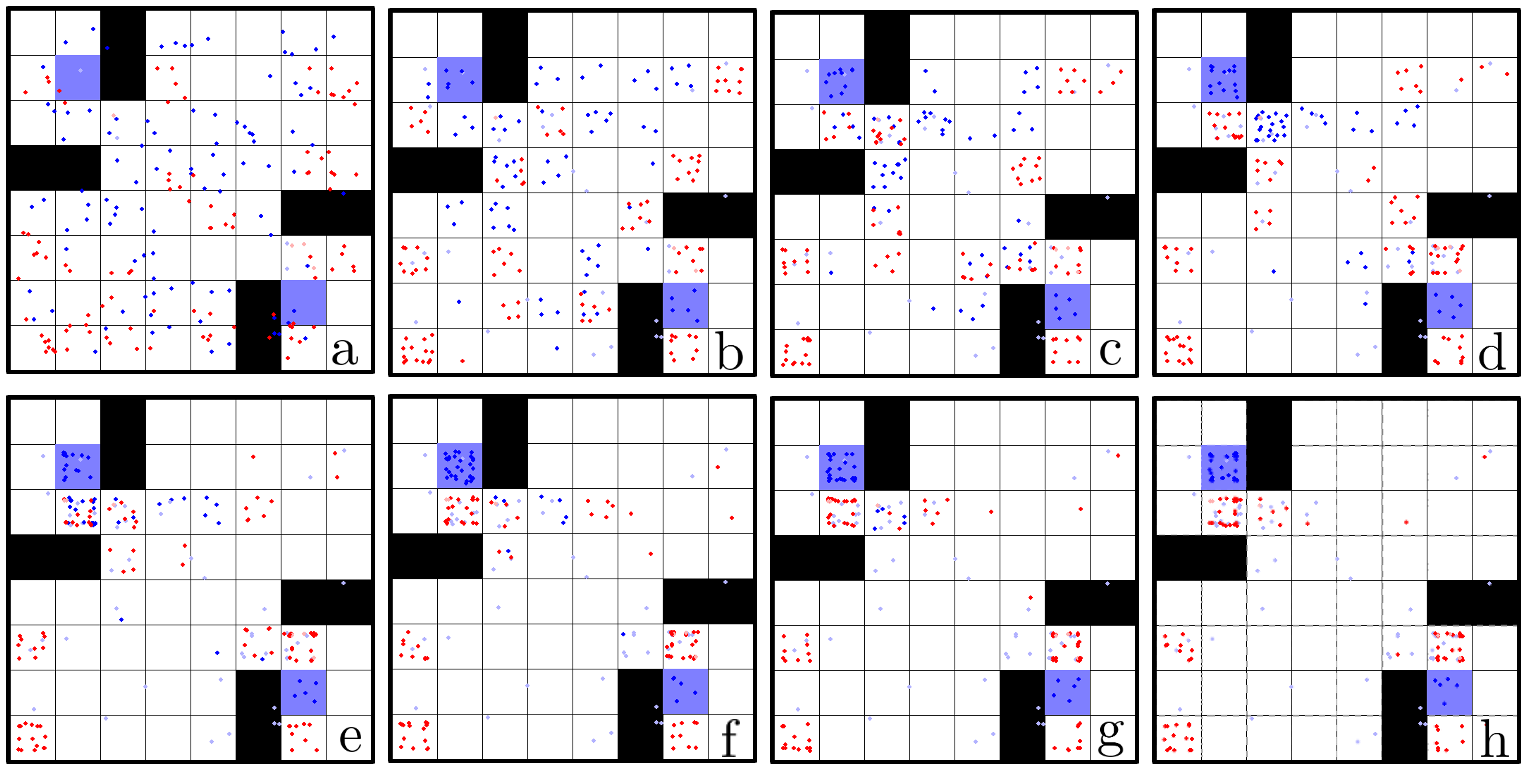}
    \caption{Blue team is randomly spread around the map.}
    \label{fig:randomblue}
\end{figure}
Turning to the randomly distributed Red team agents in 
Figure~\ref{fig:randomred}, it is observed that they concentrate near the two target entrances and successfully neutralize most Blue team subgroups.

\begin{figure}[ht]
    \centering
    \includegraphics[width=0.75\linewidth]{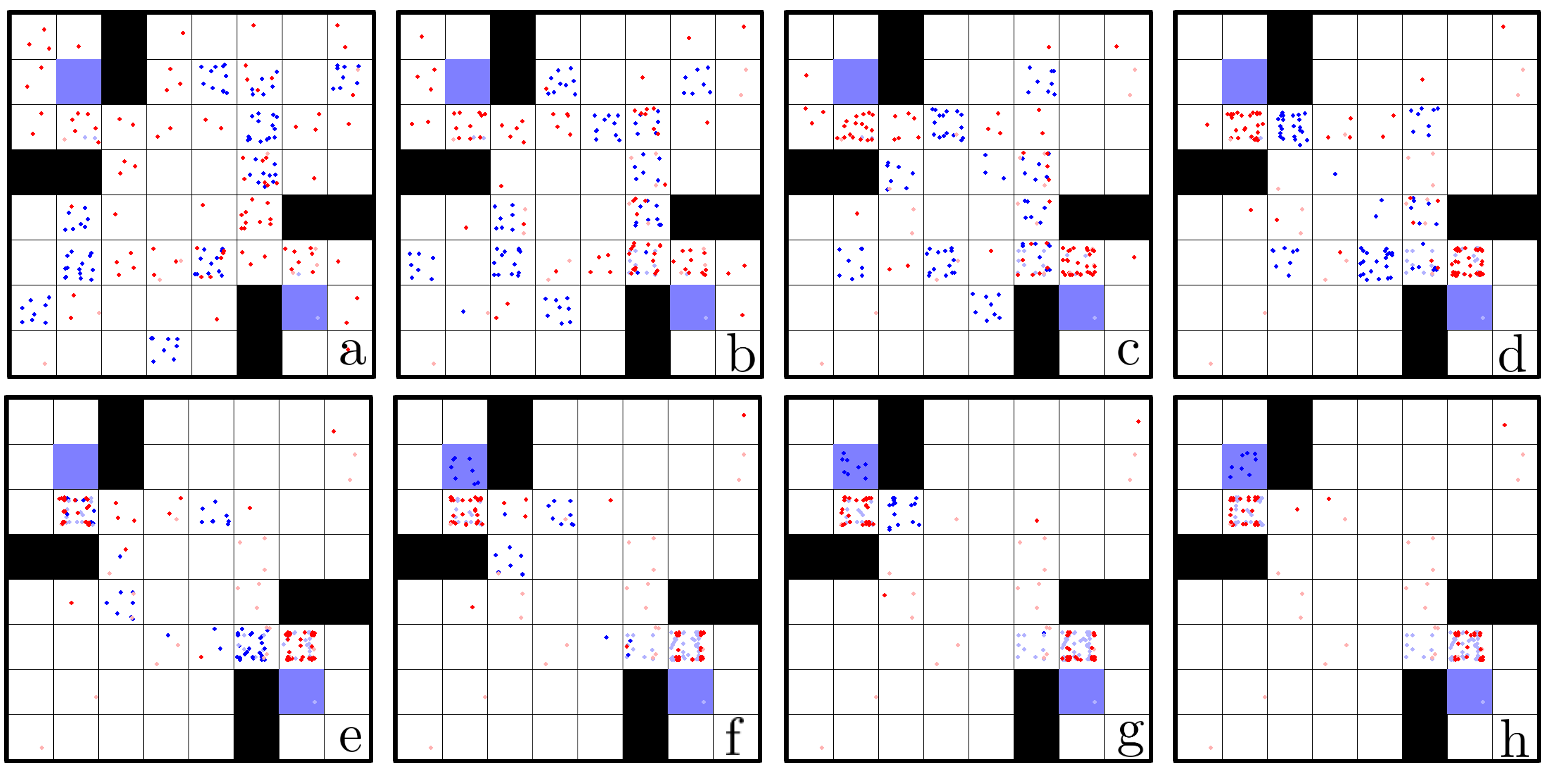}
    \caption{Red team is randomly spread around the map.}
    \label{fig:randomred}
\end{figure}

\subsubsection{Efficiency of MF-MAPPO with Baselines}

While Table~\ref{table-4x4-eval} reports the performance of the learned policies when evaluated against each other, Table~\ref{table:battlefield-new} provides a complementary analysis of the sample efficiency of the baselines in the Battlefield environment.
We define sample efficiency as the reward obtained per interaction with the environment:
\begin{align*}
    \textrm{Sample Efficiency} = \frac{\textrm{Total rewards}}{\textrm{Total environment steps}}.
\end{align*}
In addition, we report the average episode length, which serves as an indicator of how quickly episodes terminate.
Shorter episodes suggest that the Blue agents increasingly learn to reach the target, while the Red agents simultaneously learn to defend the target and deactivate Blue agents—either outcome leading to episode termination.

\begin{table}[ht!]
\centering
\caption{Performance metrics for Battlefield}
\begin{tabular}{cccc} 
 \toprule
 Algorithm  & Critic  Input & Sample Efficiency & Avg. Episode Length \\ 
 \midrule

 MF-IPPO  & $\{(x^{N_1}_{i,t}, \mu_t)\}_{i=1}^{N_1}$ & 6.3392 & 9.258  \\ 
 MAPPO-PS  & $\{(x^{N_1}_{i,t}, \mu_t, \nu_t)\}_{i=1}^{N_1}$ & 5.4581 & 12.084 \\ 
 MAPPO-CC  & $(\x^{N_1}_t, \y^{N_2}_t, \mu_t, \nu_t)$ & {5.7878} & {9.8808} \\ 
 DDPG-MFTG  & $(\phi_t, \mu_t, \nu_t)$ & 0.0548 & 20.0 \\ 
  MF-MAPPO  & $(\mu_t, \nu_t)$ & \textbf{14.6238} & \textbf{4.8448}  \\ 
 \bottomrule
 
\end{tabular}
\label{table:battlefield-new}
\end{table}

\clearpage
\newpage
\input{checklist.tex}

\end{document}

%% file: checklist.tex
\section*{NeurIPS Paper Checklist}

\begin{enumerate}

\item {\bf Claims}
    \item[] Question: Do the main claims made in the abstract and introduction accurately reflect the paper's contributions and scope?
    \item[] Answer: \answerYes{} 
    \item[] Justification: The abstract and/or introduction (under "Our contributions", Section~\ref{sec:intro}) explicitly state the claims made.

\item {\bf Limitations}
    \item[] Question: Does the paper discuss the limitations of the work performed by the authors?
    \item[] Answer: \answerYes{} 
    \item[] Justification: We discuss the limitations of this work in the conclusion (Section~\ref{sec:conclusion}) and explicitly list key assumptions in the main text (Section~\ref{sec: problem-formulation}).

\item {\bf Theory assumptions and proofs}
    \item[] Question: For each theoretical result, does the paper provide the full set of assumptions and a complete (and correct) proof?
    \item[] Answer: \answerYes{} 
    \item[] Justification: All theoretical results have been rigorously proven in the supplementary material and their implications and intuition have been provided in the main paper through experiments and general discussion (Sections~\ref{sec:mf-mappo} and~\ref{sec:numerical-exp}).

    \item {\bf Experimental result reproducibility}
    \item[] Question: Does the paper fully disclose all the information needed to reproduce the main experimental results of the paper to the extent that it affects the main claims and/or conclusions of the paper (regardless of whether the code and data are provided or not)?
    \item[] Answer: \answerYes{}
    \item[] Justification: All hyperparameters, neural network architectures and other information have been listed in the supplementary material, along with pseudocodes to the proposed algorithms.

\item {\bf Open access to data and code}
    \item[] Question: Does the paper provide open access to the data and code, with sufficient instructions to faithfully reproduce the main experimental results, as described in supplemental material?
    \item[] Answer: \answerYes{} 
    \item[] Justification: We provide a github link to the repository containing the novel benchmark environments (\texttt{MFEnv}). Further, we provide detailed descriptions of the experimental setup, hyperparameters, and training procedures in the main text and supplementary material to facilitate reproducibility.

\item {\bf Experimental setting/details}
    \item[] Question: Does the paper specify all the training and test details (e.g., data splits, hyperparameters, how they were chosen, type of optimizer) necessary to understand the results?
    \item[] Answer:\answerYes{} 
    \item[] Justification: We present the settings in the main text (Section~\ref{sec:numerical-exp}) while the auxiliary details are deferred to the supplemental material.

\item {\bf Experiment statistical significance}
    \item[] Question: Does the paper report error bars suitably and correctly defined or other appropriate information about the statistical significance of the experiments?
    \item[] Answer: \answerYes{} 
    \item[] Justification: We present our results after running the proposed algorithm over multiple seeds and initial conditions to emphasize the statistical significance and generalizability of our results. The details are mentioned in Section~\ref{sec:numerical-exp}.

\item {\bf Experiments compute resources}
    \item[] Question: For each experiment, does the paper provide sufficient information on the computer resources (type of compute workers, memory, time of execution) needed to reproduce the experiments?
    \item[] Answer: \answerYes{} 
    \item[] Justification: We discuss the details of the resources and computational complexity in the main text for the toy example (cRPS) in Section~\ref{sec:numerical-exp} and further elaborate it for the large scale battlefields and epidemiology examples in the supplemental material.
    
\item {\bf Code of ethics}
    \item[] Question: Does the research conducted in the paper conform, in every respect, with the NeurIPS Code of Ethics \url{https://neurips.cc/public/EthicsGuidelines}?
    \item[] Answer: \answerYes{} 
    \item[] Justification: The authors have reviewed the NeurIPS Code of Ethics.

\item {\bf Broader impacts}
    \item[] Question: Does the paper discuss both potential positive societal impacts and negative societal impacts of the work performed?
    \item[] Answer: \answerNo{} 
    \item[] Justification: This paper presents work whose goal is to advance the field of game theory and reinforcement learning through foundational research.
    
\item {\bf Safeguards}
    \item[] Question: Does the paper describe safeguards that have been put in place for responsible release of data or models that have a high risk for misuse (e.g., pre-trained language models, image generators, or scraped datasets)?
    \item[] Answer: \answerNA{} 
    \item[] Justification: The paper poses no such risks.

\item {\bf Licenses for existing assets}
    \item[] Question: Are the creators or original owners of assets (e.g., code, data, models), used in the paper, properly credited and are the license and terms of use explicitly mentioned and properly respected?
    \item[] Answer: \answerYes{} 
    \item[] Justification: We have provided citations wherever necessary for the code packages.

\item {\bf New assets}
    \item[] Question: Are new assets introduced in the paper well documented and is the documentation provided alongside the assets?
    \item[] Answer: \answerYes{}{}. 
    \item[] Justification: We include scripts and a well documented repository for the newly introduced \texttt{MFEnv} benchmark.

\item {\bf Crowdsourcing and research with human subjects}
    \item[] Question: For crowdsourcing experiments and research with human subjects, does the paper include the full text of instructions given to participants and screenshots, if applicable, as well as details about compensation (if any)? 
    \item[] Answer: \answerNA{} 
    \item[] Justification: The paper does not involve crowdsourcing nor research with human subjects.

\item {\bf Institutional review board (IRB) approvals or equivalent for research with human subjects}
    \item[] Question: Does the paper describe potential risks incurred by study participants, whether such risks were disclosed to the subjects, and whether Institutional Review Board (IRB) approvals (or an equivalent approval/review based on the requirements of your country or institution) were obtained?
    \item[] Answer: \answerNA{} 
    \item[] Justification: The paper does not involve crowdsourcing nor research with human subjects.

\item {\bf Declaration of LLM usage}
    \item[] Question: Does the paper describe the usage of LLMs if it is an important, original, or non-standard component of the core methods in this research? Note that if the LLM is used only for writing, editing, or formatting purposes and does \emph{not} impact the core methodology, scientific rigor, or originality of the research, declaration is not required.
    \item[] Answer: \answerNA{} 
    \item[] Justification: The core method development in this research does not involve LLMs as any important, original, or non-standard components. The authors acknowledge the use of GPT-4 and GPT-5 for polishing the main text and for quick access to existing results in probability theory, in particular, for Theorem~\ref{thm:convergence-of-gradients} and Proposition~\ref{thm:estimator-eval-dyn}.

\end{enumerate}